\documentclass[a4paper,onecolumn,unpublished]{quantumarticle}
\pdfoutput=1

\usepackage[textheight=8.8in,textwidth=6in,top=1in,headheight=12pt,headsep=25pt,footskip=30pt]{geometry}

\usepackage{graphicx} 
\usepackage{pdfpages}
\usepackage{marginfix}
\usepackage{multicol}

\usepackage{tikz}
\usepackage{tikzscale}  
\usepackage{preamble/tikzfig}["draft"] 

\usepackage[utf8]{inputenc} 
\usepackage[T1]{fontenc}    
\usepackage[hyphens]{url}            
\usepackage{hyperref}       

\usepackage{booktabs}       
\usepackage{amsfonts}       
\usepackage{nicefrac}       
\usepackage{microtype}      
\usepackage{multirow}       
\usepackage{float}
\usepackage{braket}
\usepackage[font=small]{caption}
\usepackage{amssymb}
\usepackage{amsthm}
\usepackage{mathtools}
\usepackage{thmtools}
\usepackage{wrapfig}
\usepackage[style=numeric, maxbibnames=99, sortcites, sorting = none]{biblatex}
\usepackage[UKenglish]{babel}
\usepackage{bookmark}   
\usepackage[doipre={doi:}]{uri}   
\usepackage{csquotes}
\usepackage{import}
\usepackage{subfiles}
\usepackage{etoolbox}
\apptocmd{\sloppy}{\hbadness 10000\relax}{}{} 
\usepackage{physics}
\usepackage{amsmath}
\usepackage{mathrsfs}   
\usepackage{bigints}    
\usepackage{stmaryrd}   
\SetSymbolFont{stmry}{bold}{U}{stmry}{m}{n}   

\usepackage{extarrows} 
\usepackage{enumerate}
\usepackage{thm-restate}
\usepackage{placeins}
\usepackage{bm}   
\graphicspath{figures}
\usepackage{subcaption}
\usepackage{pgfplots} 
\pgfplotsset{compat=newest, compat/show suggested version=false} 
\usepackage{xtab}
\usepackage{xspace}
\usepackage{tabu}

\usepackage[capitalise, noabbrev]{cleveref}

\usepackage{todonotes}
\hypersetup{
    colorlinks = true,
    linkcolor = lime!50!black,
    anchorcolor = lime!50!black,
    citecolor = red!70!black,
    filecolor = lime!50!black,
    urlcolor = lime!50!black,
    breaklinks=true
}

\definecolor{cbpink}{RGB}{214,130,211}
\definecolor{cbyellow}{RGB}{241,131,108}
\definecolor{cbblue}{RGB}{110,148,189}
\definecolor{cbred}{RGB}{162,4,162}
\definecolor{cborange}{RGB}{251,145,10}
\definecolor{cbgreen}{RGB}{5,162,162}
\definecolor{cbcyan}{RGB}{204,234,207}
\definecolor{red}{RGB}{232, 165, 165}
\definecolor{darkRed}{RGB}{223, 130, 130}
\definecolor{green}{RGB}{216, 248, 216}
\definecolor{darkGreen}{RGB}{145, 235, 145}
\definecolor{purple}{RGB}{180, 105, 255}
\definecolor{yellow}{RGB}{255, 255, 130}

\theoremstyle{definition}
\newtheorem{theorem}{Theorem}[section]
\newtheorem*{theorem*}{Theorem}

\newtheorem{proposition}[theorem]{Proposition}

\newtheorem{definition}[theorem]{Definition}
\newtheorem{remark}[theorem]{Remark}

\newtheorem{example*}[theorem]{Example*}
\newtheorem{examples*}[theorem]{Examples*}

\def\bR{\begin{color}{red}}
\def\bB{\begin{color}{blue}}
\def\bG{\begin{color}{green}}
\def\bP{\begin{color}{purple}}
        \def\e{\end{color}}

\tikzstyle{Z dot}=[inner sep=0mm, minimum size=2mm, shape=circle, draw=black, fill=green, outer sep=-0.5mm, tikzit fill=green]
\tikzstyle{Z phase dot}=[draw=black, fill=green, shape=rectangle, minimum size=4.5mm, rounded corners=1.8mm, inner sep=0.5mm, outer sep=-0.5mm, scale=0.8, tikzit shape=circle, font={\footnotesize\boldmath}, tikzit fill=green]
\tikzstyle{X dot}=[shape=circle, draw=black, fill=red, inner sep=0 mm, minimum size=2 mm, outer sep=-0.5mm, tikzit fill=red]
\tikzstyle{X phase dot}=[Z phase dot, draw=black, fill=red, tikzit fill=red]
\tikzstyle{H box}=[fill=yellow, draw=black, shape=rectangle, inner sep=0.6mm, minimum height=1.5mm, minimum width=1.5mm, tikzit fill=yellow, font={\footnotesize\boldmath}]
\tikzstyle{box}=[draw=black, shape=rectangle, fill=white, minimum size=1em, inner sep=0.2em, scale=0.85, font={\scriptsize}, outer sep=-0.5mm]
\tikzstyle{black dot}=[fill=black, draw=black, shape=circle, inner sep=1pt]
\tikzstyle{sLabel}=[font={\scriptsize}, tikzit draw=black, auto]
\tikzstyle{H X Web}=[hadamard, preaction={ultra thick, draw=zx_x_web, opacity=0.4}]
\tikzstyle{H Z Web}=[hadamard, preaction={ultra thick, draw=zx_z_web, opacity=0.4}]
\tikzstyle{H XZ Web}=[hadamard, preaction={thick, draw=zx_x_web, opacity=0.4, offset=-.67pt}, preaction={thick, draw=zx_z_web, opacity=0.4, offset=.67pt}]
\tikzstyle{H ZX Web}=[hadamard, preaction={thick, draw=zx_x_web, opacity=0.4, offset=.67pt}, preaction={thick, draw=zx_z_web, opacity=0.4, offset=-.67pt}]
\tikzstyle{fault-location}=[fill=white, draw=black, shape=circle, minimum size=2mm, inner sep=0mm, outer sep=-0.5 mm, regular polygon, regular polygon sides=8, font={\tiny}]
\tikzstyle{not}=[draw=black, circle, addcross, minimum size=2mm, outer sep=-0.5mm, inner sep=0mm]

\tikzstyle{fault-free}=[-, draw={zx_fault_free}, line width=1pt, tikzit draw=magenta]
\tikzstyle{component}=[-, style=dashed, draw={zx_component}]
\tikzstyle{dotted}=[-, style=dashed, draw={zx_dotted}]
\tikzstyle{hadamard}=[-, style=dashed, draw=blue]
\tikzstyle{X Web}=[-, preaction={ultra thick, draw=zx_x_web, opacity=0.4}, tikzit draw=red]
\tikzstyle{Z Web}=[-, preaction={ultra thick, draw=zx_z_web, opacity=0.4}, tikzit draw=green]
\tikzstyle{XZ Web}=[-, preaction={thick, draw=zx_x_web, opacity=0.4, offset=-.6pt}, preaction={thick, draw=zx_z_web, opacity=0.4, offset=.6pt}]
\tikzstyle{ZX Web}=[-, preaction={thick, draw=zx_x_web, opacity=0.4, offset=.6pt}, preaction={thick, draw=zx_z_web, opacity=0.4, offset=-.6pt}]
\tikzstyle{X to Z Web}=[-, decoration={show path construction,
    lineto code={
      \draw [H X Web] (\tikzinputsegmentfirst) --
         ($(\tikzinputsegmentfirst)!0.5!(\tikzinputsegmentlast)$);
      \draw [H Z Web] ($(\tikzinputsegmentfirst)!0.5!(\tikzinputsegmentlast)$)
        -- (\tikzinputsegmentlast);
    },
  }, decorate]
\tikzstyle{XZ to XZ Web}=[-, decoration={show path construction,
    lineto code={
      \draw [H XZ Web] (\tikzinputsegmentfirst) --
         ($(\tikzinputsegmentfirst)!0.5!(\tikzinputsegmentlast)$);
      \draw [H ZX Web] ($(\tikzinputsegmentfirst)!0.5!(\tikzinputsegmentlast)$)
        -- (\tikzinputsegmentlast);
    },
  }, decorate]
\tikzstyle{braceedge}=[-, decorate, decoration={brace, amplitude=2mm, raise=-1mm}]
\tikzstyle{arrow}=[->]
\tikzstyle{logical}=[-, draw=blue]
\tikzstyle{ebit}=[-, style={double=blue, dashed}, tikzit draw=blue]

\input{preamble/floquet.tikzdefs}

\numberwithin{equation}{section}

\AtBeginDocument{%
  \DeclareFieldFormat{doi}{%
    doi\addcolon\space
    \ifhyperref
      {\href{https://doi.org/#1}{\nolinkurl{#1}}}
      {\nolinkurl{#1}}%
  }%
  \setcounter{biburllcpenalty}{100}
  \setcounter{biburlucpenalty}{100}
  \setcounter{biburlnumpenalty}{100}
}

\definecolor{darkG}{rgb}{0.,0.5,0.}

\newcommand\extrafootertext[1]{%
    \bgroup
    \renewcommand\thefootnote{\fnsymbol{footnote}}%
    \renewcommand\thempfootnote{\fnsymbol{mpfootnote}}%
    \footnotetext[0]{#1}%
    \egroup
}

\newcommand{\vone}{sequential repeat }

\newcommand{\vtwo}{inner parallel repeat }

\newcommand{\ed}{QPC distillation }

\newcommand{\dout}{d_\text{outer}}
\newcommand{\din}{d_\text{inner}}

\definecolor{purp}{rgb}{138,0.,196}
\definecolor{darkR}{rgb}{0.5,0.,0.}
\definecolor{darkB}{rgb}{0,0.,0.5}

\usepackage[normalem]{ulem}

\title{Resource-adaptive distributed fault tolerance with very noisy Bell pairs}

\author[1]{Moritz Schmidt}
\email{moritz.schmidt@uni-bremen.de}
\author[2]{Martin Moureau}
\author[3]{Benjamin Rodatz}
\author[3]{Boldizsár Poór}
\author[4]{Elie Mounzer}
\author[5]{Linnea Grans-Samuelsson}
\affil[1]{University of Bremen, Department of Mathematics and Computer Science, Bremen, Germany}
\affil[2]{École Polytechnique Fédérale de Lausanne (EPFL), Department of Physics, Lausanne, Switzerland}
\affil[3]{University of Oxford, Department of Computer Science, Oxford, United Kingdom}
\affil[4]{German Research Center for Artificial Intelligence GmbH,  Bremen, Germany}
\affil[5]{University of Oxford, Department of Theoretical Physics, Oxford, United Kingdom}
\date{}

\begin{document}
\maketitle

\begin{abstract}
    Distributed architectures have been proposed as a pathway to large-scale quantum computers.
Combined with the need for fault-tolerance, such architectures require distributed quantum error correction (DQEC) and distributed logical gates.
An important challenge is how to realize DQEC primitives in the setting where interaction between modules is restricted to shared Bell pairs that are significantly noisier than on-chip operations.
We extend the work in~\parencite{rodatzFaultToleranceConstruction2025} on fault tolerance by construction to this setting, deriving different strategies for handling the additional noise.
Through \emph{fault-improvement} we recover conventional entanglement distillation, and also find more dynamical protocols that enable space-time trade-offs.
We show that integrated decoding halves the required distillation code distance compared to entanglement distillation implemented using separate decoding, thus requiring significantly fewer Bell pairs. 
As a main focus of the work, we synthesize efficient circuits for an important primitive in distributed fault tolerance: distributed stabilizer measurements.
These circuits can be adapted to resource constraints, e.g. on the Bell pair generation rate or the space available for on-chip auxiliary qubits.
Noting that full local fault-tolerance is not always needed to preserve the correct scaling of logical error rates, we further optimize the circuits depending on the surrounding context.
We consider in particular the surface code and the color code, both as distributed memories and in the case of lattice surgery across separate modules.
Here, robustness to certain hook and readout errors reduces the number of required Bell pairs even further, compared to the context-free setting.
We numerically benchmark the resulting implementations under circuit level noise with additional interconnect noise.

\end{abstract}

\section{Introduction}

Distributed architectures, also referred to as modular architectures, provide a route towards large-scale fault-tolerant quantum computation in which several comparatively small quantum processing units (QPUs, or modules) are connected by long-range couplers or photonic interconnects.
Non-local operations can then be performed across the modules using entangled states shared between them.

At the level of fault-tolerant circuits, multiple distributed primitives have been proposed.
For instance, non-local logical operations can be realized by transversal CNOTs applied through shared Bell pairs~\parencite{stackTransversalFaultTolerant2026}, or by using logical Bell pairs generated in the same code as logical qubits stored on the modules~\parencite{sunamiEntanglementBoostingLowVolume2026}.
In this work we develop general methods that apply to different approaches, and use them to derive fault-tolerant distributed stabilizer measurements. This is a versatile primitive that can be  used to perform lattice surgery between logical qubits hosted in quantum error correcting codes implemented on different modules, and that can also be used to implement distributed quantum error correcting codes. The latter application opens the door for the implementation of high-rate codes in hardware with boundary-connected modules that each have 2D qubit connectivity: an example that has seen recent interest is hyperbolic surface codes and hyperbolic color codes, where the faces of a hyperbolic lattice map to different modules and where interconnects allow for the realization of a hyperbolic geometry~\parencite{higgottConstructionsPerformance2024,higgottHyperbolic2026}.

A defining feature of distributed architectures is the clear contrast between intra- and inter-QPU operations in both fidelity and speed.
Local two-qubit gates have reached infidelities below $10^{-4}$~\parencite{hughesTrappedionTwoqubitGates2025}, while remotely generated Bell pairs typically have infidelities at the level of single-digit percentages~\parencite{mainDistributedQuantumComputing2025,sahaHighfidelityRemoteEntanglement2025} and are often produced probabilistically at rates far below the on-chip gate speed. While the exact comparison depends on the specific implementation, the differences can span several orders of magnitude.
Distributed fault tolerance must jointly account for higher interconnect noise, limited Bell-pair generation rates, and available space for auxiliary qubits, as discussed in e.g.~\parencite{yuTamingSpacetimeOverhead2026}.

Existing approaches have primarily addressed this challenge at one of two levels.
Works at the ``code level'', meaning the level of the design or evaluation of distributed quantum error-correcting codes, treat inter-QPU gates similar to intra-QPU gates apart from their noise levels, and focus on the consequences of different noise level assumptions.
These assumptions might include that the inter-QPU gates have already been improved to a sufficient level via entanglement distillation, without considering the details of the distillation. Recent work at the code level has focused on specific code families, such as using hyperbolic Floquet codes to reduce the weight and Bell-pair cost of distributed measurements~\parencite{sutcliffeDistributedQuantumError2025}, and partitioning color codes across multiple QPUs with boundary stabilizers measured using nonlocal CNOT gates~\parencite{chandraDistributedRealizationColor2025}.

At the other end, works at the ``distillation level'' consider how multiple noisy Bell pairs can be used to prepare fewer, higher-fidelity  Bell pairs or GHZ states, without considering the intended usage of the improved Bell pairs or GHZ states.
Recent works at the distillation level include using quantum-LDPC codes with iterative decoders to make GHZ-state distillation more scalable and fault-tolerant~\parencite{rengaswamyEntanglementPurificationQuantum2024}, and developing
constant-rate Bell-pair distillation using concatenated error-detecting codes~\parencite{pattisonConstantRateEntanglementDistillation2025} and high-rate quantum-LDPC codes~\parencite{bonillaataidesConstantOverheadFaultTolerantBellPair2025}.

In works at the code level or the distillation level, the handling of interconnect noise through distillation and the implementation of distributed quantum error correction are treated as two separate processes. While this separation of the two processes is conceptually convenient,
optimizing each process independently can miss out on further optimizations that can be done in an integrated setting.
As an example, a treatment that integrates distillation into the larger context
allows for integrated decoding of measurement outcomes originating both from the distillation and the syndrome extraction of the distributed code, making it possible to more accurately identify errors. It also allows for more flexible use of the physical Bell pairs.
Recent works have started to consider both processes together within specific contexts, focusing on particular codes, operations, or noise mechanisms.
\cite{liuRemoteEntanglementLattice2026} analyzes the resource cost of direct and distillation-assisted lattice surgery compared to adjustments in circuit distance in a joint setting, but does not propose an integrated solution,
\parencite{haugLatticeSurgeryBell2025} uses Bell measurements to combine intermodule entanglement with syndrome extraction, and \parencite{vakninFaulttolerantDistributedQuantum2026}
exploits biased interconnect noise to suppress dominant Bell-pair faults through repeated syndrome measurements~\parencite{vakninFaulttolerantDistributedQuantum2026}.
However, a general and systematic approach to integrating the handling of interconnect noise within a larger fault-tolerant context is still lacking.

In the present work, we provide a new approach for systematically reasoning about the synthesis of syndrome extraction and lattice surgery circuits for distributed systems, and more broadly the synthesis of distributed fault-tolerant  circuits.
We build on the correct-by-construction circuit synthesis procedure proposed by~\parencite{rodatzFaultToleranceConstruction2025} and extend it to account for subcircuits that are substantially noisier than the rest of the circuit, such as the links between computing modules.
While our main focus is on distributed stabilizer measurements, this general framework makes the underlying methods extendable to other primitives such as distributed transversal CNOT gates.

The perspective proposed by~\cite{rodatzFaultToleranceConstruction2025} represents a \emph{relational} approach:
rather than asking whether a circuit is fault-tolerant, asking whether two circuits are ``equally fault-tolerant''. This establishes fault tolerance as an invariant that can be preserved under rewriting.
In this relational approach, there are different paths to constructing a fault-tolerant implementation of a given computation.
In the original work, the authors propose starting with an idealised \emph{specification} of the desired computation and its behavior under noise.
While such specifications are often easy to define, they are not usually expressed in terms of the noisy gates available on a quantum computer.
Fault-tolerant circuit synthesis then consists of transforming this specification, using local, provably sound rewrites, into an implementable circuit while preserving the original specification's fault tolerance.
In this work, we also emphasize the utility of a second path: taking an existing \emph{circuit} as a starting point.
In particular, one can start with a circuit known to perform well in the monolithic setting and adapt it to the distributed setting, again using local, provably sound rewrites that preserve fault tolerance despite the higher interconnect noise rate.
This second approach has the advantage of directly incorporating optimizations from the monolithic setting, such as resource savings from tolerance to certain hook errors within a certain global context. In the context of distillation and distributed fault tolerance, we consider \emph{weighted} fault equivalence, defined in~\cite{ruschCompletenessFaultEquivalence2025}, to account for different noise levels.

By applying the approach of fault-tolerance by construction to the construction of distributed stabilizer measurement circuits, we show that we can lower the resource requirements compared to conventional approaches where the distillation is treated as a separate process.
We provide a general method to construct distributed weight-$N$ stabilizer measurements, and focus in detail on weight-4 and weight-6 measurements both in the context-free case and in the context of the color code and surface code.
The relational nature of the synthesis framework enables fault tolerance requirements under monolithic circuit-level noise to be translated into requirements on the distributed circuit.

The synthesized circuits can be adapted to the surrounding context of a given quantum error correcting code or lattice surgery operation, and to the error rate $p_\text{Bell}$ of the physical Bell pairs relative to the on-chip error rate $p$.
In particular, we provide circuits that leverage known
tolerance to certain hook errors in order to further reduce resource costs in the context of distributed memory and inter-QPU lattice surgery for both the surface code and the color code.
In the example of a distributed rotated surface code memory at an interconnect error rate $p_\text{Bell} = p^{1/3}$,
we reduce the number of necessary Bell pairs by $75$\% compared to a context-free approach. (For physical error rates $p\in [10^{-4}, 10^{-3}]$ the interconnect noise in this example corresponds to $p_\text{Bell} \in [1\%, 5\%]$.) This reduction in required Bell pairs is \emph{on top of} the reduction that is achieved by integrated decoding, which halves the required distance of the distillation code compared to when distillation is treated separately.

Beyond reducing overall resource costs, practical applications require consideration of specific hardware constraints. The number of auxiliary qubits that can be used locally to facilitate a distributed circuit will be limited in practice, as will the rate of Bell pair generation.
The synthesis framework of this paper enables flexible space-time tradeoffs that allow for adaptations of distributed implementations to a given hardware.
As an example of this flexibility we transform a circuit with a larger spacelike footprint into a circuit with a larger timelike footprint. The first circuit uses the quantum parity code for entanglement distillation and requires all Bell pairs involved in the distillation to be processed in parallel, while the second circuit processes the Bell pairs sequentially in time.
We also show how to interpolate between the extremes of fully parallel and fully sequential execution, in order to construct intermediate circuits with partial parallelization.

We benchmark the circuits that we construct through Monte Carlo simulations under the SI1000 noise model, with additional interconnect noise. We simulate distributed weight-4 and weight-6 stabilizer measurements, both in the context-free case and for memory and lattice surgery in the surface code and color code. Across all derived approaches, the simulations demonstrate that the circuits have the expected asymptotic behavior, showing that the reductions in resource requirements do not come at the cost of reducing the circuit distance. As an example result that is particularly relevant for near term systems, for distributed lattice surgery in the surface code we find a threshold of $p~10^{-3}$ under tight hardware constraints and with $p_\text{Bell} = p^{1/2} \sim 0.03$. These noise levels are within reach of current hardware, supporting the feasibility of distributed quantum error correction.

The paper is structured as follows:
In \autoref{sec:prelim} we describe the assumptions we make about the interconnects and define the weighted adversarial noise model. We also provide a brief introduction to ZX calculus and fault-equivalent rewrites.
In \autoref{sec:fault_improvement} we introduce the fundamental building blocks that we use to derive fault-tolerant distributed circuits: fault-improvement and distribution edges.
In \autoref{sec:FT_dist_stab_meas}, we derive distributed stabilizer measurement circuits that ensure context-free fault tolerance, with a primary focus on weight-4 and weight-6 stabilizers.
In \autoref{sec:Context_aware} we describe how fault tolerance requirements can be reduced through consideration of a surrounding global context and through integrated decoding.
In \autoref{sec:surface_code} and \autoref{sec:color_code} we derive context-aware distributed circuits for memory and lattice surgery in the rotated surface code and the triangular color code, respectively.
Finally, in \autoref{sec:numerics} we present the results of our numerical benchmarks of the derived circuits.

\section{Preliminaries}
\label{sec:prelim}

In this section, we present the assumptions about distributed quantum computing setups that are used in this work, and how noise is modeled inside the modules and interconnects of such systems.
We then give a brief introduction to the prominent language for building the results of this work: ZX-calculus.
Finally, we summarize the main concepts of fault equivalence, which we use extensively in this work to reason about the behavior of circuits under noise.

\subsection{Setting and Noise Model}

A distributed quantum computing system is defined by multiple independent QPU modules that can interact through long-range interconnects, as shown in \autoref{fig:noise_model_illustration}.
Each module contains qubits dedicated to computation, and we assume interconnects that share photonic Bell pairs between the connected QPUs.
Communication qubits sit at the boundary of the modules; these function like compute qubits, but can also receive Bell pairs through transduction from photon to on-chip qubit technology.
We consider a setting in which the process of Bell pair generation and subsequent transduction has a higher error rate than on-chip operations. The process is treated as a black box, and our focus is on modifications to on-chip circuits to account for the higher error rate.\footnote{In particular, we do not treat photonic entanglement distillation within the present work; the Bell pairs may be pre-distilled through such methods before they are transferred onto the on-chip qubits, in which case the main source of additional noise could be the transduction itself rather than the Bell pair generation and transmission.}

More specifically, the noise model we consider in simulations is circuit level noise of strength $p$ for on-chip operations, and noise of strength $p^{1/\ebitnoise}$ across the entire process of Bell-pair generation and transduction to on-chip qubits, similar to~\cite{victoraEntanglementPurificationQuantum2023,bonillaataidesConstantOverheadFaultTolerantBellPair2025}. This is illustrated in Fig.~\ref{fig:noise_model_illustration}.
In the later sections of the present work, we will take $p^{1/3}$ as the go-to example. This corresponds to interconnect error rates in the range of $\left[0.05, 0.1\right]$ for an on-chip error rate in the range of $\left[10^{-4}, 10^{-3}\right]$. Given recent experimental demonstrations of long-range interconnects, these numbers are already on the conservative side~\parencite{qiu2025, vezvaee2026, niu2023}.

We refer to a noise model in which different parts of the circuit have distinct fault weights as \emph{weighted adversarial noise}. The weights correspond to the exponent $w$ of $p^w$. While this noise model is rather simplified, it captures the challenge of different noise levels and allows notions of fault-tolerance and fault-equivalent circuit rewrites to be conveniently extended from the case of same-weight adversarial noise. 
(While different noise levels could also be modeled through different prefactors, $p$ vs $\alpha p$, it is more cumbersome to track and update prefactors during ZX rewrites than to track and update weights.) 
In numerical simulations, we consider a slightly less simplified noise model based on the SI1000 model, again including additional interconnect noise. 

In addition to assuming that the interconnects have higher noise levels, we assume that the Bell pair generation rate may be limited and that constraints on space and connectivity may limit how many auxiliary on-chip qubits can be used to mediate a distributed stabilizer measurement.
For this reason, we are particularly interested in stabilizer measurement circuits for memory and lattice surgery that trade space for time, thereby relaxing the requirements on the number of auxiliary qubits per distributed stabilizer.

\begin{figure}
    \begin{center}
        \includegraphics[width=.92\textwidth]{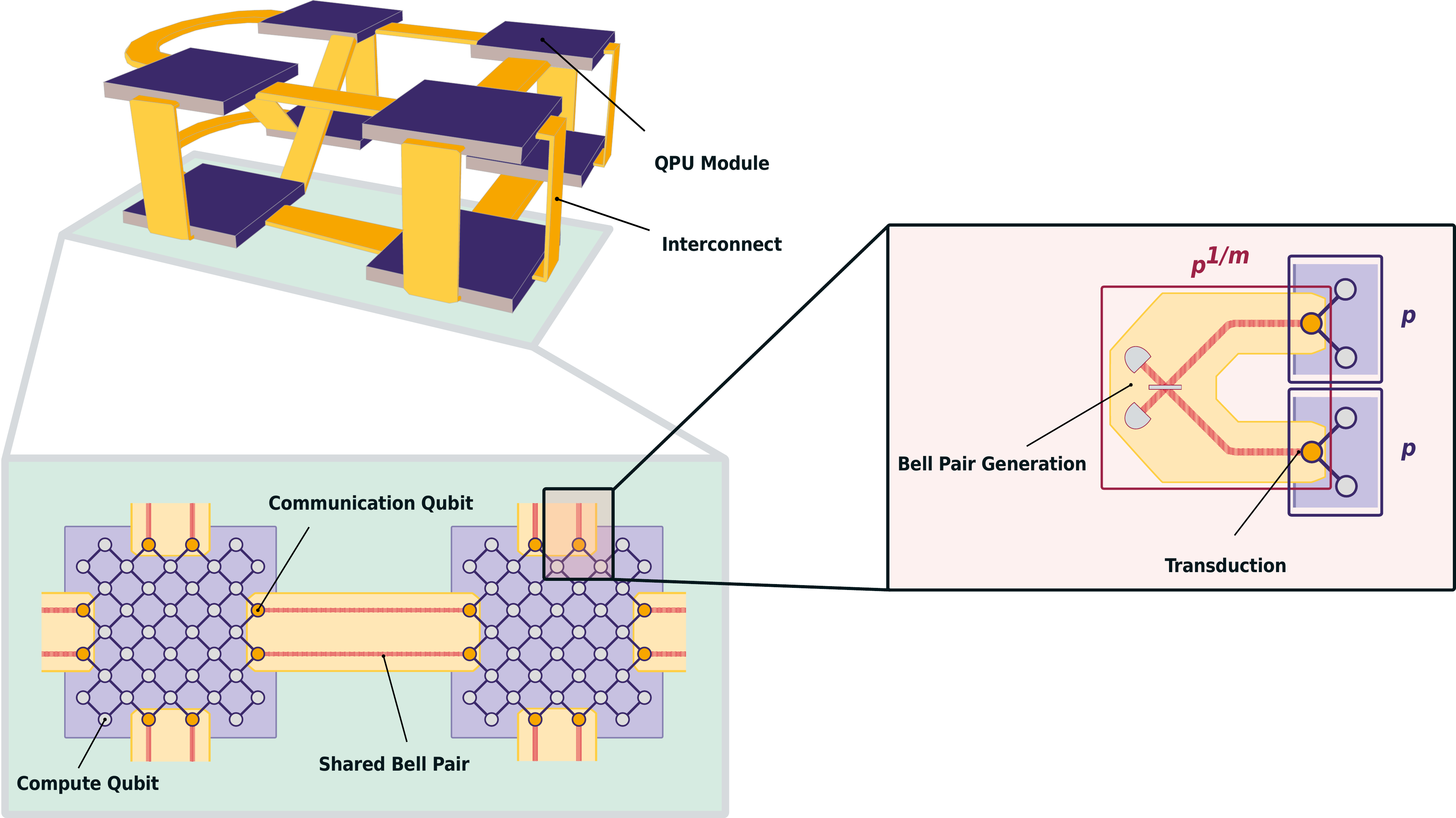}
    \end{center}
    \caption{Illustration of the setting and noise model considered. 
    Multiple QPU modules are connected via interconnects. 
    Each module contains dedicated compute qubits and communication qubits, with the latter used for sharing Bell pairs between QPUs. Following Bell pair generation and transduction, two-qubit Pauli errors $\{I,X,Y,Z\}^{\otimes 2}  \setminus \{I \otimes I\}$ on the on-chip qubits occur with probability $\frac{1}{15}p^{1/\ebitnoise}$ each, while on-chip gates are subjected to circuit level noise with probability $p$. We consider different integer values $\ebitnoise > 1$, so that Bell pairs experience noise at a higher rate than on-chip operations. 
The example setup for eight boundary-connected modules is taken from~\cite{higgottHyperbolic2026}, where it was shown that this setup allows for the implementation of a hyperbolic color code on a tiling of the Bolza surface. (See also~\cite{higgottConstructionsPerformance2024} for earlier work on codes on the Bolza surface and on the use of long-range inter-module connections to realize the geometry.) 
    }\label{fig:noise_model_illustration}
\end{figure}

\begin{figure}[!htb]
    \centering
    \[ \tikzfig[scale=0.4]{02_prelims/TP_LS} \]
    \caption{Logical state teleportation using lattice surgery.}
    \label{fig:tp_ls}
\end{figure}

In a distributed setting of multiple QPUs, there are two scenarios that require interaction between QPUs: (1) If the number of qubits per QPU is heavily constrained and the size of the patch of a quantum error correcting code exceeds the available space, the patch can be distributed over multiple QPUs.
Alternatively, if the on-chip connectivity is not geometrically suited for a desired code, interconnects can be used to realize a desired geometry, as in the example of hyperbolic surface codes and color codes.  
Whether motivated by qubit count or geometry, distributed codes require continuous stabilizer checks via distributed plaquettes along the seam between QPUs throughout the lifetime of the stored qubit. 
 (2) For larger architectures where logical qubits are stored locally, inter-QPU operations are only necessary whenever a multi-qubit interaction between logical qubits stored on separate QPUs needs to be executed. In this scenario, \emph{lattice surgery}
can be used to entangle logical qubits on different QPUs, or to teleport logical qubits between QPUs as illustrated in \autoref{fig:tp_ls}. Lattice surgery is particularly useful in settings with 2D QPUs that are connected only along their boundaries.

\subsection{ZX Calculus}
\label{sec:zx-intro}

When constructing distributed stabilizer measurement circuits, we will use the language of \emph{ZX-calculus}~\parencite{coeckeInteractingQuantumObservables2008}, which is a diagrammatic language for representing and reasoning about quantum circuits.
ZX diagrams correspond to particular types of tensor networks, representing arbitrary linear maps between qubits.
We have:
\begin{definition}
    A ZX diagram consists of:
    \begin{itemize}
        \item A graph $G = (V, E)$
        \item A phase assignment $\alpha: V \to [0,2\pi)$
        \item A type assignment $t : V \to \{X, Z, In(put), Out(put)\}$
    \end{itemize}
    Boundary vertices, i.e.\@ input and output vertices, are restricted to have degree one.
\end{definition}
ZX diagrams consist of the green $Z$ vertices and the red $X$ vertices, respectively representing matrices in the $Z$ and $X$ basis, along with boundary vertices.
We often refer to $X$ and $Z$ vertices as ``spiders'' and their edges as ``legs'':
\begin{align*}
    \textit{Z-spider:} \qquad
    \tikzfig{02_prelims/z-spider}
    \quad & \coloneqq \quad
    \ket{0}^{\otimes n}\! \bra{0}^{\otimes m} + e^{i \alpha} \ket{1}^{\otimes n}\! \bra{1}^{\otimes m} \\[8pt]
    \textit{X-spider:} \qquad
    \tikzfig{02_prelims/x-spider}
    \quad & \coloneqq \quad
    \ket{+}^{\otimes n}\! \bra{+}^{\otimes m} + e^{i \alpha} \ket{-}^{\otimes n}\! \bra{-}^{\otimes m}
\end{align*}

As ZX diagrams represent linear maps, we can translate quantum gates into corresponding ZX diagrams: \\
\begin{minipage}{.32\textwidth}
    \begin{align*}
        \ket{0}
        \quad = & \quad
        \tikzfig{02_prelims/x-state-zero} \\[8pt]
        \ket{1}
        \quad = & \quad
        \tikzfig{02_prelims/x-state-one}  \\[8pt]
        \ket{+}
        \quad = & \quad
        \tikzfig{02_prelims/z-state-zero} \\[8pt]
        \ket{-}
        \quad = & \quad
        \tikzfig{02_prelims/z-state-one}  \\[8pt]
        \tikzfig{02_prelims/not}
        \quad = & \quad
        \tikzfig{02_prelims/not-zx}       \\[8pt]
        \tikzfig{02_prelims/y}
        \quad = & \quad
        \tikzfig{02_prelims/y-zx}
    \end{align*}
\end{minipage}
\begin{minipage}{.32\textwidth}
    \begin{align*}
        \tikzfig{02_prelims/s}
        \quad =                             & \quad
        \tikzfig{02_prelims/s-zx}                                                  \\[8pt]
        \tikzfig{02_prelims/cnot}
        \quad =                             & \quad
        \tikzfig{02_prelims/cnot-zx}                                               \\[8pt]
        \tikzfig{02_prelims/cz}
        \quad =                             & \quad
        \tikzfig{02_prelims/cz-zx}                                                 \\[8pt]
        \tikzfig{02_prelims/z-meas} \quad = & \quad \tikzfig{02_prelims/z-meas-zx} \\[8pt]
        \tikzfig{02_prelims/x-meas} \quad = & \quad \tikzfig{02_prelims/x-meas-zx}
    \end{align*}
\end{minipage}
\begin{minipage}{.32\textwidth}
    \begin{align*}
        \tikzfig{02_prelims/cat-state-4}
        \quad = & \quad
        \tikzfig{02_prelims/cat-state-4-zx} \\[8pt]
        \tikzfig{02_prelims/z-dots-z-parity-check}
        \quad = & \quad
        \tikzfig{02_prelims/z-dots-z-parity-check-zx}
    \end{align*}
\end{minipage}
\\

Beyond representing linear maps, ZX diagrams can also be used for reasoning.
In particular, the ZX calculus is equipped with a set of rewrites that allow for manipulating diagrams while preserving the underlying linear map the diagram represents:
\[\tikzfig{02_prelims/axioms}\]
For more details on the ZX-calculus, see \textcite{vandeweteringZXcalculusWorkingQuantum2020}.

In this work, we will restrict ourselves to the \emph{Pauli fragment} of the ZX calculus by only allowing phases to $\{0,\pi\}$.
It can represent any circuit consisting of Pauli gates, CNOTs, state preparations and measurements~\cite{kissinger2022}.
By convention, spiders drawn without any annotation have a phase of 0.
The phase of a (multi-qubit) measurement is determined by its binary result $k \in \{0,1\}$. For simplicity, we only annotate these spiders with $k$ instead of $k \pi$:

\begin{equation*}
    \tikzfig{02_prelims/z-meas} \quad =  \quad \tikzfig{02_prelims/z-meas-k}\hspace{3cm}
    \tikzfig{02_prelims/x-meas} \quad =  \quad \tikzfig{02_prelims/x-meas-k}
\end{equation*}

\subsection{Fault Equivalence}
ZX-calculus allows us to manipulate quantum circuits while preserving the underlying linear map that they represent.
However, two circuits that implement the same linear map may still behave very differently under noise.
As such, when considering quantum circuits under noise, we not only need to preserve the computation the circuit implements, but also its behaviour under noise.
This stricter notion of equivalence is called \emph{fault equivalence}~\parencite{rodatzFaultToleranceConstruction2025,rodatzFloquetifyingStabiliserCodes2024}.

First we need to define noise on both ZX diagrams and circuits. For the latter, we will focus on Clifford circuits.
By defining noise on tensor networks, the definition naturally extends to both settings.

\begin{definition}[Faults in space-time]
    Let $D$ be a Clifford tensor network with edges $E$.
    A fault $F$ is a Pauli acting on these edges, i.e. an element in $\mathcal{P}^{|E|}$, with $\mathcal{P}$ the one-qubit Pauli group, indicating the action of $F$ on each edge of $D$. We denote $D^F$ for the tensor network we get when placing the Pauli rotations indicated by $F$ on each edge in $D$.
\end{definition}

For Clifford circuits, the locations where errors can occur correspond to the qubits at the time steps in between the gates.
As we consider faults all over the circuit, we consider faults in space-time, accounting for the interaction of faults at different time steps.

Next we consider the likelihood of faults.
Intuitively, we should consider faults that act on all fault locations as less likely than faults acting on only one fault location.
More formally, we model this by assigning faults weights:
\begin{definition}[Noise model]
    Let $D$ be a Clifford tensor network with edges $E$.
    Then an atomic noise function $\mathcal{A}: \mathcal{P}^{|E|} \to \mathbb{R}^+ \cup \{+\infty\}$ assigns each fault an atomic weight.
    This induces a noise model $\mathcal{N}_{\mathcal{A}}: \mathcal{P}^{|E|} \to \mathbb{R}^+ \cup \{+\infty\}$ defined as:
    \[
        \mathcal{N}_{\mathcal{A}}(F) \;=\; \min_{\mathcal{A}_F \subseteq \mathcal{P}^{|E|} \text{ such that } \prod_{F' \in \mathcal{A}_F} F' = F}  \left(
        \sum_{F' \in \mathcal{A}_F} \mathcal{A}(F)
        \right)
    \]
\end{definition}

The atomic weight function defines the independent error mechanisms, such as gate errors, measurement errors or qubit flips.
Each of the atomic faults gets a weight indicating its likelihood, with larger weight faults being considered less likely.
For example, we could imagine a noise model, where measurement errors have weight $1$ and idling errors have weight $1.5$.
A fault having weight $w$ corresponds to the probability of that fault scaling as $p^w$ for some physical error rate $p$.
The weights need not be integers; they can naturally be extended to non-negative reals. We will use fractional weights in the present work.  

The independent atomic errors each may or may not happen in a given execution of the circuit, and combinations of atomic faults create specific space-time errors when they occur.  
The weight of a space-time error is determined from the noise model by iterating over all combinations of atomic faults that result in the error and picking the lowest-weight one.
We can view the weight as the cost an adversary has to pay when trying to create a specific fault: to create a large combination of atomic faults, the adversary can pick the cheapest equivalent combination. 
In this framework, the circuit distance $d$ of a circuit-level implementation of an error correcting code is the minimum weight of an undetectable logical fault.
Since we do not restrict weights to integers, circuit distances can be fractional.

Within the general noise model definition we can consider multiple noise models, including circuit-level noise and SI$1000$-like models. 
In particular, we will consider two different families of noise models.
We have:
\begin{definition}[Circuit-level noise]
    Let $C$ be a quantum circuit, then the atomic noise function for circuit-level noise assigns all faults a weight of $+\infty$, except for the following:
    \begin{itemize}
        \item For each wire:
              \begin{itemize}
                  \item \textbf{Qubit flips:} A fault acting only on that wire.
              \end{itemize}
        \item For each gate (including Pauli measurements):
              \begin{itemize}
                  \item \textbf{Gate faults:} Faults acting on any subset of the output wires of the gate.
              \end{itemize}
        \item For each Pauli measurement:
              \begin{itemize}
                  \item \textbf{Measurement flip:} An anticommuting Pauli applied immediately before and after the measurement.
                  \item \textbf{Measurement flip + gate faults:} A measurement flip and a fault acting on any subset of the outputs of the measurement.
              \end{itemize}
    \end{itemize}
\end{definition}
In the most simple model, any non-$\infty$ atomic faults get a weight of $1$.
However, at times, we might give them different weights.
For example, we might give qubit flips on links between modules lower weights, indicating that they are more likely.

For ZX diagrams, on the other hand, we do not generally have a corresponding notion of gates with inputs and outputs.
Therefore, we define a different noise model on ZX diagrams:
\begin{definition}[Edge-flip noise]
    Let $D$ be a ZX diagram, then the atomic noise function for circuit-level noise assigns all faults a weight of $+\infty$, except for faults acting non-trivially on exactly one edge.
\end{definition}

In both cases, we can assign more complicated expressions for the likelihood of any given fault. As an example, we may keep a single variable $p$ and describe the relative noise level of each atomic fault by two parameters $\alpha,w$ to describe a fault occurring with probability $\alpha p^w$. Within the circuit-level noise family, the two-parameter subfamily contains the SI1000 noise model used in our numerical benchmarks (see \autoref{app:CLN}).

Throughout this work, we will be considering ZX diagrams under edge-flip noise and quantum circuits under circuit-level noise.
Since we consider different noise models, a direct translation between a circuit and a ZX diagram under noise is \emph{possible}, but \emph{not guaranteed}, even if their linear maps are equivalent.

Next we define:
\begin{definition}[Detectable faults]
    Let $D$ be a non-zero Clifford tensor network with edges $E$.
    A fault $F \in \mathcal{P}^{|E|}$ is \emph{detectable} if $D^F = 0$.
\end{definition}
We can view circuits and ZX diagrams as being post-selected on an expected measurement outcome, i.e. on an outcome with trivial syndrome.
If $D^F = 0$, then, according to the Born rule, the probability of observing that outcome is zero, or, in other words, we will always see a measurement outcome with a non-trivial syndrome --- the error is detected. (We emphasize that while the postselection viewpoint can be useful for derivations, the circuits that we derive do not rely on postselection, and the simulations only use postselection for examples with $d<3$.)

Finally, we can define:
\begin{definition}[Fault equivalence]
    Let $D_1, D_2$ be two tensor networks that implement the same linear map.
    Let $\mathcal{N}_{\mathcal{A}_1}, \mathcal{N}_{\mathcal{A}_2}$ be noise models on $D_1$ and $D_2$, respectively. We say $D_1$ under $\mathcal{N}_{\mathcal{A}_1}$ is $w$-fault-equivalent to $D_2$ under $\mathcal{N}_{\mathcal{A}_2}$ if and only if for all undetectable faults $F_1$ on $D_1$ with $\mathcal{N}_{\mathcal{A}_1}(F_1) \leq w$, there exists a fault $F_2$ on $D_2$ such that $D_1^{F_1} = D_2^{F_2}$ and $\mathcal{N}_{\mathcal{A}_2}(F_2) \leq \mathcal{N}_{\mathcal{A}_1}(F_1)$ and analogously with $(D_1, \mathcal{N}_{\mathcal{A}_1})$ and $(D_2, \mathcal{N}_{\mathcal{A}_2})$ interchanged.
    We write $(D_1, \mathcal{N}_{\mathcal{A}_1}) \wFaultEq{w} (D_2, \mathcal{N}_{\mathcal{A}_2})$ and omit $\mathcal{N}_{\mathcal{A}_1}, \mathcal{N}_{\mathcal{A}_2}$ when they are obvious from context.
    We say $(D_1, \mathcal{N}_{\mathcal{A}_1})$ and $(D_2, \mathcal{N}_{\mathcal{A}_2})$ are fault-equivalent, written $(D_1, \mathcal{N}_{\mathcal{A}_1}) \FaultEq (D_2, \mathcal{N}_{\mathcal{A}_2})$, when they are $\infty$-fault-equivalent.
\end{definition}

The main idea of fault equivalence is that for every fault on $D_1$ there exists a corresponding fault on $D_2$ of at most the same weight, and vice versa.
This means that the cost of creating a fault in some equivalence class on $D_1$ must be the same as the cost of creating that fault on $D_2$.

The rewrite rules shown in \autoref{sec:zx-intro} do not necessarily preserve fault equivalence. But, by restricting to the subset of \emph{fault-equivalent} rewrites we can manipulate ZX diagrams while maintaining their behaviour under edge-flip noise.
Similarly, we have to be careful when translating between ZX diagrams under edge-flip noise and quantum circuits under circuit-level noise.
In particular, the translations outlined in \autoref{sec:zx-intro} are not necessarily fault-equivalent.
However, we have:
\begin{proposition}
    \label{prop:CLN_CX}
    The following two ZX diagrams under edge-flip noise are fault-equivalent to the corresponding quantum circuits under circuit-level noise:
    \[
        \tikzfig{02_prelims/cnot-plus-state}
        \quad \FaultEq \quad
        \tikzfig{02_prelims/fault-accounting-cnot-plus-state}
        \qquad \qquad \qquad \qquad \qquad
        \tikzfig{02_prelims/cnot-measurement}
        \quad \FaultEq \quad
        \tikzfig{02_prelims/fault-accounting-cnot-measurement}
    \]
\end{proposition}
\begin{proof}
    See~\textcite[Prop. 6.8 and 6.10]{rodatzFaultToleranceConstruction2025}.
\end{proof}

These two circuits and diagrams correspond to the split and merge operation, known from lattice surgery~\parencite{debeaudrapZXCalculusLanguage2020}.
As fault equivalence is compositional~\parencite[Prop. 3.12]{rodatzFaultToleranceConstruction2025}, this means that whenever we can identify a ZX diagram that consists of splits and merges, we can fault-equivalently extract a quantum circuit.
\FloatBarrier
\section{Weighted Fault Tolerance and Fault Improvement}\label{sec:fault_improvement}
The key insight underlying this work is that we can use weighted noise to account for the effects of particularly noisy components in a quantum circuit, such as the links connecting modules in a distributed quantum computer.
Our goal is to ensure that an implementation containing such noisy interconnects is fault-equivalent to an idealised system in which all circuit components operate at a common, lower physical error rate.
In other words, the effect of the noisy interconnects is suppressed to make the computation behave as though it were executed on a monolithic quantum computer.
We do this by providing stronger protection to the noisier circuit elements, compensating for their increased error rate.

For the purposes of constructing a distributed syndrome measurement circuit, we will use distributed $XX$ and $ZZ$ measurements as key building blocks. Up to classical corrections to account for the measurement outcomes, these correspond to edges joining same-color phase-free spiders. We refer to these as \emph{distribution edges}.

We consider on-chip physical errors to have an error weight of $1$, while the noisy interconnects have an error weight of ${w} = \frac{1}{\ebitnoise} < 1$, so that errors on the interconnects scale with $p^{w}$. As a notational convention, while edges normally have three weights (for $X, Y$ and $Z$ respectively), if we only write one weight, we assume all of them to be the same.
Unlabelled edges are assumed to have weight $1$.
If components or edges are idealised as fault-free, we omit a weight annotation and draw them in purple. In terms of weights, fault-free corresponds to weight $+\infty$. An example of weight annotations is shown in \autoref{fig:w_example}.
\begin{figure}[!htb]
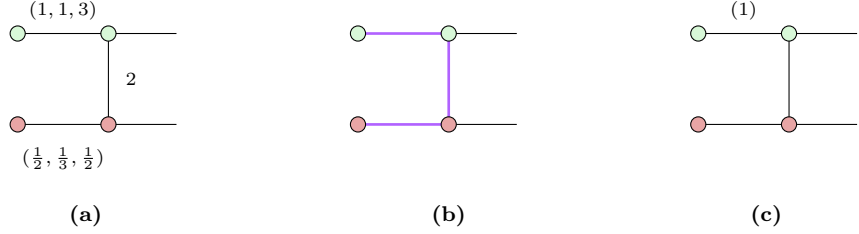

  \[ \tikzfig[scale=0.6]{03_FT_weighted/weighted_example} \]
  \caption{Bell state preparation under different edge-flip noise models: (a) Symmetric and asymmetric noise: For edges with asymmetric noise, fault-weights are annotated as tuples $(w_x, w_y,w_z)$. Edges with symmetric noise $w_x=w_y=w_z=w$ are simply labeled with $w$. (b) Noise-free components: Edges idealized as fault-free are drawn in purple. (c) Default noise: Edges with symmetric noise $w=1$ are left unlabelled.}
  \label{fig:w_example}
\end{figure}

We define:
\begin{definition} [Fault-improved distribution edge]
  \label{def:FI-DE}
  A circuit $C$ implements a fault-improved distribution edge of $XX$ type based on noisier distribution edges with fault weights $(w,w,w)$, $w = \frac{1}{\ebitnoise}$, with an improvement $(d_X,d_Y,d_Z)$, if and only if:
  \[ \tikzfig[scale=0.5]{03_FT_weighted/FI-DE} \]
\end{definition}
\noindent
where the dashed line indicates the separation between the two modules.
The $ZZ$ type definition follows analogously.

\subsection{Distillation}
A key strategy to protect quantum information against noise is using quantum error correction.
Similarly, we can use quantum error correcting codes to protect subcircuits against increased noise levels.

First, observe:
\begin{proposition}\label{prop:distillation}
  Let $enc: \mathbb{C}^2 \to (\mathbb{C}^2)^{\otimes n}$ be the encoder of an $[\![n, 1, d]\!]$ code.
  Then we have:
  \[ \tikzfig{03_FT_weighted/ED_FT_enc} \]
\end{proposition}
\begin{proof}
  This first step follows from the fact that the code can detect any non-trivial fault of weight less than $d$.
  Therefore, if $d$ faults of weight $w$ occur, they can create a logical operator in the code that propagates through the encoder to a fault on the outer edge.
  The second step follows from the fact that, if everything is completely idealised, we can simply perform regular (semantic) ZX rewrites.   As encoding our data and immediately decoding it again is semantically equivalent to doing nothing, we can remove the encoder and its adjoint.
\end{proof}

The above proposition states that if we can encode quantum information fault-freely into a larger code space, then the encoded information is more protected against increased noise levels. By bending the edges of the diagram, we recover the approach in~\cite{shi2025a}, where Bell pair distillation protocols are constructed from stabilizer codes\footnote{See also Craig Gidney's pedagogical summary of the parallels between error correction in quantum memories and Bell pair distillation in the stack exchange thread \url{https://quantumcomputing.stackexchange.com/questions/33661/understanding-entanglement-distillation-via-stabilizer-codes}.}:

\[\tikzfig{03_FT_weighted/distillation}\]
This formally proves that we can use $n$ noisy Bell pairs to distill a single, higher quality Bell pair.
In particular, it gives us a precise quantification of what code we should choose: if the Bell pairs have weight $w=\frac{1}{\ebitnoise}$, we should choose a code with distance $d = \ebitnoise$. In later sections, where distillation is performed in order to implement stabilizer measurements in a distributed code, we will sometimes refer to the distance of a code used for distillation as $\dBell$, to distinguish it from the distance of the distributed code itself.

\begin{remark}
  We note that we consider only codes with one logical qubit.
  This is because, in general, we have no guarantees on the weight of the logical operator that might be created from $d$ faults of weight $\frac{1}{\ebitnoise}$ among the noisy Bell pairs.
  Thus, if we have multiple logical qubits, $d$ faults of weight $\frac{1}{\ebitnoise}$ may affect multiple distilled Bell pairs, and achieving the same guarantees as above would require a code with distance $d = \ebitnoise \times k$ to compensate for this. This constitutes such a significant increase in distance that, in practice, it is easier to encode the Bell pairs individually.
  However, having stronger guarantees for the effect of minimum-weight logical errors on the distilled Bell pairs could allow for less stringent requirements and make (some) high-rate distillation codes useful.
\end{remark}

The above derivations assume that the encoder is implemented fault-freely, which is impossible in practice.
Therefore, we define:
\begin{definition}[Fault-tolerant implementation of an encoder]
  Let $enc: \mathbb{C}^2 \to (\mathbb{C}^2)^{\otimes n}$ be the encoder of an $[\![n, 1, d]\!]$ code. A fault-tolerant implementation $C$ of $enc$ is a circuit that satisfies:
  \[ \tikzfig{03_FT_weighted/state_encoding_weighted} \]

\end{definition}

The implementation of a fault-tolerant encoder depends on three characteristics: the weights of the internal errors, the set of available operations, and the code structure. More specifically, fault tolerance depends on the interplay among these three: Error weights identify which errors require extra care, the available operations set the error propagation pattern, and the code structure determines both the order of operations and the relevance of the errors.

The proposition below is a sufficient condition for implementing a fault-tolerant entanglement distillation protocol.

\begin{proposition}[On-chip distillation]\label{prop:on-chip-dist}
  Let $enc: \mathbb{C}^2 \to (\mathbb{C}^2)^{\otimes n}$ be the encoder of an $[\![n, 1, d]\!]$ code, and $C$ the circuit that implements $enc$.
  An entanglement distillation protocol using $C$ is fault-tolerant if all atomic faults in $C$ have weight $w=1$.
\end{proposition}
\begin{proof}
  Errors that directly occur on one of the two output edges of a distillation protocol have a weight $w=1$. An atomic fault inside $C$ of weight $w=1$ can, in the worst case, propagate to a correlated fault on both output edges. However, since the output edges form a Bell pair, all correlated faults on both output edges are equivalent to (at most) a single fault on one output edge. Therefore, any error inside $C$ is at most fault-equivalent to an output error and can be neglected.
\end{proof}

This means that in the scenario under consideration, where only the Bell pairs have lower noise weights and all other operations of the distillation are executed \emph{on-chip} with weight $w=1$, we can use any valid encoding circuit and can disregard internal error propagation.\footnote{
  In contrast, if parts of the distillation get moved \emph{off-chip} where lower noise weights apply, e.g. through photonic operations before transduction, the same encoding circuits may no longer be sufficient.}

If we have a fault-tolerant distillation protocol, we can use it to construct a fault-improved distribution edge:
\begin{proposition}\label{prop:FE-DE-distillation}
  Given a fault-tolerant distillation protocol using an encoder $enc$ and circuit $C$ of a code with distance $d$, we have:
  \[\tikzfig{03_FT_weighted/FE-DE-with-distillation}\]
\end{proposition}
\begin{proof}
  \[\tikzfig{03_FT_weighted/FE-DE-with-distillation-proof}\]
\end{proof}

\subsection{Repetition}\label{sec:Repetition}

While distillation generally protects edges, we will specifically use \emph{distribution edges} as building blocks. As an alternative to distillation, distribution edges can also be protected via {repetition}:
\begin{proposition}
  For fault weights $v_X, v_Y, v_Z, w_X, w_Y, w_Z > 0$ the following rewrite is fault-equivalent:
  \label{prop:fault-improve}
  \[ \tikzfig{03_FT_weighted/improve_rewrite} \]
  where the weights on the RHS are bounded at $w'_Y, w'_Z \leq 2$ and $w'_X \leq 1$.
\end{proposition}
\begin{proof}
  To show this, we use the approach presented in~\cite{ruschCompletenessFaultEquivalence2025}: We use fault-equivalent rewrites to bring both constructions to an equal common diagram that consists of an idealised semantically equivalent circuit together with fault gadgets that account for all noise.
  We separately show the equivalence for $X$ and $Z$ type errors, to keep the diagrams manageable, and discuss $Y$ type errors in \autoref{sec:yproof}.
  Beginning with $X$ type errors, we can push out all inner faults to a single edge as the entire diagram just has a single corresponding stabilizer, and then just keep the highest probability fault:
  \[ \tikzfig{03_FT_weighted/square_x} \]
  This applies analogously to the improved edge case.
  \[ \tikzfig{03_FT_weighted/vert_x} \]
  For  $Z$-type errors, we must first resolve the inner detecting region, generating all possible undetectable pairs of inner faults. We then simplify and remove redundant and less likely faults.
  The only remaining faults are the already existing single outer faults, and a new correlated two-qubit error of weight $v_Z + w_Z$ that spreads to either the top or bottom qubits:
  \[ \tikzfig{03_FT_weighted/square_z} \]
  The case of the improved distribution edge is simpler: here, we only have a single inner fault, which gets pushed out to the same kind of two-qubit error as in the other case:

  \[ \tikzfig{03_FT_weighted/vert_z} \]
\end{proof}

\autoref{prop:fault-improve} protects distribution edges against errors that flip the corresponding measurement outcome: for a $ZZ$ type distribution edge, it yields an improvement in the $Y$ and $X$ bases. With a larger number of repetitions, this improvement is made stronger.
Assuming for the sake of simplicity that all faults on the noisy edges have the same weights, we have:
\begin{proposition}\label{prop:repetition_protection}
  For $w_X=w_Y=w_Z=\frac{1}{\ebitnoise}$, we have:
  \[\tikzfig{03_FT_weighted/repeating-measurement}\]
\end{proposition}
\begin{proof}
  This follows from repeatedly applying \autoref{prop:fault-improve}.
\end{proof}

Interpreting $Z$ faults on a noisy $XX$ type distribution edge as measurement errors, and $Y$ faults as measurement errors that furthermore create an outgoing $X$ fault, an undetectable measurement error can be created by flipping all $d$ measurements, while an undetectable combined measurement error and outgoing $X$ error can be created through $d-1$ measurement flips together with one measurement flip that also causes an outgoing $X$ fault. More generally, to create these undetectable errors, an even or odd number of the measurement flips should be combined with an outgoing $X$ error, respectively. Taking $\ebitnoise_Y = \ebitnoise_Z = \ebitnoise$, the lowest-weight fault that can cause either of these undetectable errors has $w = \frac{d}{\ebitnoise}$, meaning that the protection against $Y$ faults on the distribution edge is as strong as the protection against $Z$ faults.

To further add protection against $X$ faults, we employ nested repetition:
\begin{proposition}\label{prop:nested_reps}
  \begin{equation}
    \tikzfig{03_FT_weighted/edge-both-bases-updated}
  \end{equation}

  and equivalently for an improved distribution edge of $ZZ$ type, swapping the bases of the inner and outer repetitions.
\end{proposition}

\begin{proof}
  \[\tikzfig{03_FT_weighted/both-edge-bases-proof}\]
\end{proof}
From here on we will use $\dout$ to denote the distance created by outer repetitions and $\din$ as the distance from inner repetitions, e.g. $\dout=d_Z$ and $\din=d_X$ for an $XX$ type distribution edge. For consistency, we also keep this labeling even if the distribution edge was constructed using entanglement distillation (\autoref{prop:FE-DE-distillation}).

Finally, we translate the above diagram into an implementable distributed circuit and consider circuit level noise. We have:
\begin{proposition}\label{prop:repeat_circuit}
  \[\tikzfig{03_FT_weighted/repeat-circuit}\]
\end{proposition}
\begin{proof}
  First, we can translate the quantum circuit into a fault-equivalent ZX diagram:
  \[\tikzfig[scale=0.45]{03_FT_weighted/repeat-circuit-proof}\]
  All inner CNOT gates fulfill the requirements of \autoref{prop:CLN_CX} and are fault-equivalent to their ZX counterparts. For the outer CNOT gates, we can make use of the fact that their target sits in the first inner detecting region and only combined errors which have equal weight under edge-flip noise remain undetected.
  Also notice that we translate single-qubit measurements into empty phase spiders instead of $k$-spiders.
  In practice, we have to insert classically controlled correction operations into the quantum circuit to have a faithful translation.
  We omit this here for visual clarity (see \autoref{sec:corrections} for details).

  Now, we can simplify and apply the nested improvement of \autoref{prop:nested_reps}:
  \[\tikzfig[scale=0.45]{03_FT_weighted/repeat-circuit-proof2}\]
\end{proof}

We note that the overall circuit is better protected against $Y$ faults than against $X$ and $Z$ faults, as both the inner and outer measurements protect against them. In the case of noisy interconnects with biased noise, the choice of basis can be made to provide the strongest protection against the most likely type of fault.

\subsection{Quantum Parity Code Encoder}\label{sec:QPC_encoder}

The concatenation in \autoref{prop:nested_reps} of an inner measurement repetition protecting against $X$ errors with an outer measurement repetition protecting against $Z$ errors is reminiscent of the Shor code at $d_X = d_Z = 3 $, or more generally of quantum parity codes (QPCs).
Indeed, the scheme of nested repetitions can be fault-equivalently rewritten as a distillation protocol using a QPC, with an encoder and an inverse encoder\footnote{We avoid the term decoder here, as it is also used to refer to classical algorithms for processing syndrome data.} defined as follows:

\begin{definition}
  \label{def:shor_encoder}
  An encoder $enc_\text{QPC}$ and inverse encoder  $enc_\text{QPC}^\dagger$ for the QPC with distance $d_X$ for $X$-type faults and distance $d_Z$ for $Z$-type faults implement the following map, with the encoder mapping right to left and the inverse encoder mapping left to right:
  \[ \tikzfig{03_FT_weighted/shor_enc_map} \]
\end{definition}
\noindent
The above diagram can be rewritten into an implementable circuit of the form
\[ \tikzfig{03_FT_weighted/shor_enc_circ} \]
This  implementation is not fault-tolerant on its own, but the full distillation circuit using it accomplishes the desired edge improvement:

\begin{proposition}
  \label{prop:shor_improv}
  Entanglement distillation using the above implementation of $enc_\text{QPC}$ with $X$-type distance $d_X$ and $Z$-type distance $d_Z$ for Bell pairs with fault weights $\frac{1}{\ebitnoise}$ and local fault weights of 1 is fault-equivalent to a fault-improved single edge as follows:
  \[ \tikzfig{03_FT_weighted/ED_shor_improv} \]
  where the weights on the RHS are bounded at $w'_X, w'_Y, w'_Z \leq 1$.
\end{proposition}
\begin{proof}
  See \autoref{app:Shor}.
\end{proof}

\begin{proposition}\label{prop:QPC_encoder}
  The default QPC-encoding circuit under circuit-level noise is fault-equivalent to its edge-flip noise ZX diagram:
  \[\tikzfig{03_FT_weighted/shor_enc_CLN}\]
\end{proposition}
\begin{proof}
  The $X$- and $Z$-basis measurements circuit-level and edge-flip noise are directly equivalent.
  All CNOT gates have a fitting directly following measurement according to \autoref{prop:CLN_CX}, such that all correlated faults from the circuit-level noise model have a direct corresponding fault in the edge-flip noise model of the same weight.
\end{proof}

\subsection{Generalizations}

It is possible to combine the two methods of distillation and repetition. One example is to use repetition to protect against errors in one basis and distillation with a classical error correcting code to protect against errors in the other basis. Alternatively, distillation can be done with a quantum error correcting code of $d_X \neq d_Z$, with repetition making up for the lacking distance in one basis. Through different choices, the edge improvement strategy can be adapted to available resources, and the extent of improvement across different bases can be tuned as needed. (We will see in later sections that in certain contexts, one basis requires more improvement than the other in order to achieve fault tolerance of the larger circuit that the distribution edge is embedded in.)

Yet another generalization follows from the observation that the nested measurement repetitions in \autoref{sec:Repetition} can be viewed as a version of the QPC encoder that is ``stretched out in time'', thereby reducing the space overhead at any given time step. In \autoref{prop:QPC_encoder}, $n=d_Xd_Z$ auxiliary qubits are required per QPU to implement the QPC encoder.
Meanwhile, in \autoref{prop:repeat_circuit}, only two auxiliary qubits per QPU are needed at any given time. (Recall the setup illustrated in \autoref{fig:tp_ls}: we consider the transfer to on-chip auxiliary qubits as part of the black-box interconnect, and the top and bottom half-edges of the bent edges in \autoref{prop:repeat_circuit} correspond to such auxiliary qubits.)
Encoding circuits for other codes can be similarly deformed to reduce the space overhead.

To showcase the generalization to other codes, we consider the Steane code. Compared to the $d_X=d_Z=3$ QPC code (the Shor code), the Steane code requires fewer data qubits (in the distillation setting: Bell pairs), 7 instead of 9, while its encoding circuit requires more CNOTs, 11 instead of 8.
In \autoref{app:Steane}, we show how different encoding circuits for the Steane code can be constructed. The following $XX$ type distribution edge improved by $d_X=d_Y=d_Z=3$ is one example:
\begin{equation}\label{eq:steane_example}
  \tikzfig[scale=0.5]{03_FT_weighted/Steane_edge}
\end{equation}

We see that the above circuit reduces the space overhead compared to default Steane code distillation, requiring 5 auxiliary qubits per QPU rather than 7.

Similarly, one can attempt to obtain stretched-out
circuits that reduce the space overhead for other codes, starting from ZX-representations of their encoding circuits. For CSS codes in particular, it has been shown in~\cite{huang2023} that the encoding circuits can be represented by a bipartite ZX diagram consisting of one layer of green spiders followed by one layer of red spiders, as in the diagram used in \autoref{app:Steane} for the Steane code.  
We infer from the Steane example that higher rate codes reduce the number of required Bell pairs at the cost of more interleaved measurements, which reduces the room for space-time tradeoffs.
For a longer discussion, see \autoref{app:Steane}.

\section{Fault-Tolerant Distributed Stabilizer Measurements in the Context-free Case}\label{sec:FT_dist_stab_meas}

After establishing the foundational tools for building fault-improved distribution edges, we now consider how they can be used to construct fault-tolerant distributed stabilizer measurement circuits.
The presentation focuses on stabilizers of weight $4$ and $6$ split across two QPUs; for general weight $N$ stabilizer measurements split across two or more QPUs, the fault-equivalent decompositions of $N$-legged spiders shown in~\cite{rodatzFaultToleranceConstruction2025} can be rewritten into implementable distributed circuits by partitioning of the decomposed diagram across the desired number of QPUs and using fault-improvement on any edge crossing between partitions.

In this section, the derivations of distributed syndrome extraction circuits start from the high-level ZX-representation of weight $N$ $Z$-type stabilizer measurements shown in \autoref{sec:zx-intro}.
Since we do not consider post-selected measurements, the $N$-legged spider will be given a phase $k\pi$ to represent the measurement outcome, with $k = 0 (1)$ representing the measurement of a $+1 (-1)$ stabilizer eigenvalue.
Stabilizers of $X$-type can be obtained by swapping the $X$ and $Z$ bases for each spider, while general Pauli stabilizers can be obtained by conjugation with $S$ or $H$ on the relevant data qubits.
The extraction circuits derived in this section are \emph{context-free} fault-tolerant, meaning that any set of inner faults of weight $w$ is only allowed to propagate to a set of data qubit and readout errors of at most equal weight.
In later sections, where specific codes are considered, we also show derivations that start from existing monolithic syndrome extraction circuits that are only fault-tolerant in the context of their specific codes and adapt them to the distributed setting.

The splitting of the weight $4$ and $6$ stabilizers is done by rewriting the ZX-diagrams of the corresponding stabilizer measurements as implementable circuits.
We indicate the transition from the abstract functional diagram to a concrete distribution setting by a squiggly arrow. The split of the ZX-diagram into two partitions is denoted as a horizontal grey dashed line and distribution edges crossing between the partitions as thick blue dashed lines with variable Bell pair noise strength $m$ and improvement $d$ for each Pauli type.
We then continue with fault-improving the distribution edges (with the amount of fault-improvement depending on the interconnect noise under consideration), and finally rewriting the fault-improved edges in terms of Bell pair mediated operations. There are multiple possible partitionings of the $4$ or $6$ data qubits between the two QPUs. We focus on the $2-2$ split for the weight $4$ stabilizer, and the $3-3$ and $4-2$ splits for the weight $6$ stabilizer.

\subsection{Weight-4 Stabilizer Measurements}\label{subsec:Weight-4-stabilizers}

Starting from the high-level representation of weight-4 $Z$ type stabilizer measurement, different implementable circuits can be derived.
One option is to unfuse the red spider to obtain a syndrome extraction that uses a single auxiliary qubit, which we will also refer to as the simple circuit:
\begin{equation}\label{naive_sc_circuit}
    \tikzfig[scale=0.5]{04_FT_dist_stab/w4_circification_mono}
\end{equation}
The two diagrams above are semantically equivalent, but they are not \emph{fault-equivalent} with respect to edge-flip noise: there is no fault on the diagram to the left that is equivalent to a single fault between the third and fourth red spider on the auxiliary qubit in the diagram to the right. The latter fault is a \emph{hook error}, propagating to the two lower (or, equivalently, two upper) data qubits.

A hook-free, fault-tolerant circuit can be constructed by instead using fault-equivalent rewrites~\cite{rodatzFaultToleranceConstruction2025}. This can be accomplished either with two auxiliary qubits if two-qubit measurements are natively available, or with a third auxiliary measurement qubit using  CNOT gates to mediate the two-qubit measurement. The version with two auxiliary qubits will form the basis for fault-tolerant distributed circuits, and takes the form
\begin{equation}\label{ft_sc_circuit}
    \tikzfig[scale=0.4]{04_FT_dist_stab/w4_circif_square}
\end{equation}
for $k_1 \oplus k_2 = k$.

We can rewrite the circuits of both \autoref{naive_sc_circuit} and \autoref{ft_sc_circuit}  into distributed weight-$4$ extraction circuits with a $2-2$ split, where the data qubits are split as pairs on each QPU. In the context of rotated surface codes, this allows for an efficient splitting of logical patches and lattice surgery operations along straight seams, as we will discuss in detail in the later sections of the paper.

The first goal of our distribution synthesis is to transform the monolithic extraction circuits into a form where the data qubits are correctly split between QPUs and only distribution edges operate between QPUs in the form of Bell pairs or two-qubit measurements.
\autoref{naive_dist_circuit} shows the distribution synthesis for the simple extraction circuit.
In the first step, we split the circuit between modules as indicated by the dotted line i.e., the first half of the circuit is executed on the first QPU and the second half on the second QPU\footnote{The syndrome extraction circuit of \autoref{naive_sc_circuit} also allows for a $3-1$ split, where the distribution edge can be bent into a Bell pair in the same way as for the $2-2$ split.
    In contexts where some hook errors are benign (as will be discussed in later sections), this can allow for a reordering of the data qubits that achieves global fault tolerance.
    If the distribution edge is not sufficiently improved to compensate for the additional interconnect noise, it also ensures that the hook error prone part of the circuit (the edge between the third and fourth spider of the auxiliary qubit) is not subject to more noise than in the monolithic setting. For the purpose of this section, where we consider context-free fault tolerance, the focus will be on the syndrome extraction circuit of \autoref{ft_sc_circuit}, where the $2-2$ split is the natural choice, and we therefore restrict our analysis to the $2-2$ split also for the circuit of \autoref{naive_sc_circuit}.}. When the auxiliary qubit crosses the seam, it introduces a distribution edge (indicated by a blue dashed line), which is not yet implementable.
Therefore, in the second step we transform the auxiliary initialization to a measurement, resulting in a circuit that is symmetric when reflected across the QPU-separation (up to a reordering of the lines representing on-chip data qubits), and where the distribution edge becomes a (possibly improved) Bell pair.
\begin{equation}\label{naive_dist_circuit}
    \tikzfig[scale=0.4]{04_FT_dist_stab/w4_simple_dist}
\end{equation}

For the hook-free circuit in \autoref{ft_sc_circuit}, we split the circuit such that each QPU carries two data qubits and one auxiliary qubit each. The two-qubit measurements between the auxiliary qubits are no longer local and become distribution edges. We can further simplify the circuit by merging the two distribution edges into a single edge with a higher level of fault improvement:
\begin{equation}\label{ft_dist_circuit}
    \tikzfig[scale=0.4]{04_FT_dist_stab/w4_square_dist}
\end{equation}
This circuit is again symmetric under reflection across the QPU-separation.

To construct a fault-tolerant distributed circuit, we finally need to add enough fault improvement $d_X, d_Y, d_Z$  to compensate for the interconnect noise.
Given that the decomposition of \autoref{naive_sc_circuit} is not fault tolerant, while the one of \autoref{ft_sc_circuit} is,
it will come as no surprise that a fault-tolerant distributed circuit can only be obtained from the latter, but it is instructive to consider precisely where fault tolerance breaks down for the distributed version of the former. To this end, we analyze the implications of the fault weights on the distribution edges.
\autoref{w4_bell_effects} below shows the isolated effect of the distribution edges under edge-flip noise by idealizing all on-chip edges (indicated in pink). The resulting faults are shown as fault gadgets:
\begin{equation}\label{w4_bell_effects}
    \tikzfig{04_FT_dist_stab/w4_plaquette3}
\end{equation}

We observe that for both circuits the same three types of resulting faults are triggered by faults on the distribution edges:
\begin{enumerate} \itemsep0em
    \item  $Z$ faults on the distribution edges lead to hook errors on one of two QPUs. This behavior matches the hook errors of the simple circuit in the monolithic setting.
    \item  $X$ faults on the distribution edges lead to readout errors of the plaquette measurement.
    \item $Y$ faults on the distribution edges lead to a correlated hook and readout error.
\end{enumerate}
Decreasing or increasing the improvements $d_X, d_Y, d_Z$ makes these three faults more or less likely to occur.
To construct a fault-tolerant extraction circuit under circuit-level noise, the weights of the three types of faults listed above must be greater than or equal to the lowest possible weights of an equivalent combination of atomic faults, which are as follows:
\begin{enumerate} \itemsep0em
    \item The combination of the two data qubit errors has a probability of $p^2$ to occur, requiring $\frac{d_Z}{\ebitnoise_Z}\geq 2$.
    \item Readout errors have a probability $p$, requiring $\frac{d_X}{\ebitnoise_X}\geq 1$.
    \item The effect of $Y$ errors matches the combination of a correlated measurement flip plus data qubit error, which is an atomic fault of probability $p$ under circuit-level noise, with a second data qubit error of probability $p$; therefore, overall requiring $\frac{d_Y}{\ebitnoise_Y}\geq 2$.
\end{enumerate}

The above considerations ensure that the fault-improvement required for the distribution edges alone does not break fault tolerance.
However, for a fault-tolerant syndrome measurement, these distribution edges must also be embedded in a circuit with local fault tolerance on each module.
The distributed simple extraction circuit is still prone to hook errors from on-chip faults, so while a fault-improvement of $\frac{d_Z}{\ebitnoise_Z}\geq 2$ would suppress hook errors from the distribution edge in the simple circuit derived from \autoref{naive_sc_circuit}, the overall circuit will remain non-fault-tolerant, and any improvement $\frac{d_Z}{\ebitnoise_Z}\geq 1$ will not increase the quality of the extraction circuit in terms of fault tolerance.

In contrast, the logical performance of the distributed version of the hook-free circuit derived from \autoref{ft_sc_circuit} is \emph{mainly held back by the distribution edge}. Hence, with sufficient fault-improvement $d_X,d_Y,d_Z$, the distributed version is also fault-tolerant:
\begin{equation}\label{dist_square_ft_goal}
    \tikzfig[scale=0.4]{04_FT_dist_stab/w4_square_dist2}
\end{equation}

We note that the tools developed in \autoref{sec:fault_improvement} allow us to fault-improve distribution edges in all three Pauli bases, but that they can neither be improved \emph{independently} nor \emph{arbitrarily}.
Due to local QPU noise, we can improve $w_Z =\frac{d_Z}{\ebitnoise_Z}$ ($w_X = \frac{d_X}{\ebitnoise_X}$) up to $w_X \leq 2$  ($w_Z \leq 2$) for a distribution edge of $XX$ type ($ZZ$ type) and the other weight is upper-bounded to $w_X \leq 1$  ($w_Z \leq 1$).
Additionally, we cannot independently protect against $Y$ errors but only indirectly improve $w_Y=\frac{d_X-1}{\ebitnoise} + \min(1, \frac{d_Z}{\ebitnoise})$ through our choice of $d_X$ and $d_Z$ (see \autoref{prop:repeat_circuit}). For asymmetric noise with $\ebitnoise_Y \gg \ebitnoise_X, \ebitnoise_Z$, an overprotection against $X$ or $Z$ errors might be necessary to sufficiently improve against $Y$ errors, e.g., choosing $d_Z=\ebitnoise_Y$ to improve $w_Y\ge1$. On the other hand, in a symmetric noise setting we overprotect against $Y$ errors with the repetition strategy. For distillation based improvement, this depends on the chosen code. Altogether, we see that despite the constraints on the amount of possible improvement,
we can {always} construct an implementable distributed extraction circuit that reaches the fault weights of \autoref{dist_square_ft_goal} and is therefore fault-tolerant, with the overhead depending on the noise levels.

As a final example, \autoref{dist_square_ft_practice} below shows how the required improvement for the fault-tolerant extraction circuit of \autoref{dist_square_ft_goal} can be achieved under symmetric interconnect noise $\ebitnoise_X = \ebitnoise_Y = \ebitnoise_Z=\ebitnoise$ using the repetition of \autoref{sec:Repetition} with $d_X=\ebitnoise$ and $d_Z=2\ebitnoise$:
\begin{equation}\label{dist_square_ft_practice}
    \tikzfig[scale=0.4]{04_FT_dist_stab/w4_square_dist3}
\end{equation}
We see that we reach the desired protection against $X$ and $Z$ errors, and also sufficiently improve against $Y$ errors, for any $\ebitnoise$. The last fault equivalence of \autoref{dist_square_ft_practice} follows from the fact that, for both diagrams, the distribution edge's fault weight matches or surpasses the weight of their corresponding atomic faults.

Considering improvements $d_X,d_Y,d_Z$ such that $d_Y \geq \min(d_X,d_Z)$, as in the repetition based strategy, we will use the notation $X(d_X,d_Z)$ to denote the improvements of $X$ stabilizers, and $Z(d_X,d_Z)$ to denote the improvements of $Z$ stabilizers. Assuming again symmetric interconnect noise $\ebitnoise$, the context-free effects of different strategies are summarized in \autoref{tb:plaq_impl} for the weight-4 $Z$ stabilizer. The $X$ stabilizer follows by swapping all $X\leftrightarrow Z$.

\begin{table}[h]
    \centering
    \begingroup
    \renewcommand{\arraystretch}{1.2}
    \setlength{\tabcolsep}{10pt}
    \begin{tabular}{l cc}

        \multirow{2}*{\textbf{Strategy}} & \multicolumn{2}{c}{\bf Error Probability}                               \\ \cmidrule(lr){2-3}
                                         & \textbf{W-2 Hook}                         & \textbf{Readout}            \\
        \midrule
        $Z({\ebitnoise,2\ebitnoise})$    & $p^2$                                     & $p$                         \\
        $Z({\ebitnoise,\ebitnoise})$     & $\mathbf{p}$                              & $p$                         \\
        $Z({1,\ebitnoise})$              & $\mathbf{p}$                              & $\mathbf{p^{1/\ebitnoise}}$ \\
        $Z({1,1})$                       & $\mathbf{p^{1/\ebitnoise}}$               & $\mathbf{p^{1/\ebitnoise}}$ \\
        \bottomrule
    \end{tabular}

    \endgroup

    ~\\

    \caption{The context-free effects of different choices of improvement strategies in distributed weight-4 stabilizer measurements under Bell pair noise $(w,w,w)$ with $w=\frac{1}{\ebitnoise}$. Reducing the amount of improvement increases the probability of readout and weight-2 hook errors. Hook errors always occur parallel to the distribution seam. With repetition-based improvement for a $Z$ stabilizer, $d_Z$ is determined by the number of outer repetitions and $d_X$ by the inner, and vice versa for an $X$ stabilizer.
        Error probabilities that break fault-tolerance are indicated in bold.
    }
    \label{tb:plaq_impl}
\end{table}

\subsection{Weight-6 Stabilizer Measurements}\label{subsec:Weight-6-stabilizers}

The approach presented in the weight-$4$ stabilizer measurement case can be similarly adapted to weight-$6$ stabilizers.
Using the fault-equivalent rewrite of a $6$-legged spider introduced in~\cite{rodatzFaultToleranceConstruction2025}, we obtain the following fault-tolerant ZX-diagram, which we will use as an initial point for distributed circuits:
\begin{equation}\label{w6_init}
    \tikzfig[scale=0.4]{04_FT_dist_stab/w6_circif_mono}
\end{equation}
We first propose a possible $3$--$3$ splitting and bring it into an implementable form by choosing three two-qubit measurements as distribution edges and utilizing fault-equivalent rewrites:
\begin{equation}
    \tikzfig[scale=0.4]{04_FT_dist_stab/w6_33_dist}
\end{equation}
for $k_1 \oplus k_2 \oplus k_3 = k$.

The circuit requires two auxiliary qubits per QPU. To find the required fault improvement on the distribution edges, and to understand the implications of \emph{insufficient} improvement, we need to consider which of the detecting regions from the initial rewrite in \autoref{w6_init} each distribution edge is part of. In Pauli web notation~\parencite{bombinUnifyingFlavorsFault2024}, the two relevant independent detecting regions are as follows:
\begin{equation}\label{w6Pauliwebs}
    \tikzfig[scale=0.45]{04_FT_dist_stab/w6_33_regions}
\end{equation}
The middle distribution edge is part of both detecting regions, and a $Z$ fault on this edge is undetected if there is also an independent (and inequivalent) second $Z$ fault in each detecting region, either from on-chip faults or from faults on the other distribution edges. (We do not count the combination of two $Z$ faults on the middle edge, as it is trivial.)
In the case where the distribution edges have not been sufficiently fault-improved, the most likely scenario is a $Z$ fault on all three distribution edges, which would propagate to a weight-$3$ hook error on one of the QPUs.
Additionally, the left and right distribution edges are part of only one of the two detecting regions. If a $Z$ fault on the distribution edge occurs together with another local $Z$ error that is also only part of the same detecting region, the errors propagate at most to a weight-$2$ hook error on one of the QPUs.
Similar to the weight-$4$ case, $X$ faults on the distribution edges translate to readout errors. In total, we see that each distribution edge must be improved to $w_X, w_Y, w_Z \geq 1$ to achieve fault-tolerance.
\begin{equation}\label{w6_33_alt}
    \tikzfig[scale=0.45]{04_FT_dist_stab/w6_33_regions_other}
\end{equation}

Alternative implementations of a $3$--$3$ splitting can be achieved, as shown in \autoref{w6_33_alt}, by shifting the distribution seam closer to one QPU than to the other. However, this means that the distribution edges all sit within the same detecting region or even outside of detecting regions, which reduces the ability to distinguish faults occurring on them from one another, and also increases the number of auxiliary qubits that need to fit on one of the QPUs. For these reasons we only consider the symmetric split.

Another possible distribution of interest for the weight-$6$ measurement is a $4$--$2$ split. Starting from the decomposition of the $6$-legged spider in \autoref{w6_init}, multiple choices of distribution edges realize the desired splitting. \autoref{w6_42_dist} shows a variant that utilizes two distribution edges:
\begin{equation}\label{w6_42_dist}
    \tikzfig[scale=0.4]{04_FT_dist_stab/w6_circif_42_v2}
\end{equation}
for $k_1 \oplus k_2 = k$.
The circuit requires two auxiliary qubits on the QPU that hold four data qubits and one auxiliary qubit on the QPU that holds two data qubits.
The distribution edges are part of both detecting regions:
\begin{equation}\label{w6webs_42}
    \tikzfig[scale=0.45]{04_FT_dist_stab/w6_42_regions}
\end{equation}
Consequently, simultaneous $Z$ errors on both edges stay undetected and propagate to a weight-2 hook error on the QPU that holds only two data qubits.
Once again, each distribution edge must be improved to $w_X, w_Y, w_Z \geq 1$ to achieve fault-tolerance.

With improvements such that $d_Y \geq \min(d_X,d_Z)$, and with symmetric noise $\ebitnoise$, the context-free effects of different strategies for weight-6 stabilizer measurements are shown in \autoref{tb:plaq_impl_wt6}.

\begin{table}[h]
    \centering
    \begingroup
    \renewcommand{\arraystretch}{1.2}
    \setlength{\tabcolsep}{10pt}

    \begin{tabular}{ll ccc}

        \multirow{2}*{\textbf{Split}} & \multirow{2}*{\textbf{Strategy}}                     & \multicolumn{3}{c}{\bf Error Probability}                                                             \\ \cmidrule(lr){3-5}
                                      &                                                      & \textbf{W-2 Hook}                         & \textbf{W-3 Hook}           & \textbf{Readout}            \\
        \toprule
        \multirow{4}*{3-3}            & $Z({\ebitnoise,\ebitnoise, \ebitnoise, \ebitnoise})$ & $p^2$                                     & $p^3$                       & $p$                         \\
                                      & $Z({\ebitnoise,1, \ebitnoise, \ebitnoise})$          & $p^2$                                     & $\mathbf{p^2}$              & $p$                         \\
                                      & $Z({\ebitnoise,1, \ebitnoise, 1})$                   & $p^2$                                     & $\mathbf{p^2}$              & $\mathbf{p^{1/\ebitnoise}}$ \\
                                      & $Z({1,1, 1, 1})$                                     & $\mathbf{p^{1 + 1/\ebitnoise}}$           & $\mathbf{p^{3/\ebitnoise}}$ & $\mathbf{p^{1/\ebitnoise}}$ \\
        \midrule
        \multirow{3}*{4-2}            & $Z({\ebitnoise, \ebitnoise, \ebitnoise})$            & $p^2$                                     & $p^3$                       & $p$                         \\
                                      & $Z({\ebitnoise, \ebitnoise, 1})$                     & $p^2$                                     & $p^3$                       & $\mathbf{p^{1/\ebitnoise}}$ \\
                                      & $Z({1,1,1})$                                         & $\mathbf{p^{2/m}}$                        & $\mathbf{p^{1+2/m}}$        & $\mathbf{p^{1/\ebitnoise}}$ \\
        \bottomrule
    \end{tabular}
    \endgroup

    \caption{The context-free effects of different choices of improvement strategies in distributed weight-6 stabilizer measurements under Bell pair noise $(w,w,w)$ with $w=\frac{1}{\ebitnoise}$. Reducing the amount of improvement increases the probability of hook and readout errors. Distributed hook errors always occur parallel to the distribution seam.
        Error probabilities that break fault-tolerance are indicated in bold.
    }
    \label{tb:plaq_impl_wt6}
\end{table}

\subsection{Adaptability to Resource Constraints and Repeat-Until-Success}\label{section4_3}

We have seen that the above circuits for distributed weight-$4$ and weight-$6$ stabilizer measurements are adaptable to varying noise levels of physical Bell pairs, and that the degree of edge improvement can be adjusted to achieve fault tolerance. However, the noisiness of Bell pairs between modules is not the only hardware constraint.
The physical Bell pairs are consumed during a distributed computation and must constantly be created and shared between the QPUs. The \emph{generation rate} of Bell pairs restricts how fast distributed computations can proceed, as do the number of
\emph{interconnect channels} between QPUs and the number of availeble \emph{communication qubits} i.e., specialized qubits for transduction.
Another constraint is space and qubit connectivity, limiting how many auxiliary qubits can be used per stabilizer and how they can interact. This consideration is especially important for hardware with fixed qubit layout such as superconducting hardware.

If the Bell pair rate is limited, one solution is to idle ahead of each round of stabilizer measurement, until all Bell pairs that will be used in the circuit have been generated and transferred onto auxiliary qubits. However, this requires sufficient space for these ``storage'' auxiliary qubits. Another option, which is more space efficient, is to spread the distribution edges maximally in time, transferring Bell pairs only as needed. This may still require idling between each distribution edge, but requires fewer auxiliary qubits.

Conversely, if the Bell pair rate is high and there is sufficient space, reducing circuit depth by maximizing parallelization reduces the effects of idle noise and the total runtime. It is clear that the desired properties of a distributed syndrome extraction circuit depend strongly on the hardware. We now show one of the major advantages of deriving circuits via fault-equivalent rewrites:
the \emph{flexibility} to optimally utilize the underlying quantum hardware while considering its resource constraints.

In the discussion below, we focus on the weight-4 stabilizer measurement. The discussion extends naturally to the weight-6 stabilizer measurement or, more generally, to weight-$N$ cases.
We begin by analyzing distribution edge improvement via repetition:
\begin{equation}\label{w4_seq_rep}
    \tikzfig[scale=0.5]{04_FT_dist_stab/w4_sequential_repeat}
\end{equation}
The circuit is efficient in terms of space requirements for the underlying distributed hardware:
It requires only one Bell pair at a time and a single auxiliary qubit per QPU.
In early fault tolerance where qubits are scarce, this gives more room to implement larger code distances or allows for more computational workspace.

On the flip side, it has two major disadvantages.
First, because it does not allow parallel processing of Bell pairs, the syndrome extraction has a slow execution time $T$.
Let us denote the time to generate a Bell pair as $\tau_\text{bell}$, and $\tau_\text{gate}$ as the time to execute a gate in the native gate set.
Since all Bell pairs are processed and generated sequentially, the overall execution time scales as $T=\mathcal{O}(\dout \cdot \din)\tau_\text{bell}$ + $\mathcal{O}(\dout \cdot \din)\tau_\text{gate}$, given the improvements $\dout$ and $\din$.
While this circuit, which we will denote from now on as the \emph{\vone} approach, can operate under harsh space constraints, it suffers from a long runtime.
The second disadvantage is that data interaction and Bell pair interaction are interleaved as a single process: the sequential Bell pair processing is wrapped in between an initial and final data interaction step.
This makes it impossible to preprocess the Bell pairs separately. Such separate processing can reduce the overall runtime, since some of the operations involved can be scheduled in parallel with other operations in the larger circuit.
Separate processing also makes \emph{repeat-until-success} schemes possible: here, instead of using the detectors present in the fault-improved edge for decoding, the outcome of the subcircuit that processes the Bell pairs is discarded whenever a fault is detected, and new Bell pairs are generated and processed until all detectors return a trivial measurement outcome. In the integrated circuit, this strategy would not be possible, as the Bell pairs are already entangled with the first data qubits by the time the detector outcomes are obtained. The repeat-until-success strategy simplifies the overall decoding problem when distributed stabilizers are used within a larger code, as discussed further in \autoref{sec:integrated_decoding}.
This way, the choice of circuit can also depend on the hardness of the decoding problem and on the \emph{classical} resources available for decoding.

Based on these observations, we leverage the strength of ZX calculus to reorder the above distributed weight-4 stabilizer measurement and divide the measurement circuit into two steps. We achieve this by bending the diagram at the distribution edge so that we first create the improved distributed connection, and only afterward interact with the data qubits:
\begin{equation}\label{RUS_hl}
    \tikzfig[scale=0.5]{04_FT_dist_stab/RUS}
\end{equation}

We now show how to transform the \vone circuit to a highly parallel circuit using such a bending. This begins with the rearrangement of the inner nested improvements:
\begin{equation}\label{RUS_inner}
    \tikzfig[scale=0.4]{04_FT_dist_stab/RUS_parallel_inner}
\end{equation}
In the first step of the above equation, we bend open the diagram along the distribution seam.
All Bell pairs are now generated and processed at the same time on individual auxiliary qubit pairs.
The two auxiliary qubits of the \vone approach are transformed into columns of three-legged spiders. These columns are not yet implementable, requiring the second step where
we keep every other pair of three-legged spiders as two-body measurements on the Bell pairs and then connect these two-body measurements with the following local Bell measurements and also move the outgoing three-legged red spiders.

We can leverage the same idea for the outer improvements. To achieve parallel nested repetition, the outer repetitions are first similarly bent and interpreted as parallel two-body measurements, and the parallel inner improvement from \autoref{RUS_inner} is then used inside:
\begin{equation}\label{RUS_nested}
    \tikzfig[scale=0.35]{04_FT_dist_stab/RUS_parallel_nested}
\end{equation}
Notice how the CNOT gates interacting with the data qubits are now separated at the end of the circuit, making a repeat-until-success strategy possible.  Also note that although these final CNOT gates are shown sequentially to make the diagram more readable, they can both be scheduled in parallel. If a repeat-until-success scheme is not applied, all measurements can also be scheduled in parallel.
In this way, the circuit achieves a significantly reduced runtime compared to the sequential circuit. It completely parallelizes the generation and processing of Bell pairs, achieving a constant scaling of the execution time with $\din$ and $\dout$: $T=\mathcal{O}(1)\tau_\text{bell}$ + $\mathcal{O}(1)\tau_\text{gate}$. However, it requires at least $\din \dout$ auxiliary qubits per QPU onto which Bell pairs can be transferred, plus additional auxiliary qubits to mediate the required two-qubit measurements if these are not available as native gates.

The circuits in \autoref{w4_seq_rep} and \autoref{RUS_nested} display the extremes of full sequential and full parallel processing.
We can interpolate between these two approaches to gain different trade-offs between space and time, which allows for more fine-grained optimization depending on available resources in terms of Bell pair generation and transmission, space for auxiliary qubits, and computation time.

One helpful tool for interpolating between levels of parallelization is local Bell state preparations and Bell measurements on each QPU, which can easily be constructed using CNOTs and resets or single-qubit measurements.
For instance, we can show that a repeat-until-success-friendly structure can also be achieved while maintaining sequential Bell pair processing:
\begin{equation}\label{RUS_sequential}
    \tikzfig[scale=0.38]{04_FT_dist_stab/RUS_sequential2}
\end{equation}
By employing local Bell states, we replace each local auxiliary qubit with two entangled auxiliary qubits.
The initial data interactions can now be moved across the bend representing the Bell state preparation, so that they can be postponed until after the fault-improved distribution edge is constructed.
We could also move half of the outer $XX$ type distribution edges across the bend, allowing two outer distribution edges at a time to be processed in parallel.

As a second and last example of the use of local Bell pairs, we show how more nuanced adaptations can be made by constructing a slightly different parallel variant of the inner improvements, inspired by the Shor code encoder:
\begin{equation}\label{v2}
    \tikzfig[scale=0.38]{04_FT_dist_stab/inner_parallel_martin}
\end{equation}
As with \autoref{RUS_inner}, this approach also requires enough auxiliary qubits to host the Bell pairs of each inner repetition all at once, but the above strategy processes the Bell pairs sequentially through a chain of CNOT gates. While this approach slightly increases execution time compared to the approach in \autoref{RUS_inner}, it avoids the need for any additional auxiliary qubits to implement local two-body measurements.
We denote the combination of this parallel variant of inner improvement with sequential outer improvement as the \emph{\vtwo} approach. In what follows, we use it as a representative middle ground between the \vone and completely parallel approaches, as in \autoref{RUS_nested} or the \ed approach.
The partial parallelization is reflected by the scaling in execution time, which falls between the extremal approaches:
The \ed approach processes all Bell pairs in parallel, followed by parallel $CNOT$ chains of first the inner and then the outer improvements. Therefore, the overall extraction time scales in $T=\mathcal{O}(1)\tau_\text{bell}$ + $\mathcal{O}(\dout+\din)\tau_\text{gate}$. In contrast, the extraction time of the \vtwo approach, which only processes inner Bell pairs in parallel, scales in $T=\mathcal{O}(\dout)\tau_\text{bell}$ + $\mathcal{O}(\dout \cdot \din)\tau_\text{gate}$.

\section{Integrating Distributed Stabilizers Within a Larger Code}\label{sec:Context_aware}

Having derived distributed stabilizer measurement circuits in isolation, we next turn to the overarching goal: distributed fault-tolerant computation.
When distributed stabilizers are integrated into a larger context, two important questions arise. First, which optimizations can be made based on the larger context? Second, how should decoding be performed?

\subsection{Context-aware Fault Tolerance}

Up until this point, we have considered context-free fault tolerance: ZX diagrams that are fault-equivalent to the corresponding  idealised ZX diagram:
\[ \tikzfig[scale=0.5]{045_context_aware/def-fault-tolerance} \]

Syndrome extraction using such fault-tolerant measurement circuits always preserves the code distance $d$ of a given stabilizer code, but may not do so in the most resource-efficient manner.

In a global context, the fault-tolerance requirements for local circuits can often be relaxed. Generally, it is not trivial to determine which relaxations of local fault tolerance are allowed by the global context. An advantage of the relational nature of fault equivalence is that this task only needs to be performed once. In particular, this means that if a relaxation has already been found in the monolithic setting, the corresponding circuit can serve as a starting point for further rewrites to adapt it to the distributed setting. Alternatively, correlated faults known to not reduce the circuit distance can be added onto an idealised specification via fault gadgets, ahead of applying fault-equivalent rewrites to find a distributed implementation.

We here detail how context informs the treatment of two classes of faults,  \emph{hook errors} and \emph{timelike errors}, both of which can occur due to Bell pair errors in the distributed setting, and both of which, in certain contexts, do not require reaching on-chip noise levels. By hook errors, we mean combinations of atomic faults that propagate to \emph{a larger number} of faults on the data qubits.  By timelike errors we mean combinations of atomic faults that change a stabilizer measurement outcome. (A more general usage of the term ``hook error'' would include fault combinations that propagate to a larger number of data errors \emph{and} readout errors, we instead refer to these as combined hook errors and readout errors.) When using the repetition-based edge-improvement of \autoref{prop:repeat_circuit}
within the syndrome extraction circuits for $Z$ type stabilizers shown in \autoref{sec:FT_dist_stab_meas},
outer $XX$ type repetitions catch $Z$ faults that lead to hook errors, while inner $ZZ$ type repetitions catch $X$ faults that lead to timelike errors. This allows for separate relaxations in the protection of the two error classes.

\subsubsection{Hook errors}

While the stabilizer-measurement circuits of \autoref{sec:FT_dist_stab_meas} protect against all hook errors, not all hook errors reduce the circuit distance of a given larger code that the measurement is embedded in.
We call a set of hook error \emph{benign} if it does not affect the circuit distance, and \emph{malign} if it reduces the circuit distance.

For hook errors to be benign within a code of distance $d$, they must not overlap with any logical representative of the code in such a manner that the corresponding logical may be created by a fault of weight less than $d$. For individual weight-$2$ hook errors, for example, they must not align with the minimum weight logical representatives. For higher-weight hook errors and sets of hook errors, larger weight representatives must also be considered.
This criterion is captured by the notion of \emph{residual distance} in~\cite{strikis2026}. In the rotated surface code, the minimum weight logicals follow a straightforward pattern and the hook errors from the CNOT based syndrome extraction circuit in \autoref{naive_sc_circuit}  have at most weight two (up to stabilizer equivalence), making it easy to determine which hook errors are benign. In other settings, such as high-rate qLDPC codes, merely finding the set of minimum weight logical representatives is not easy. As an illustrative example, in Section $4.2$ of~\cite{beverland2025a}, the authors note that testing all weight-$12$ $X$-operators in the gross code that could be caused by bit-flip noise, to see whether they form logical operators, would require evaluating approximately $10^{17}$ cases.

A proposed automated approach to circuit design is to map the identification of minimum weight logicals to a MaxSAT problem formulation on the decoding graph and then iteratively improve extraction circuits to increase the residual distance~\cite{gidneyStimFastStabilizer2021,viszlaiProphuntAutomatedOptimization2026}. However, due to the NP-hardness of MaxSAT, this approach has limited capability.
Another approach is the left-right circuits of~\cite{strikis2026}, which creates efficient extraction circuits using graph-coloring and evaluates their {residual distance}.

\subsubsection{Timelike errors}

Faults that only flip the stabilizer measurement outcome form an interesting subclass of errors, as the effect of such errors on the logical performance depends strongly on the context. In many settings, there is a baseline of stabilizer measurement outcomes to compare to, so that even a long sequence of stabilizer measurement errors can be detected. The clearest example is a memory experiment in which the logical qubit is initialized and measured in a given Pauli basis. In this case, the only timelike boundaries are the initialization and the final readout. The only timelike error strings that can terminate there are those that commute with the logical operator, meaning that all timelike errors are benign.

Even in more general scenarios where the initial or final values of the logical operator cannot be used, the surrounding past and future stabilizer outcomes can provide a baseline for detecting erroneous stabilizer outcomes. In the setting of CSS codes under transversal operations~\cite{cainCorrelatedDecodingLogical2024, zhouLowoverheadTransversalFault2025} have shown that it suffices to perform $\mathcal{O}(1)$ rounds of syndrome extraction between logical gates to achieve protection against timelike errors, a result that relies on the fact that the resulting syndrome outcomes can be compared to baselines via correlated decoding.

The setting of lattice surgery differs from that of a quantum memory since it requires adding to the instantaneous stabilizer group a \emph{new} stabilizer: the one corresponding to the joint parity measurement of the logical qubits. There is no baseline to compare this new stabilizer to. Stated differently, timelike boundaries are added onto which logical error strings can terminate. Unlike in the case of a memory experiment, these strings can affect the logical information, and the new timelike boundaries must be sufficiently far apart that no such strings have a weight less than $d$. An intuitive sketch of how lattice surgery differs from the memory setting is shown in \autoref{fig:pipes}. \cite{cainCorrelatedDecodingLogical2024, zhouLowoverheadTransversalFault2025} also recover this result, explicitly showing how the logical performance is affected if the timelike boundaries during lattice surgery are too close.

\begin{figure}
  \hspace{.7cm}
  \begin{subfigure}[t]{0.3\textwidth}
    \centering
    \includegraphics[height=3cm]{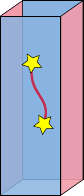}
    \caption{}
    \label{fig:sc_mono_mem}
  \end{subfigure}
  \hspace{3cm}
  \begin{subfigure}[t]{0.3\textwidth}
    \centering
    \includegraphics[height=3cm]{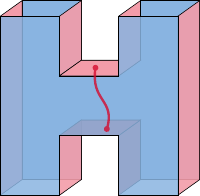}
    \caption{}
    \label{fig:sc_mono_mem}
  \end{subfigure}
  \caption{Examples of timelike errors in the surface code, illustrated in the style of pipe diagrams: (a) A timelike error string that is detectable, with the stars representing detection events due to a mismatch in sequential stabilizer outcomes. (b) Joint logical measurement, with a timelike error string that is undetectable due to terminating on the timelike boundaries.}\label{fig:pipes}
\end{figure}

\subsection{Integrated Decoding}\label{sec:integrated_decoding}

When integrating distributed stabilizers within a larger code, the decoding strategy becomes an important question. Here, treating the entanglement distillation as an integrated part of the stabilizer measurements in the larger code has a distinct advantage over treating it as a separate process: integrated decoding can significantly reduce the resource requirements compared to a separate, two-step decoding process.

To see why, it is illustrative to consider decoding in the setting of concatenated CSS codes and to compare the distance achieved by a two-step decoder with the true code distance.\footnote{Discussion inspired by the stack exchange thread \url{https://quantumcomputing.stackexchange.com/questions/34334/distance-of-the-concatenated-quantum-error-correcting-code}.} Consider an inner code with distance $d_i$ and an outer code with distance $d_o$. A codeword in the concatenated code has weight $d_i d_o$,
requiring $\left\lfloor \frac{d_i d_o +1}{2} \right\rfloor$
faults for an optimal decoder to fail. However, a two-step decoder will only reach half the distance: a failure in the outer decoder can be induced by $\left\lfloor \frac{d_o+1}{2} \right\rfloor$ failures of the inner decoders, each of which can be induced by $\left\lfloor \frac{d_i+1}{2} \right\rfloor$ faults.

Similarly, the required distance $\dBell$ can be larger when a stabilizer code used for distillation is decoded separately than in an integrated setup. In what follows, we consider, for simplicity, the same improvement in the $X,Y,Z$ bases.
In \autoref{sec:fault_improvement}, distribution edges constructed out of ebits with weight $w=\frac{1}{\ebitnoise}$ are improved up to weight $w=1$ through codes of distance $\dBell=\ebitnoise$. The distribution edges are used to rewrite monolithic syndrome extraction circuits into distributed ones.
As the distributed circuits are \emph{fault-equivalent} to their monolithic counterparts, the chosen distance $\dBell=\ebitnoise$ is sufficient in order for the distributed code to reach the same circuit distance as its monolithic counterpart under optimal decoding. Meanwhile, if a stabilizer code is used for distillation and decoded by a separate decoder, protection against Bell pair noise of weight $w=\frac{1}{\ebitnoise}$ requires twice the distance in order for the logical Bell pair noise to reach weight 1, as the probability for the separate decoder failure scales as $p^{w \left\lfloor \frac{\dBell+1}{2}\right\rfloor}$.

The above reasoning shows that \emph{if} an efficient integrated decoder that reaches full distance can be constructed, it can lead to lower resource requirements. However, constructing such a decoder may be nontrivial, whereas constructing a two-step decoder is straightforward under the assumption that both the distillation code and the larger code have efficient decoders. One possibility that falls between these two options is that an efficient integrated decoder can be constructed that achieves a higher distance than a two-step one, but not the full distance. Such decoding can also lead to lower resource requirements than the two-step process, even if it requires $\dBell > \ebitnoise$ to compensate for the loss of distance during decoding.

We note that in the case of a larger, matchable code, such as the surface code, restricting to \emph{matchability-preserving} rewrites in the derivation of the distributed stabilizers guarantees that the decoding problem remains matchable and that an efficient integrated decoder, capable of reaching full distance, can be constructed.
Such rewrites are described in~\cite{schweikart2026}.
As an example, the ``square rewrite'' of \autoref{ft_sc_circuit}, which forms the basis for distributed weight-4 plaquette measurements, is CSS matchability preserving, as are the insertion of a phase-free spider and spider-unfuse, so that the final diagram on the right-hand side of \autoref{ft_sc_circuit} is a matchability-preserving rewrite of the left-hand side.

Finally, we note that in the repeat-until-success setting, $\dBell$ atomic faults on the interconnects are required in order for the postselection to fail to reject a nontrivial fault.
Therefore, just like for integrated decoding $\dBell = \ebitnoise$ is sufficient for each repetition in order to obtain $w=1$ on the accepted improved distribution edge.
Since repeat-until-success requires an average number of repetitions larger than one (assuming nonzero noise), this means that integrated decoding still reduces the required number of Bell pairs more than repeat-until-success. However, when the average number of repetitions is not very large, repeat-until-success will still significantly reduce the required number of Bell pairs compared to two-step decoding, while only requiring decoding for the larger code. When the integrated decoding problem is too hard, this strategy is therefore an interesting alternative.

\FloatBarrier
\section{Distributed Surface Code}\label{sec:surface_code}

In this section, we demonstrate how distributed stabilizer measurements can be adapted to the context of the rotated surface code.
With context-aware fault tolerance, we can reduce the overhead compared to the weight-4 stabilizer measurements shown in \autoref{subsec:Weight-4-stabilizers}.

\begin{figure}
  \begin{subfigure}[t]{0.3\textwidth}
    \centering
    \includegraphics[height=4cm]{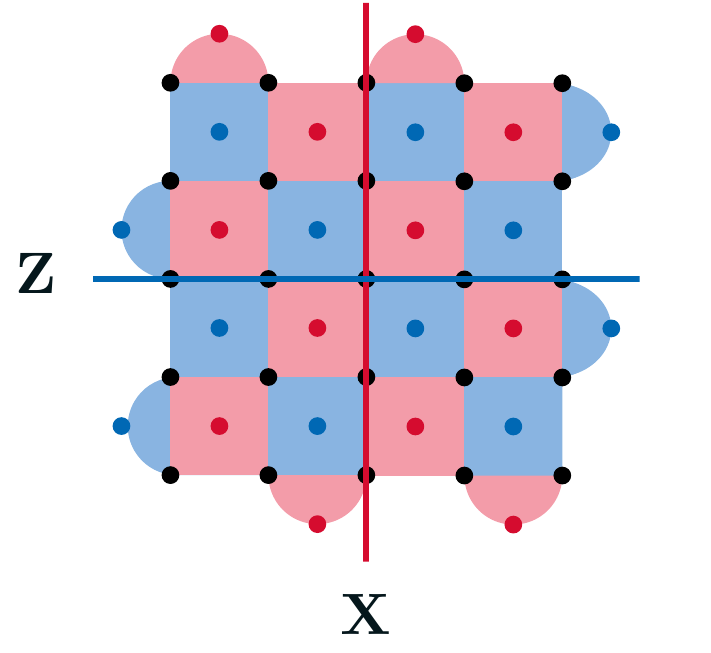}
    \caption{}
    \label{fig:sc_mono_mem}
  \end{subfigure}
  \hfill
  \begin{subfigure}[t]{0.65\textwidth}
    \centering
    \includegraphics[height=4cm]{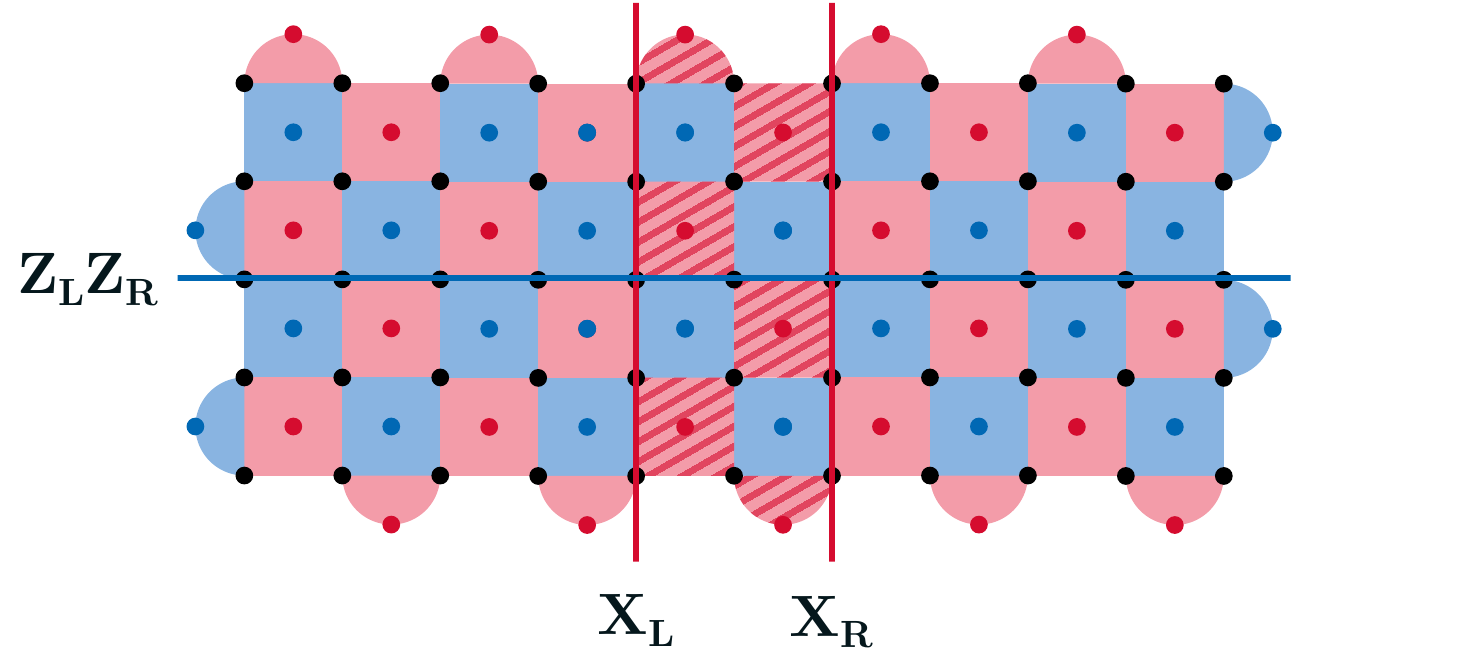}
    \caption{}
    \label{fig:sc_mono_surg}
  \end{subfigure}
  \caption{Rotated surface code. (a) A distance $d=5$ surface code patch, hosting one logical qubit. The red and blue plaquettes are the $X$ and $Z$ stabilizers, respectively. The black dots represent the data qubits, while the colored dots represent auxiliary qubits used for the stabilizer measurements. Logical $X$ forms a vertical line and logical $Z$ forms a horizontal line. (b) Surface code lattice surgery. Two patches with independent logical operators $X_L$ and $Z_L$ for the left patch and $X_R$ and $Z_R$ for the right patch are merged into a single patch.
    The added $X$ plaquettes (shaded) measure the joint $X_L X_R$ operator, classically correlating the individual logicals. By adding $X_L X_R$ to the stabilizer group, the $Z_L$ and $Z_R$ logicals are replaced by the joined $Z_L Z_R$ operator, which may create entanglement.}
  \label{fig:sc_mono}
\end{figure}

The rotated surface code is defined from a checkerboard tiling of weight-$4$ $X$ and $Z$ plaquettes (together with weight-$2$ boundary plaquettes), as shown in  \autoref{fig:sc_mono_mem}. Logical $X$ ($Z$) operators connect boundaries with $X$ ($Z$) boundary plaquettes.
Both in the context of a single distributed surface code patch and distributed lattice surgery between surface code patches hosted on separate QPUs, an important consideration  is how to choose the seam
that separates stabilizers between QPUs.
With seams that cut either vertically or horizontally through surface code patches, the number of distributed stabilizers can be minimized. However, this minimization can only be achieved with 2-2 splitting of the distributed stabilizers. With a 3-1 splitting, twice the number of stabilizers must be distributed, as shown in \autoref{fig:sc_splits}.
At the same time, the choice of splitting affects the types of hook errors that can be created, as discussed further in \autoref{sec:sc_hook_errors}.
There, we find that the difference in hook error behavior can be compensated for, so that the $2$-$2$ splitting remains the better option.
We will choose a vertical seam for both memory and lattice surgery, aligning it with the surgery boundaries in the latter case. This is shown in~\autoref{fig:sc_surg_comb}.

\begin{figure}[!htb]
  \hspace{.3cm}
  \begin{subfigure}[t]{0.4\textwidth}
    \centering
    \includegraphics[width=\textwidth]{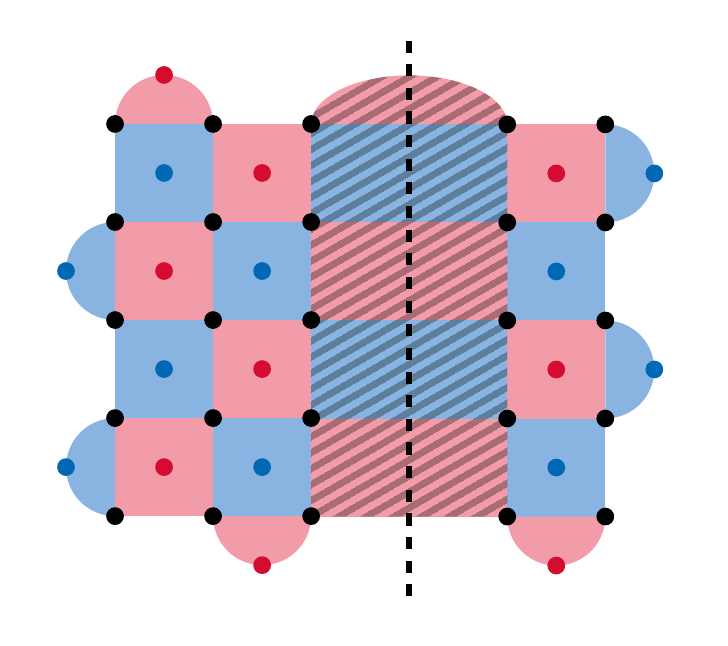}
    \caption{}
    \label{fig:sc_22}
  \end{subfigure}
  \hfill
  \begin{subfigure}[t]{0.4\textwidth}
    \centering
    \includegraphics[width=\textwidth]{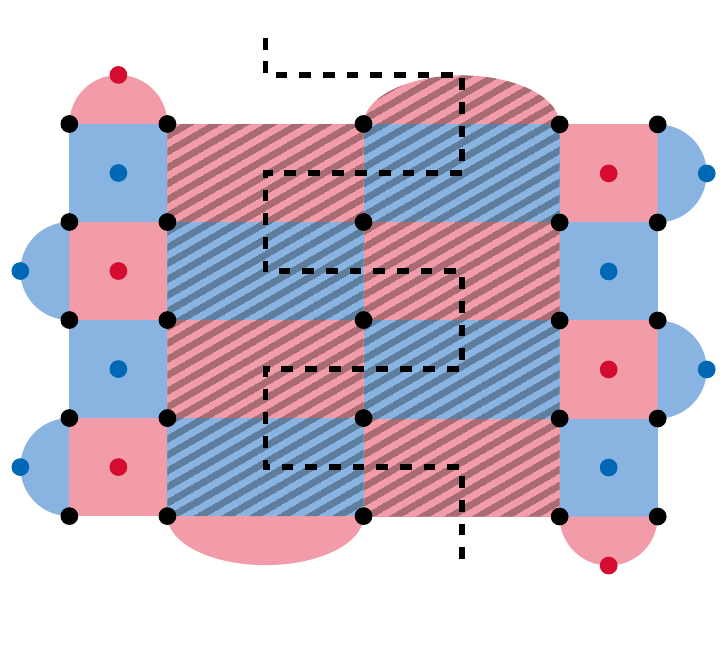}
    \caption{}
    \label{fig:sc_31}
  \end{subfigure}
  \hspace{.3cm}
  \caption{Distributed surface code memory, color conventions as in \autoref{fig:sc_mono}. The choice of how to split weight-4 plaquettes determines the distribution seam. (a) The 2-2 split introduces distributed hook errors but aligns well with the chessboard pattern of the code, requiring only a single column of plaquettes to be distributed. (b) The 3-1 split introduces only local hook errors but requires two plaquette columns to be distributed. The dashed line indicates the distribution seam. All shaded plaquettes operate between QPUs and must be distributed, possibly requiring multiple auxiliary and communication qubits.}
  \label{fig:sc_splits}
\end{figure}

\begin{figure}[!htb]
  \begin{subfigure}[t]{0.45\textwidth}
    \centering
    \includegraphics[width=\textwidth]{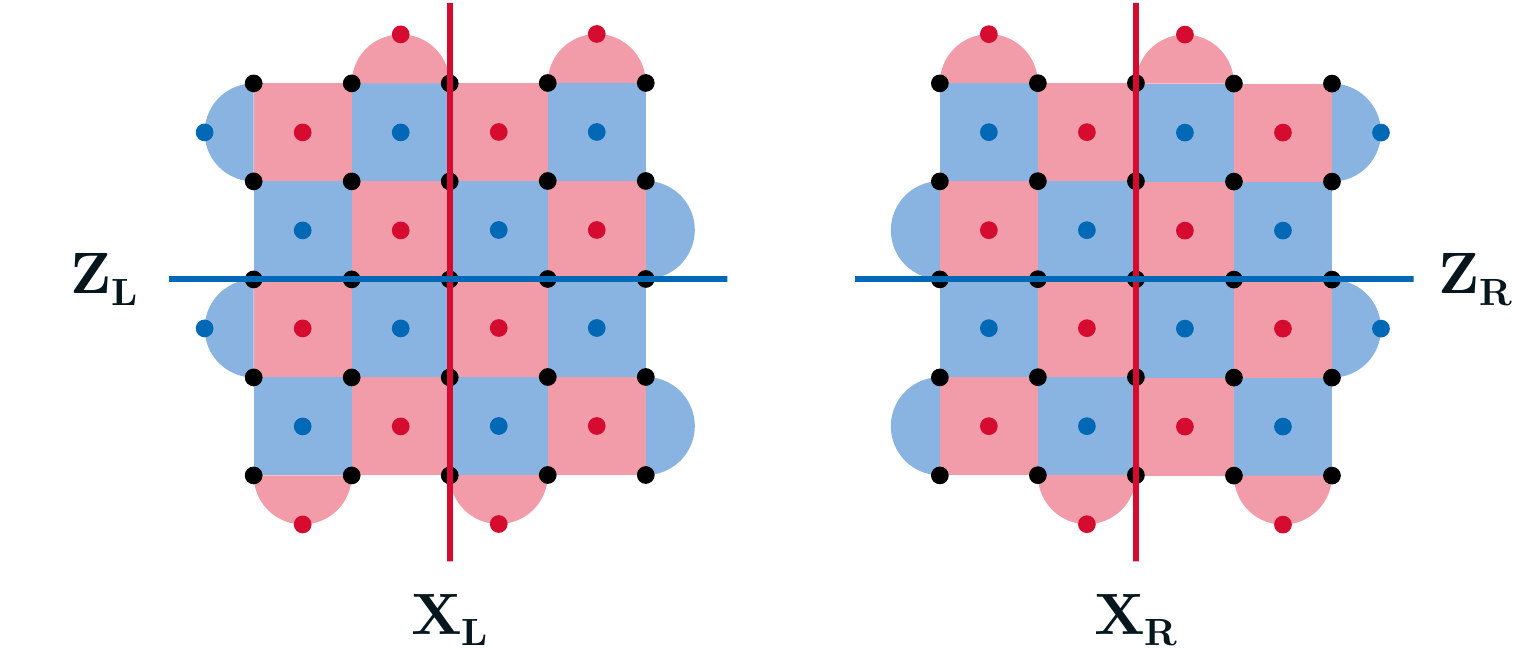}
    \caption{}
    \label{fig:sc_surg_pre}
  \end{subfigure}
  \hfill
  \begin{subfigure}[t]{0.45\textwidth}
    \centering
    \includegraphics[width=\textwidth]{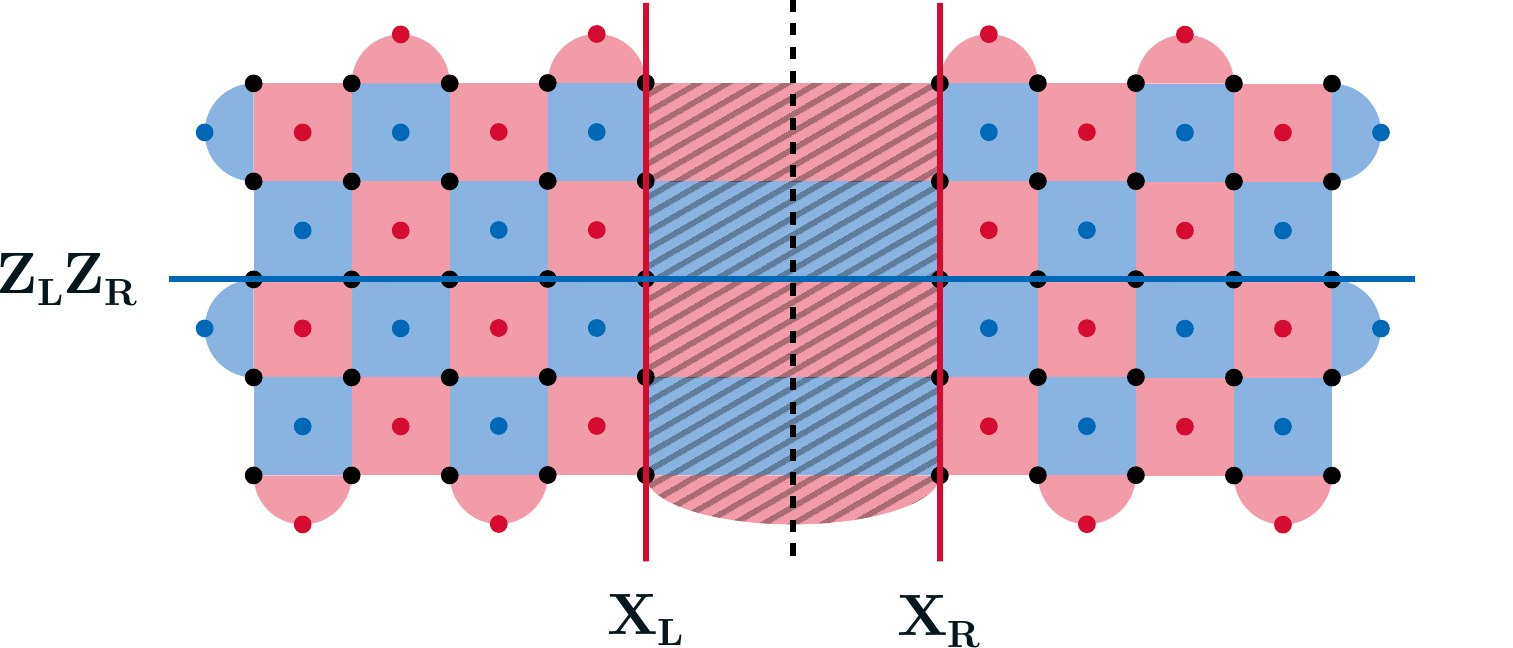}
    \caption{}
    \label{fig:sc_surg}
  \end{subfigure}
  \caption{Distributed lattice surgery: (a) Two $d=5$ surface code patches on separate QPUs. Color conventions as in \autoref{fig:sc_mono}. (b) $M_{XX}$ lattice surgery is performed by merging the two patches along $Z$ boundaries. The left surface code patch is mirrored to align plaquettes at the seam. The dashed line indicates the distribution seam. All shaded plaquettes operate between QPUs and must be distributed, possibly requiring multiple auxiliary and communication qubits.}
  \label{fig:sc_surg_comb}
\end{figure}

\FloatBarrier
\subsection{Syndrome Extraction and Benign Hook Errors in the Surface Code}\label{sec:sc_hook_errors}

For monolithic systems with a square-grid layout, the most commonly used surface-code syndrome-extraction circuit is the simple circuit introduced in \autoref{naive_sc_circuit}. In the context-free setting, this circuit is not fault-tolerant. Still, with a careful ordering of the CNOT gates that ensures hook errors are not aligned with minimum weight logicals, it can be made fault-tolerant within the larger context of the rotated surface code.

In the distributed setting, a $2$-$2$ split of the weight-$4$ stabilizer does not allow for the same CNOT ordering. Here, both $X$ and $Z$ type hook errors propagate to two data qubits on the same QPU, as shown in~\autoref{fig:sc_mem_dist}. This means that either $X$ or $Z$ type hook errors will be aligned with the corresponding logical operator, so that the circuit distance is halved.
Without additional protection against hook errors, this makes the 3-1 splitting more favorable, despite the increase in the number of distributed stabilizers. This splitting ensures both that the CNOT gates can be ordered to make hook errors benign, and that the location of atomic faults that generate hook errors does not coincide with the noisier part of the circuit~\cite{jacintoNetworkRequirementsDistributed2026, shalbyOptimizedNoiseresilientSurface2025}.

A strategic choice of fault improvement changes which splitting is the most favorable.
Given that only the hook errors from \emph{half} of the plaquettes in the 2-2 split are malign, doubling the number of distributed plaquettes through a 3-1 split is more expensive than protecting against malign hook errors in the 2-2 split.
The strategy that we employ for the 2-2 split is based on the hook-free circuit in \autoref{ft_sc_circuit}, but we relax the amount of improvement when the hook errors generated by a plaquette are benign.
Recall that in \autoref{dist_square_ft_goal}, the single distribution edge in the $Z$ type plaquette must be fault-improved to weight $(1,2,2)$ for context-free fault tolerance. With a distribution seam orthogonal to the logical $Z$ operator, $Z$ type hook errors are benign, so that an improvement to weight $(1,1,1)$ suffices for the distributed $Z$ stabilizers. Meanwhile, the $X$ type hook errors of the distributed $X$ stabilizers align with the logical $X$ operator, so that an improvement of $(2,2,1)$ is still necessary.
Compared to the context-free case, we can halve the cost of fault improvement for half the plaquettes. In comparison, a 3-1 split requires only an improvement to the weight $(1,1,1)$ per stabilizer, but doubling the number of stabilizers makes the overall cost higher.

The above hook error considerations apply not only to the memory setting, but also to distributed lattice surgery.
The logical operators of the merged patches are shown in~\autoref{fig:sc_surg_comb} (b).
Distributed hook errors are benign to the joined operator orthogonal to the seam, but are malignant to the correlated operators parallel to the seam, as shown in~\autoref{fig:sc_surg_comb_hooks}.
The plaquettes with malignant hook errors again need to be fault-improved up to context-free fault tolerance levels, while for the plaquettes with benign hook errors only half of the fault-improvement is needed.

\begin{figure}[!htb]
  \begin{subfigure}[t]{0.4\textwidth}
    \centering
    \includegraphics[width=\textwidth]{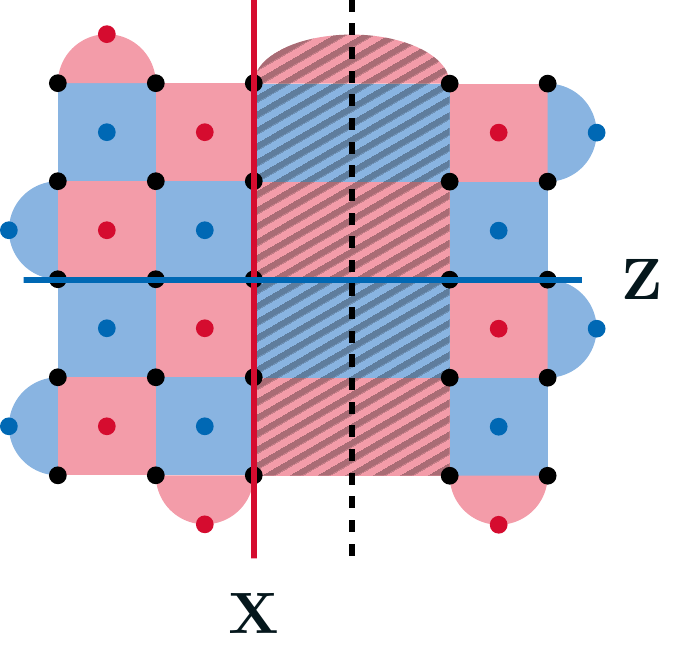}
    \caption{}
    \label{fig:sc_mem_obs}
  \end{subfigure}
  \hfill
  \begin{subfigure}[t]{0.4\textwidth}
    \centering
    \includegraphics[width=\textwidth]{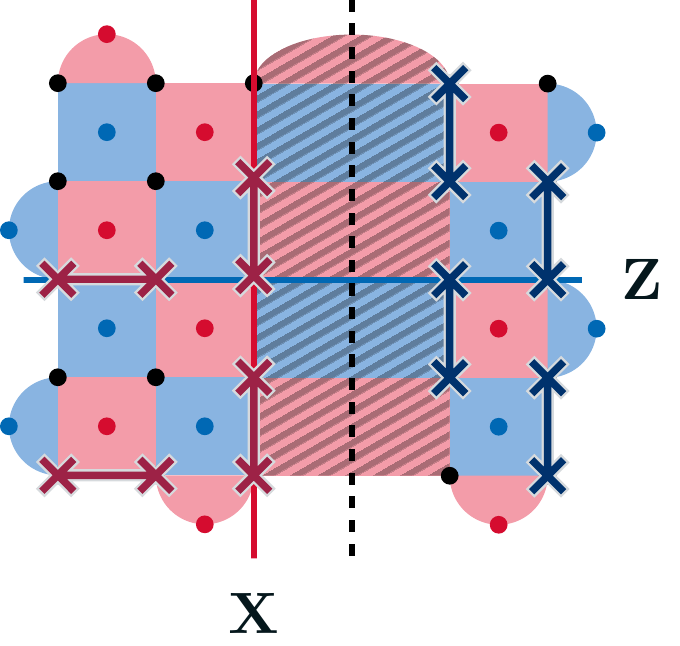}
    \caption{}
    \label{fig:sc_mem_hook}
  \end{subfigure}
  \caption{(a) A surface code patch hosting one logical qubit is distributed along a vertical seam illustrated by the dashed line. The red vertical line represents a logical $X$ operation, while the blue horizontal line represents a logical $Z$ operation. (b) The monolithic hook errors can be orthogonally aligned with the corresponding logical operators. The distributed plaquettes always produce hook errors parallel to the distribution seam. With the seam shown in the figure, $Z$ hook errors can be neglected but $X$ hook errors reduce the circuit distance, allowing a logical $X$ operation to be realized from fewer atomic faults.}
  \label{fig:sc_mem_dist}
\end{figure}

\begin{figure}[!htb]
  \begin{subfigure}[t]{0.45\textwidth}
    \centering
    \includegraphics[width=\textwidth]{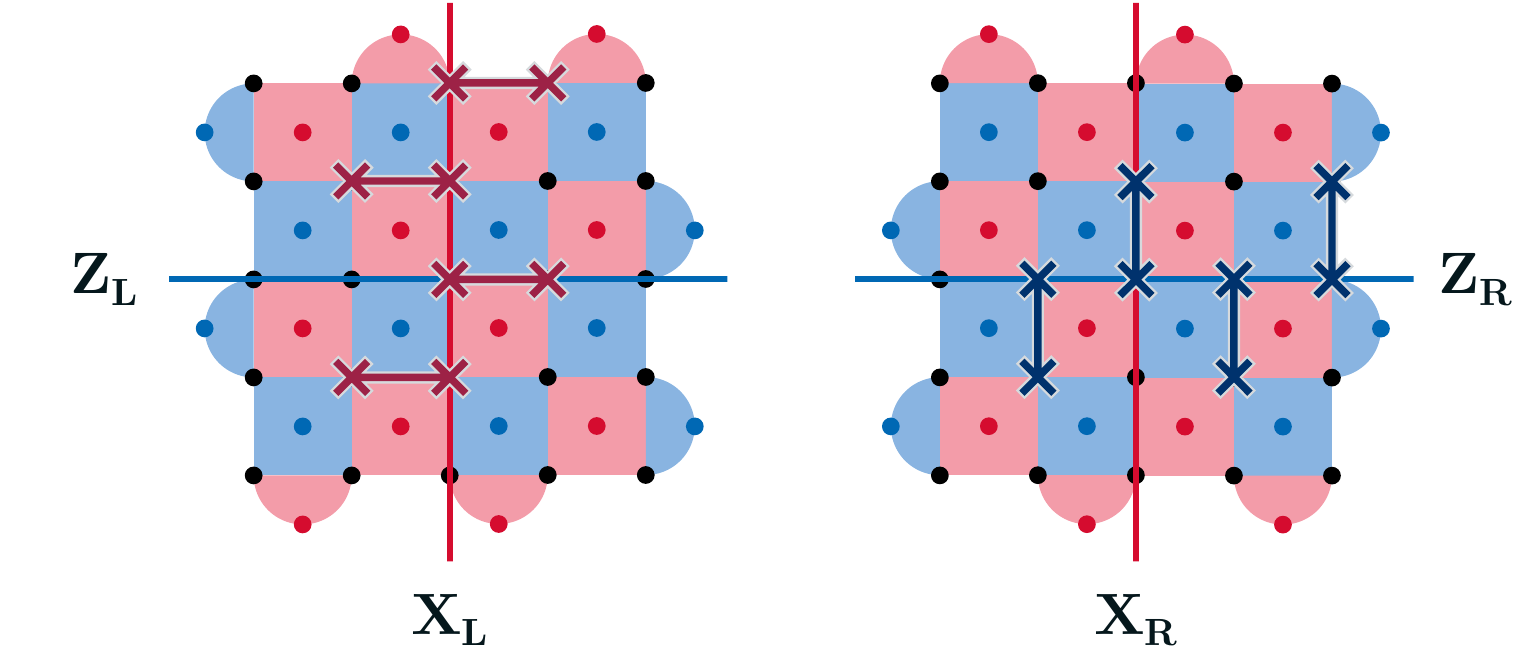}
    \caption{}
    \label{fig:sc_surg_pre_hooks}
  \end{subfigure}
  \hfill
  \begin{subfigure}[t]{0.45\textwidth}
    \centering
    \includegraphics[width=\textwidth]{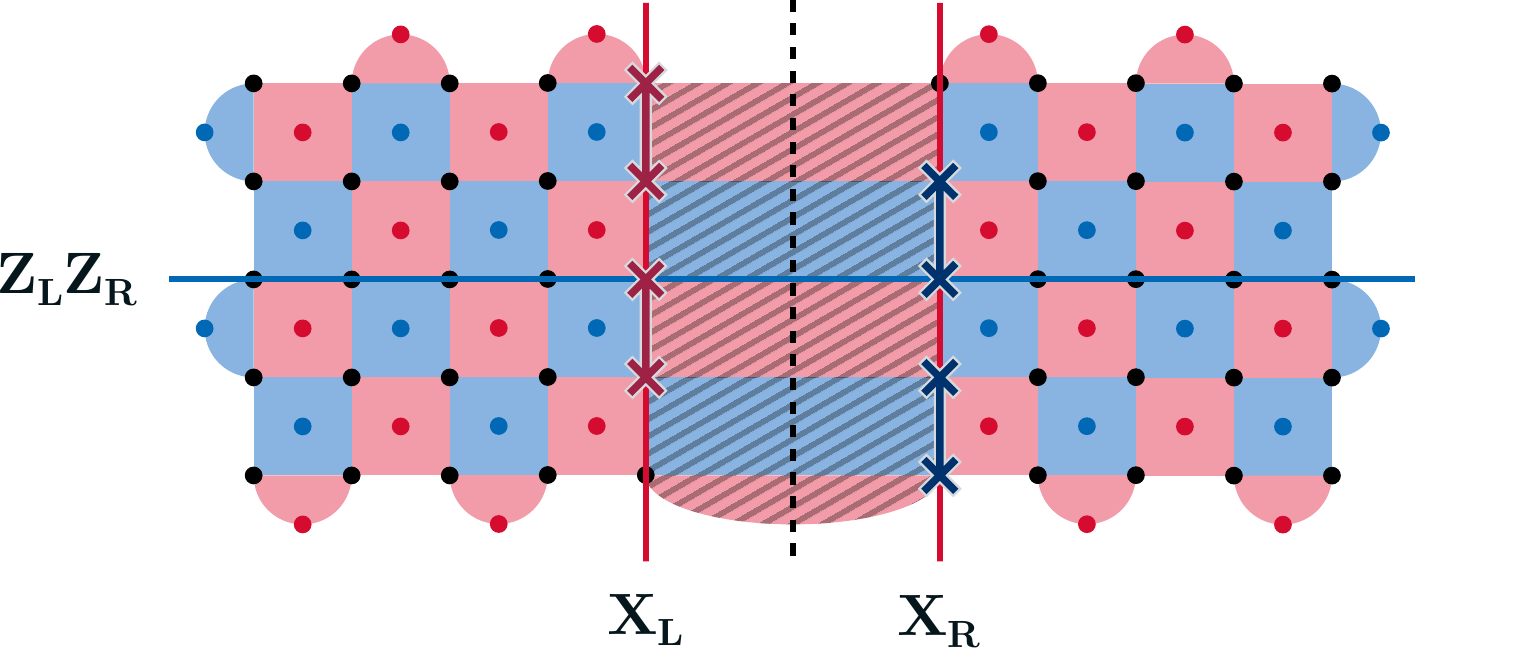}
    \caption{}
    \label{fig:sc_surg_hooks}
  \end{subfigure}
  \caption{Distributed lattice surgery: (a) Using the default extraction circuits in the monolithic setting, hook errors occur but are aligned orthogonal to the respective logical operators, preserving full circuit distance. (b) During the distributed lattice surgery, hook errors can occur and always spread parallel to the seam direction. In the example shown in the figure, $Z_{L}Z_R$ is unaffected but the circuit distance of the vertical $X$ logical operators is reduced.}
  \label{fig:sc_surg_comb_hooks}
\end{figure}

\FloatBarrier

\subsection{Timelike Errors and Lattice Surgery in the Surface Code}
\label{sec:sc_tle}

In the setting of a distributed surface code memory, any timelike errors on the distributed plaquettes can be detected by comparing to the surrounding baseline, just as in the monolithic setting. In the distributed lattice surgery setting, the same holds for \emph{half} of the distributed stabilizer measurements across the seam. The other half of the stabilizers are used jointly to extract the outcome of the logical two-qubit $XX$ or $ZZ$ measurement. Neither the logical measurement nor the individual plaquettes that constitute it have baselines for their outcomes, and the 2D surface code does not have any metachecks that would signal if any of the measurement outcomes is wrong. Instead, the logical measurement must be repeated, so that the outcome can be determined through a majority vote.

In monolithic lattice surgery, assuming fault-tolerant syndrome extraction and noise that scales as $p$, a given stabilizer measurement will produce the wrong outcome with a probability that scales with $p$. To compensate, we repeat the stabilizer measurements within the seam $d$ times for a distance $d$ surface code, and take a majority vote. This creates a timelike distance of $d$ as well:
Only if $\lfloor \frac{d+1}{2}\rfloor$ many readouts fail does the majority shift and the measurement result is incorrectly flipped.

In the distributed setting, without sufficient edge improvement, the increased error rate on the Bell pairs can cause timelike errors to occur with a probability that scales as  $p^{\frac{1}{\ebitnoise}}$. To compensate, the measurement must now be repeated $d\ebitnoise$ times to reach a timelike distance $d$.

For instance, the state teleportation in \autoref{fig:tp_ls} uses a distributed $XX$ measurement that is performed via lattice surgery, where the $X$ plaquettes added to measure out the joint $XX$ stabilizer have no baseline to relate to. In contrast, the added $Z$ plaquettes can be related to the history of the boundary weight-two $Z$ plaquettes of previous time steps\footnote{In the monolithic setting of \autoref{fig:sc_mono_surg}, auxiliary data qubits are added at the seam. Their initialization is chosen as $Z$ eigenstates, so they also function as baselines.}, as shown in \autoref{fig:sc_surg_tle}.
\begin{figure}[!htb]
  \begin{subfigure}[t]{0.44\textwidth}
    \centering
    \includegraphics[width=\textwidth]{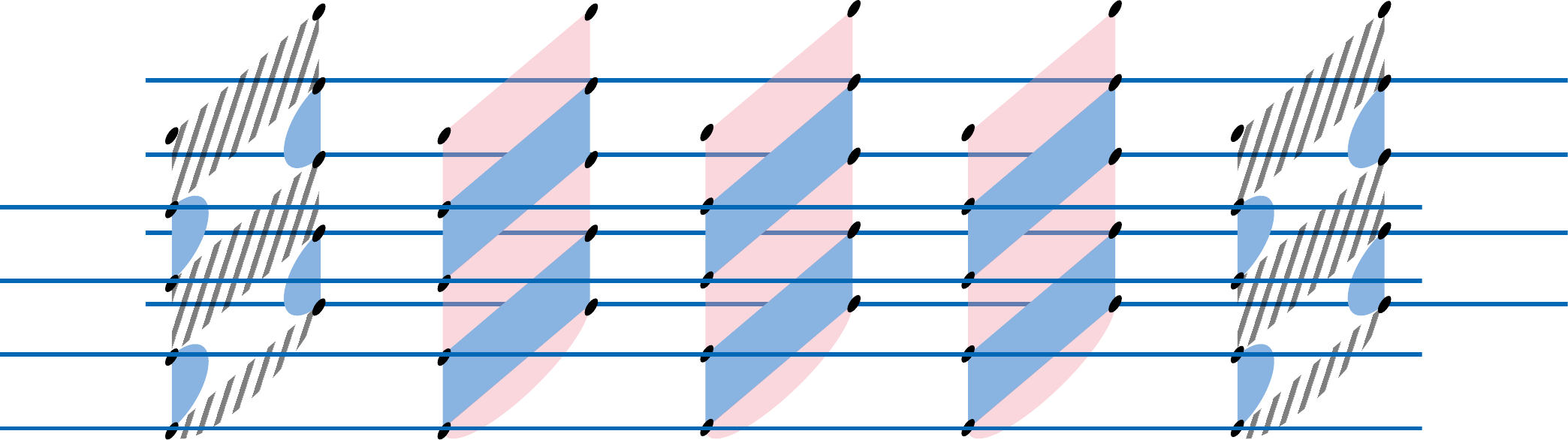}
    \caption{}
    \label{fig:sc_tle_Z}
  \end{subfigure}
  \hfill
  \begin{subfigure}[t]{0.44\textwidth}
    \centering
    \includegraphics[width=\textwidth]{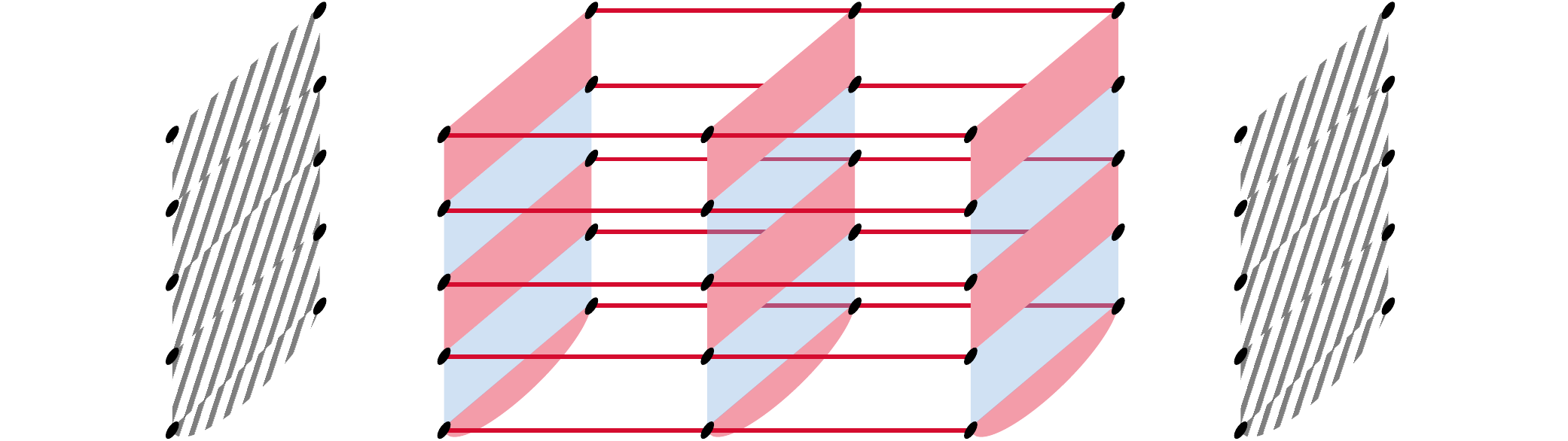}
    \caption{}
    \label{fig:sc_tle_X}
  \end{subfigure}
  \begin{subfigure}[t]{0.44\textwidth}
    \centering
    \[ \tikzfig[scale=0.4]{05_dist_sc/ls_seam_baseline_Z} \]
    \caption{}
    \label{fig:sc_tle_Z_ZX}
  \end{subfigure}
  \hfill
  \begin{subfigure}[t]{0.44\textwidth}
    \centering
    \[ \tikzfig[scale=0.4]{05_dist_sc/ls_seam_baseline_X} \]
    \caption{}
    \label{fig:sc_tle_X_ZX}
  \end{subfigure}
  \caption{During distributed lattice surgery, new distributed plaquettes are constructed between QPUs: (a) For the plaquettes matching the boundary type, there is a baseline of expected measurement outcomes from previous and future small boundary plaquettes. (b) For the opposite type of plaquettes, there is no local baseline, making them less robust against noisy readouts. (c) The timeline of a $Z$ plaquette at the seam. The dashed box indicates the surgery time frame. Weight-4 surgery plaquette measurements and weight-2 boundary plaquette measurements form detectors during the transitions between surgery and operation as separate patches. (d) The timeline of an $X$ plaquette at the seam. With no baseline, local seam detectors are confined to the surgery time frame and require precise readouts to ensure time-like code distance.}
  \label{fig:sc_surg_tle}
\end{figure}
Hence, for the distributed $Z$ plaquettes, the readout errors can be neglected, as in the memory case, whereas for the distributed $X$ plaquettes, the readout error must be reduced to the local readout probability $p$ to ensure a timelike distance of $d$ at the seam.

\FloatBarrier

\subsection{Plaquette Improvement Strategies for the Surface Code}\label{sec:sc_plaquette_improvement}

The stabilizer measurement circuits derived in
\autoref{sec:FT_dist_stab_meas}
allow separate tuning of fault improvement in the $X$ and $Z$ bases, which we can use to tailor the distributed measurement circuits for $X$ and $Z$ stabilizers to the surface-code context.
As in \autoref{tb:plaq_impl},
we will use the notation $X(d_X,d_Z)$ and $Z(d_X,d_Z)$ to denote the improvements of $X$ and $Z$ plaquettes, respectively, and we will assume an improvement strategy such that $d_Y \geq \min(d_X,d_Z)$. Each plaquette requires $d_X d_Z$ many Bell pairs to realize the improvement strategy. We summarize the results of different improvement strategies for lattice code memory and lattice surgery in \autoref{tb:str_mem} and \autoref{tb:str_ls}, respectively.
\begin{table}[h]
  \centering

  \begingroup
  \renewcommand{\arraystretch}{1.2}
  \setlength{\tabcolsep}{10pt}

  \begin{tabular}{l cccc  cc}
                                                                  & \multicolumn{4}{c}{\bf Error Probability} &
    \\
    \cmidrule(lr){2-5}
                                                                  & \multicolumn{2}{c}{\bf W-2 Hook}          & \multicolumn{2}{c}{\bf Readout} & \multicolumn{2}{c}{\bf Circuit Distance}                                                                                                  \\\cmidrule(lr){2-5}\cmidrule(lr){6-7}
    \textbf{Strategy}                                             & $X$                                       & $Z$                             & $X$                                      & $Z$                & $X$                                & $Z$                                  \\
    \midrule
    $X({2\ebitnoise,\ebitnoise})$,  $Z({\ebitnoise,2\ebitnoise})$ & $p^2$                                     & $p^2$                           & $p$                                      & $p$                & $d$                                & $d$                                  \\
    $X({2\ebitnoise,\ebitnoise})$,   $Z({\ebitnoise,\ebitnoise})$ & $p^2$                                     & $p$                             & $p$                                      & $p$                & $d$                                & $d$                                  \\
    $X({2\ebitnoise,1})$,  $Z({1,\ebitnoise})$                    & $p^2$                                     & $p$                             & $p^{1/\ebitnoise}$                       & $p^{1/\ebitnoise}$ & $d$                                & $d$                                  \\
    $X({\ebitnoise,1})$,  $Z({1,1})$                              & $\mathbf{p}$                              & $\mathbf{p^{1/\ebitnoise}}$     & $p^{1/\ebitnoise}$                       & $p^{1/\ebitnoise}$ & $\lceil\mathbf{\frac{d}{2}}\rceil$ & $\mathbf{d-1+\frac{1}{\ebitnoise}} $ \\
    \bottomrule
  \end{tabular}
  \endgroup

  ~\\

  \caption{Distributed surface code memory improvement strategies: Starting with context-free fault-tolerant plaquettes, Bell pair requirements can be reduced to allow benign hook and readout errors that do not degrade logical performance. Error probabilities that reduce circuit distance are indicated in bold. For the insufficient strategy $X({\ebitnoise,1}), Z({1,1})$, the circuit distance in the $X$ basis, which determines the memory's performance when operated in the $Z$ basis, is halved by malign $ X$-hook errors. Meanwhile, the circuit distance in the $Z$ basis is only slightly reduced by the high-probability $Z$ errors on the seam.}
  \label{tb:str_mem}
\end{table}

In the memory setting, the strategy that achieves full circuit distance with the least amount of improvement is $X(2\ebitnoise,1)$, $Z(1,\ebitnoise)$. With repetition based improvement, this strategy reduces the number of Bell pairs needed by a factor of $\ebitnoise$ for $X$ plaquettes and $2\ebitnoise$ for $Z$ plaquettes, compared to the context-free strategy. In the example of $\ebitnoise=3$, we reduce the number of Bell pairs needed by $75\%$.
\begin{table}[h]
  \centering

  \begingroup
  \renewcommand{\arraystretch}{1.2}
  \setlength{\tabcolsep}{10pt}
  \begin{tabular}{l cccc  cc}
                                                                  & \multicolumn{4}{c}{\bf Error Probability} &
    \\
    \cmidrule(lr){2-5}
                                                                  & \multicolumn{2}{c}{\bf W-2 Hook}          & \multicolumn{2}{c}{\bf Readout} & \multicolumn{2}{c}{\bf Circuit Distance}                                                                                                           \\\cmidrule(lr){2-5}\cmidrule(lr){6-7}
    \textbf{Strategy}                                             & $X$                                       & $Z$                             & $X$                                      & $Z$                & $X$                                         & $Z$                                  \\
    \midrule
    $X({2\ebitnoise,\ebitnoise})$,  $Z({\ebitnoise,2\ebitnoise})$ & $p^2$                                     & $p^2$                           & $p$                                      & $p$                & $d$                                         & $d$                                  \\
    $X({2\ebitnoise,\ebitnoise})$,   $Z({\ebitnoise,\ebitnoise})$ & $p^2$                                     & $p$                             & $p$                                      & $p$                & $d$                                         & $d$                                  \\
    $X({2\ebitnoise,\ebitnoise})$,  $Z({1,\ebitnoise})$           & $p^2$                                     & $p$                             & $p$                                      & $p^{1/\ebitnoise}$ & $d$                                         & $d$                                  \\
    $X({2\ebitnoise,1})$,  $Z({1,\ebitnoise})$                    & $p^2$                                     & $p$                             & $\mathbf{p^{1/\ebitnoise}}$              & $p^{1/\ebitnoise}$ & $\lceil\mathbf{\frac{d}{\ebitnoise}}\rceil$ & $d$                                  \\
    $X({\ebitnoise,\ebitnoise})$,  $Z({1,1})$                     & $\mathbf{p}$                              & $\mathbf{p^{1/\ebitnoise}}$     & $p^{1/\ebitnoise}$                       & $p^{1/\ebitnoise}$ & $\lceil\mathbf{\frac{d}{2}}\rceil$          & $\mathbf{d-1+\frac{1}{\ebitnoise}} $ \\
    \bottomrule
  \end{tabular}
  \endgroup

  ~\\

  \caption{Distributed surface code $XX$ lattice surgery improvement strategies: Starting with context-free fault-tolerant plaquettes, Bell pair requirements can be reduced to allow benign hook and readout errors that do not degrade logical performance. Error probabilities that reduce circuit distance are indicated in bold. In addition to the malign effects of surface code memory, insufficient readout suppression for $X$ plaquettes substantially reduces the timelike distance of the logical joint $X$ measurement.}
  \label{tb:str_ls}
\end{table}

In the lattice surgery setting, the most resource efficient strategy for reaching full circuit distance is $X(2\ebitnoise,\ebitnoise)$, $Z(1,\ebitnoise)$, saving a factor of $2\ebitnoise$ Bell pairs on the $Z$ plaquettes in the repetition based approach. In the example of $\ebitnoise=3$, the number of Bell pairs needed is reduced by $\sim 41.7\%$.

\FloatBarrier

\section{Distributed Color Code}\label{sec:color_code}

In this section, we demonstrate how distributed stabilizer measurements can be adapted to the context of the triangular color code on a hexagonal lattice.
As with the case of the surface code, we can reduce the overhead compared to the weight-4 and weight-6 stabilizer measurements shown in \autoref{subsec:Weight-4-stabilizers} and \autoref{subsec:Weight-6-stabilizers}.

As shown in \autoref{fig:cc_intui}, the color code consists of plaquettes of weight-$4$ and weight-$6$, each of which hosts both an  $X$ and a $Z$ stabilizer. Commonly chosen minimum weight logical $X$ and $Z$ operators connect two out of the three boundary corners. Other representatives (of possibly higher weight) can be created by stabilizer products, e.g. connecting one boundary corner and the opposite boundary edge, or three boundary edges, where the logical operator in the latter case contains a ``Y-shape'' intersection where three strings meet at one point\footnote{In general, logical $X$ and $Z$ operators can be represented as three-coloured paths branching from a common data qubit to connect the three distinct boundaries. Our choice of logicals as boundary edges matches with the case of a corner as the branching point. For more details, see~\cite{bombinIntroductionTopologicalQuantum2013}}.
Similar to the surface code, we can choose different types of distribution seams that use different plaquette splits. \autoref{fig:cc_dist_mem} shows a possible $3-3$ split approach as a vertical cut through the patch, as well as a $4-2$ split approach that results in a distribution seam that is parallel to a patch boundary.
\begin{figure}[!htb]
  \begin{subfigure}[t]{0.4\textwidth}
    \centering
    \includegraphics[width=\textwidth]{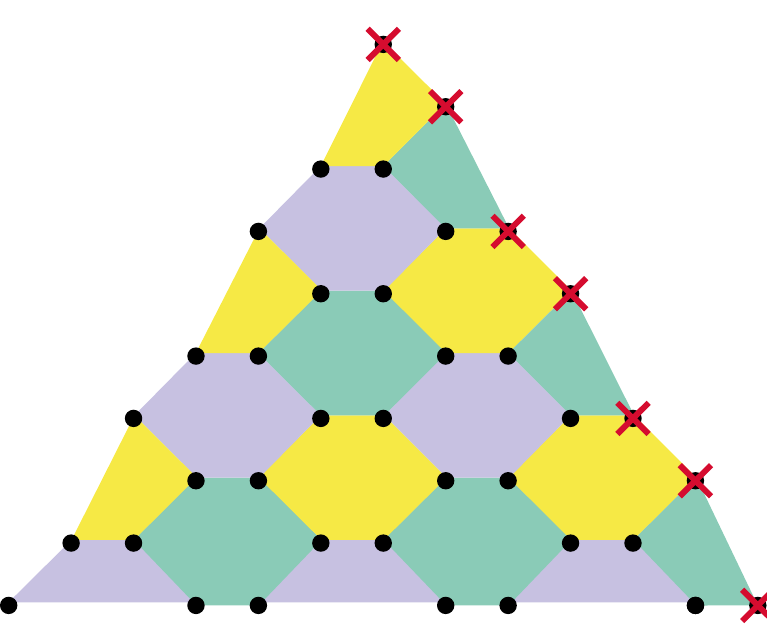}
    \caption{}
    \label{fig:cc_intuiA}
  \end{subfigure}
  \hfill
  \begin{subfigure}[t]{0.4\textwidth}
    \centering
    \includegraphics[width=\textwidth]{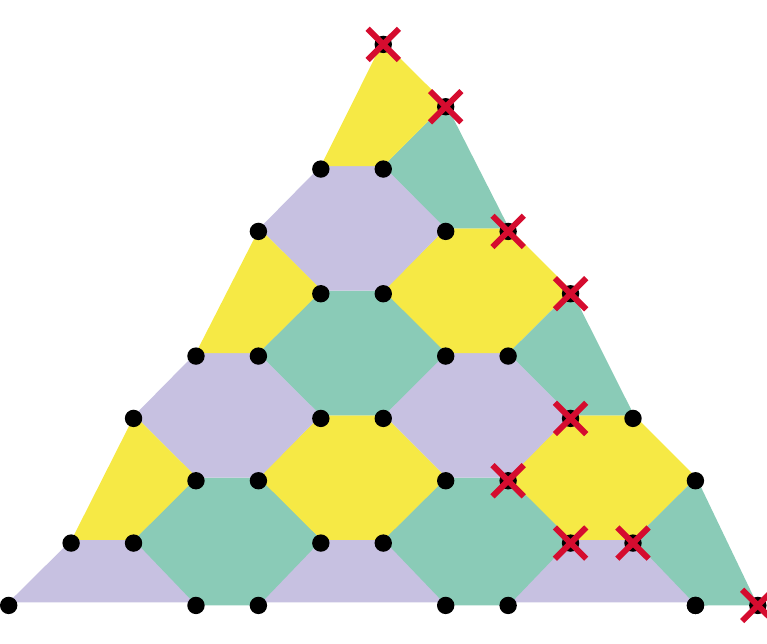}
    \caption{}
    \label{fig:cc_intuiB}
  \end{subfigure}
  \caption{Representation of a $d=7$ triangular color code. The black dots represent physical data qubits, and each face supports both an $ X$-type and a $Z$-type stabilizer acting on the data qubits it touches. The red $X$ symbols indicate physical errors of the same type, either $X$ or $Z$.
    (a) Example of a shortest error pattern leading to a logical error.
    (b) Deformation of the same logical error by multiplication with a stabilizer, yielding an equivalent logical path that passes through more than two data qubits on one plaquette.
  }
  \label{fig:cc_intui}
\end{figure}

\begin{figure}[!htb]
  \begin{subfigure}[t]{0.4\textwidth}
    \centering
    \includegraphics[width=\textwidth]{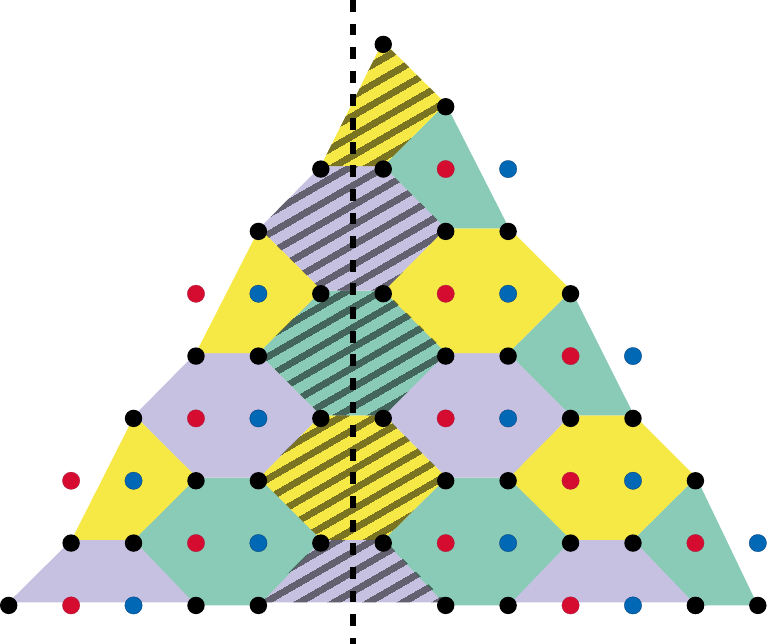}
    \caption{}
    \label{fig:cc_dist_mem_33}
  \end{subfigure}
  \hfill
  \begin{subfigure}[t]{0.4\textwidth}
    \centering
    \includegraphics[width=\textwidth]{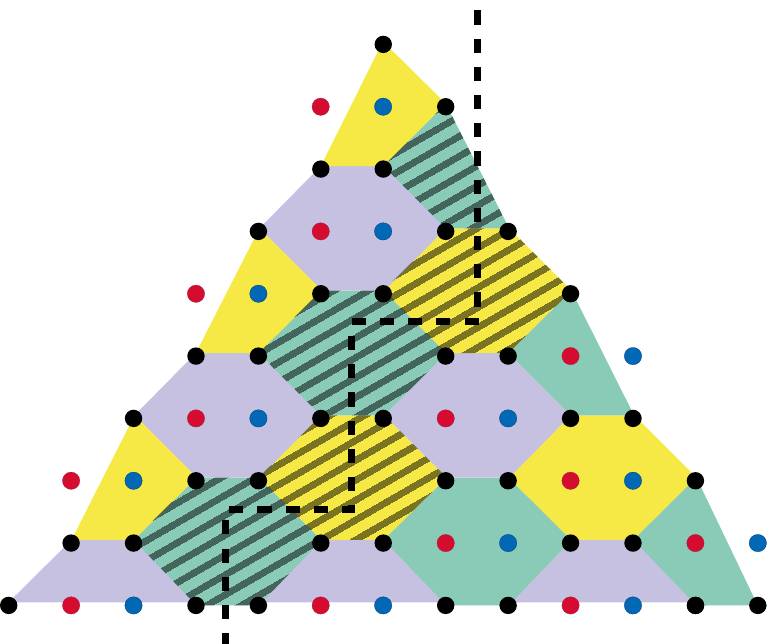}
    \caption{}
    \label{fig:cc_dist_mem_42}
  \end{subfigure}
  \caption{Distributed  $d=7$ triangular color code. Black dots represent data qubits, red and blue dots represent auxiliary qubits for $X$- and $Z$-stabilizer measurements.  The dashed line indicates the distribution seam. (a) Using the 3-3 splitting, the distribution seam aligns with a column of plaquettes. (b) Esing the 4-2 splitting, the distribution seam aligns with the boundaries of the triangular patch. All shaded plaquettes operate between QPUs and must be distributed, possibly requiring multiple auxiliary and communication qubits.}
  \label{fig:cc_dist_mem}
\end{figure}

\subsection{Syndrome Extraction and Benign Hook Errors in the Color Code}

In the monolithic setting, the color code can be implemented fault-tolerantly on a square grid layout using two auxiliary qubits per plaquette. These auxiliary qubits are shown in red and blue in \autoref{fig:cc_dist_mem}. One syndrome extraction method that has been used in the monolithic setting~\cite{lacroixScalingLogicColour2025} is the \emph{superdense} circuit introduced in Ref.~\cite{gidneyNewCircuitsOpen2023}. In addition to adapting the circuits of \autoref{subsec:Weight-4-stabilizers} and \autoref{subsec:Weight-6-stabilizers} to the color code context, we will also show how to adapt the superdense circuit to the distributed setting.

The logical operators most sensitive to a seam going across the triangle will be those that align with the seam. For this reason, the main concern will be logical operators that connect either two boundary corners, or a boundary corner and a boundary edge. An example of a logical operator with support on a minimal number of data qubits is a string along a boundary, as shown in \autoref{fig:cc_intuiA}.
Along each boundary, both weight-$4$ and weight-$6$ plaquettes contain exactly two boundary qubits.
Hence, a single fault during stabilizer measurement that propagates to faults on both data qubits can reduce the effective code distance.

For weight-$4$ plaquettes, the only possible hook errors are single faults that propagate to two data qubits.
Weight-$6$ plaquettes can additionally give rise to hook errors that affect three data qubits,
but due to the structure of the color code, these cannot be part of one of the shortest logical operators: they require deformations of the logical operator in order to avoid violating any of the code stabilizers (see~\cite{bombinTopologicalQuantumDistillation2006}).
For instance, the logical error in \autoref{fig:cc_intuiB} can be realized with a combination of hook errors that affect two data qubits and hook errors that affect three data qubits, but overall requires the same number of plaquettes to have an error as \autoref{fig:cc_intuiB}, so the residual distance stays the same.
This intuition was made rigorous in the monolithic case by~\cite{chamberlandTriangularColorCodes2020}, which considered different syndrome extraction circuits with varying levels of protection against hook errors. They showed that hooks affecting three data qubits have the same effect on circuit distance as hooks that affect two data qubits.
Therefore, a circuit that catches weight-2 hook errors and requires \emph{two} circuit faults for a propagation into three faults on the data qubits still reached full circuit distance.

As long as a weight-$3$ hook error is at most as likely as a weight-$2$ hook error, it can be neglected.
In the distributed scenario, this is no longer guaranteed, as an error propagation to two and three data qubits may include Bell pair errors and thus have vastly different error probabilities:
We can deform the logicals from the boundary to move them to the distributed plaquettes.
\autoref{fig:cc_dist_log_err} shows two scenarios, where the logical of \autoref{fig:cc_dist_log_2} can be applied by three weight-$2$ hook errors and a data qubit error, while the one from \autoref{fig:cc_dist_log_3} can be realized by one weight-$2$ hook, two weight-$3$ hooks, and a data qubit error.
In the monolithic setting without hook error suppression, the weight of the logical is $w=4$ since four atomic faults must occur in both cases. With flagging, the hook error weight increases, and we recover the code distance $w=d=7$.
\begin{figure}[!htb]
  \begin{subfigure}[t]{0.4\textwidth}
    \centering
    \includegraphics[width=\textwidth]{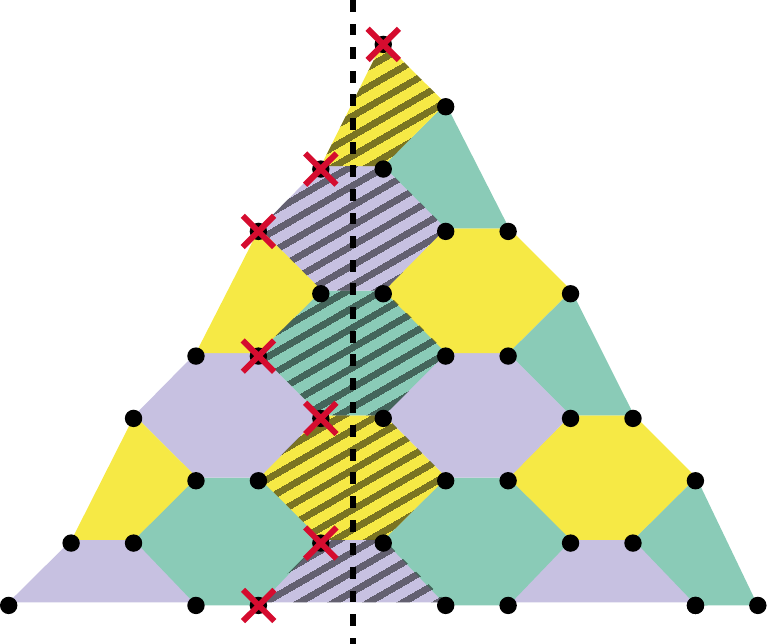}
    \caption{}
    \label{fig:cc_dist_log_2}
  \end{subfigure}
  \hfill
  \begin{subfigure}[t]{0.4\textwidth}
    \centering
    \includegraphics[width=\textwidth]{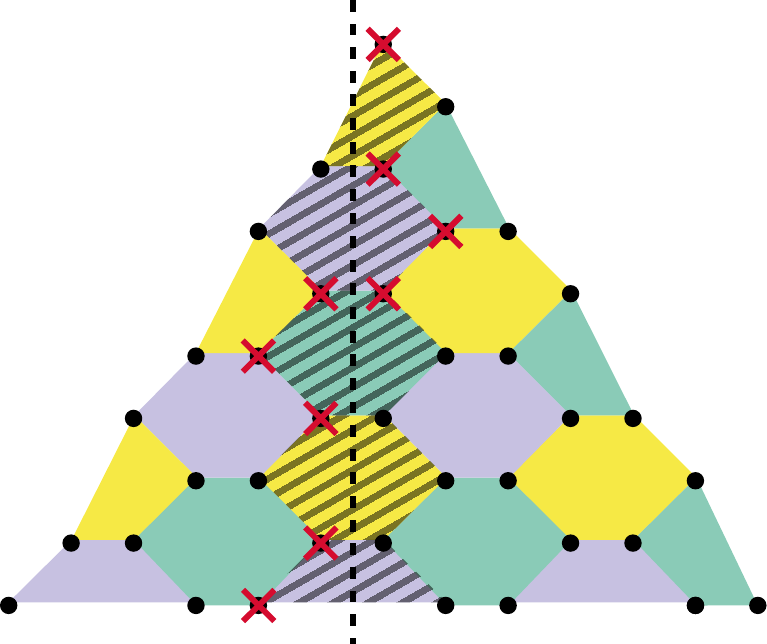}
    \caption{}
    \label{fig:cc_dist_log_3}
  \end{subfigure}
  \caption{Logical errors in a distributed $d=7$ triangular color code. The left boundary logical can be deformed so that its support is fully on the distributed 3-3 split plaquettes. (a) An error path with the same length as the boundary logical and at most two errors touching any given plaquette.
    (b) A longer error path where two distributed plaquettes are touched by three data qubit errors on the same side of the distribution seam.
    When malign hook errors are not suppressed, both paths can occur by 4 plaquettes failing. While in the monolithic case both paths have equal probability, in the distributed case (a) and (b) can have different probabilities depending on improvement and Bell pair noise strength $\ebitnoise$.
  }
  \label{fig:cc_dist_log_err}
\end{figure}
For the distributed weight-6 stabilizer measurements of \autoref{sec:FT_dist_stab_meas}, weight-2 hook errors require an error on a distribution edge combined with a \emph{local} error.
Therefore, without fault improvement, distributed weight-2 hooks are always less likely than unsuppressed monolithic hooks, and the logical path only containing weight-2 hooks is always less likely than its monolithic counterpart.
On the other hand, weight-3 hook errors require three independent faults across all distribution edges, each of which is highly likely depending on the Bell pair noise. The logical path containing weight-3 hooks can therefore be more likely than the weight-2 case, and even more likely than the lowest weight path in the monolithic case, assuming that $m$ is large and that not enough resources for fault improvement are available.

While the analysis of hook errors and circuit level fault tolerance in the color code setting is nontrivial,
the above results regarding hook errors provide a shortcut: adapting a monolithic circuit, and setting the fault-improvement of the distribution edges such that no logical is more likely than in the monolithic setting by following the above requirements on the allowed likelihood on different hook errors.
The previously introduced distributed weight-6 stabilizer measurements could be efficiently adapted to the color code setting by only improving the outer distribution edges to monolithic weight $w=1$, but an even more efficient construction starts from the following circuit:
\begin{equation}
  \tikzfig[scale=0.45]{06_dist_cc/cc_bell_flag}
\end{equation}
A complete round of syndrome extraction, including both $X$-type and $Z$-type stabilizer measurements, is shown.
The resulting circuit corresponds to the Bell-flagging construction of~\cite{baireutherNeuralNetworkDecoder2019}, which (similarly to the older flag circuit constructions) is shown to be fault tolerant.
Its distributed implementation requires two distribution edges per stabilizer.
Hence, a full $X+Z$ stabilizer-measurement round requires four distribution edges. The above reasoning about hook errors then sets the target improvement $w\geq 1$ for each distribution edge.

Another state-of-the-art construction is the  \emph{superdense} circuit. Its validity can be shown using ZX-calculus:
\begin{equation}
  \label{eq:cc_SD_dist}
  \tikzfig[scale=0.45]{06_dist_cc/cc_SD_dist}
\end{equation}
However, the Hopf rewrite is not a fault equivalent rewrite, and the fault tolerance of the superdense circuit remains an open question.
Recent work~\cite{lacroixScalingLogicColour2025} provides evidence that the superdense circuit preserves the full code distance, proving $d_{\mathrm{circ}}=d$ for $d=3,5,7$ via exact MaxSAT calculations and finding no counterexample up to $d=15$ through heuristic searches.
The superdense circuit is particularly interesting in the distributed setting.
Since it requires only two interactions between its auxiliary qubits to support both the $X$- and $Z$-type stabilizer measurements, instead of the four interactions required in the Bell-flag implementation above, the number of distribution edges required for a full round of $X$ and $Z$ measurements is reduced by a factor of two.
While the superdense circuit can also be easily adapted to the $4$-$2$ splitting, for the memory case we prefer the $3$-$3$ splitting introduced in~\autoref{eq:cc_SD_dist}, as the symmetric split can preserve the minimal cycle time.
Improving the distribution edges to weights $w_X,w_Y,w_Z \geq 1$ recovers the monolithic behavior for the distributed memory architecture shown in~\autoref {fig:cc_dist_mem}.

As a final note on hook errors, a recent study~\cite{kishonyColorCodeOffthehook2026} succeeded in avoiding malignant hook errors in the bulk without flagging through careful schedule engineering,
requiring Bell flagging only at the boundaries to achieve fault tolerance.
Investigating the effect of scheduling for distributed implementations is therefore a natural future direction. Meanwhile, in the numerical simulations within the present work, we will focus on the superdense circuit.

\subsection{Timelike Errors and Lattice Surgery in the Color Code}

\begin{figure}[!htb]
  \hspace{.2cm}
  \begin{subfigure}[t]{0.45\textwidth}
    \centering
    \includegraphics[width=\textwidth]{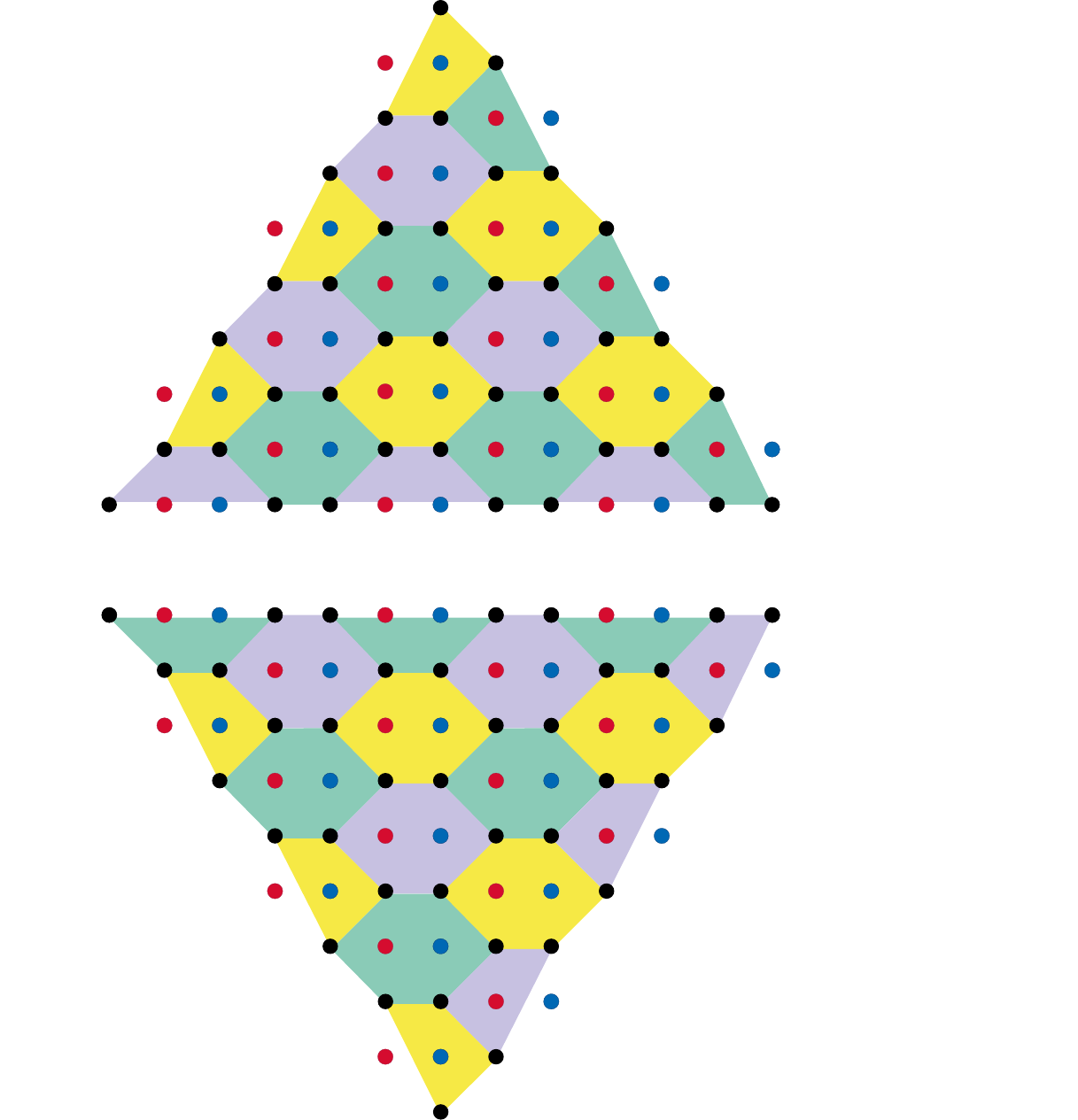}
    \caption{}
    \label{fig:cc_ls_pre}
  \end{subfigure}
  \hfill
  \begin{subfigure}[t]{0.45\textwidth}
    \centering
    \includegraphics[width=\textwidth]{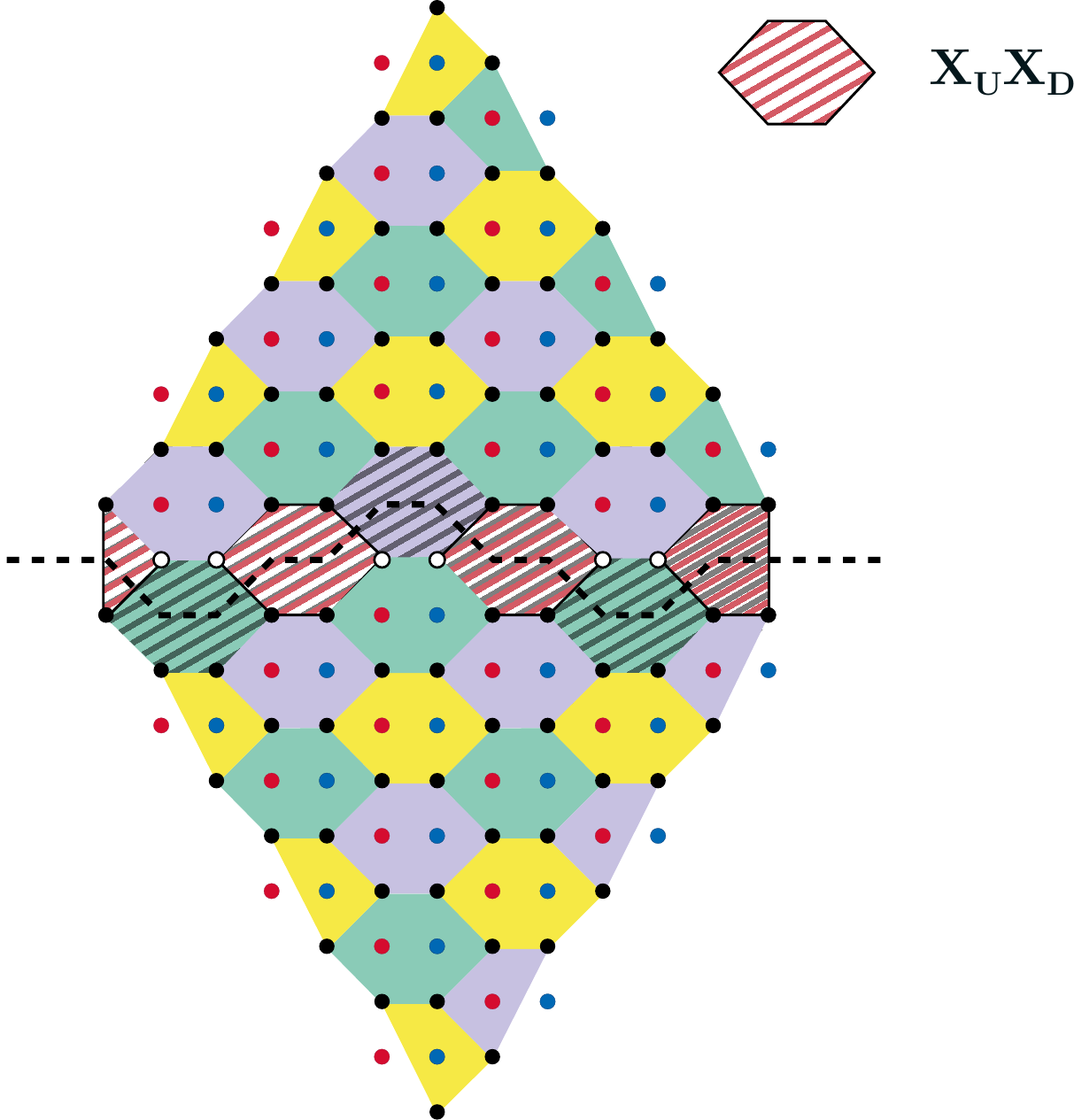}
    \caption{}
    \label{fig:cc_ls_merge}
  \end{subfigure}
  \hspace{.2cm}
  \caption{Distributed lattice surgery in the triangular color code: (a) Two $d=7$ color code patches with aligned boundaries. (b) Merged color code patches. The dashed line indicates the distribution seam. All shaded plaquettes operate between QPUs and must be distributed. The green weight-4 plaquettes on the bottom patch and the blue weight-4 plaquettes on the top patch are extended to (distributed) weight-6 plaquettes of both $X$ and $Z$ type, while the newly introduced distributed weight-6 plaquettes (shown in dashed red) extract the joint $XX$ operator by only measuring $X$ stabilizers.}
  \label{fig:cc_ls}
\end{figure}

Lattice surgery in the color code can be implemented by aligning two color code patches along their boundaries, as shown in \autoref{fig:cc_ls}, and choosing a coloring such that neighboring boundary plaquettes have compatible colors, as indicated.
The joint $XX$ (or $ZZ$) operator is measured by adding new weight-6 plaquettes between the boundaries of only $X$ (or $Z$) type (given the remaining color differing from all boundary plaquettes), which requires adding auxiliary data qubits within the surgery region.
Additionally, the weight-4 plaquettes at the borders need to be extended to weight-6 plaquettes of both $X$ and $Z$ type to properly cover the auxiliary data qubits, ensure commutativity, and restrict the degrees of freedom of the combined patch.
In a distributed scenario, both the plaquettes for logical extraction and extended local syndrome extraction need to be distributed, resulting in both 4-2 and 3-3 splits. An asymmetric splitting with all auxiliary data qubits residing on one QPU and only distributed 4-2 split plaquettes is also possible.
For the extended weight-6 plaquettes, any of the derived distributed circuits can be used to measure the $X$ and $Z$ stabilizers.
Since the plaquettes for logical extraction measure only one stabilizer type, the superdense extraction circuit reduces to the Bell-flagging extraction circuit.

The behavior of timelike errors is similar to that of the surface code. The extended plaquettes can be compared to the entire history of local $X$ and $Z$ plaquette measurements, and are therefore not affected by high-probability readout errors from the distribution edges.
In contrast, the plaquettes that contribute to the joint logical measurement outcome do not have a local baseline, and require readouts to be improved to the local readout probability $p$ to avoid degrading the timelike code distance.

\subsection{Plaquette Improvement Strategies for the Color Code}\label{sec:cc_plaquette_improvement}

As in \autoref{sec:sc_plaquette_improvement},
we summarize the plaquette improvements.
In contrast to the surface code, the symmetry of $X$ and $Z$ logicals in the color code means we do not need to distinguish between improvements for $X$ and $Z$ measurements.
We show repetition-based fault-improvement strategies for the weight-6 stabilizer adapted from the context-free case (\autoref{subsec:Weight-6-stabilizers}), denoted W6 for short,  as $\text{W6}_{33}({d_\text{left},d_\text{center},d_\text{right} | \din})$ for the 3-3 split and $\text{W6}_{42}({d_\text{left},d_\text{right} | \din})$ for the 4-2 split. Here $d_\text{left}, d_\text{center}, d_\text{right}$ denote the outer improvement $\dout$ for the left, center, and right distribution edge, respectively. The inner improvement $\din$ is chosen to be the same for each distribution edge.
For the Bell flagging (BF) and superdense (SD) circuits, the same improvements are applied to each distribution edge in the respective circuits and are denoted as $(\din, \dout)$.
Each $\text{W6}_{33}$ plaquette including $X$ and $Z$ measurement requires $2(d_\text{left}+d_\text{center}+d_\text{right})\din$ many Bell pairs to realize the improvement strategy, while $\text{W6}_{42}$ requires $2(d_\text{left}+d_\text{right})\din$, Bell flagging requires $4 \dout \din$ and supderdense requires $2 \dout \din$ many Bell pairs.

\autoref{tb:cc_str_mem} shows the effects of improvement strategies in the memory context.
For the 3-3 split, the center distribution edge of the W6 circuit does not need to be improved. Additionally, inner improvement is not required for W6 or BF in the memory context, but is required for the superdense circuit, since it measures both the $X$- and $ Z$-stabilizers.
On the other hand, the superdense circuit reduces the overall distribution edge count and achieves the most efficient cycle time if fault improvement is implemented in parallel.
For $m\leq2$, it is the most efficient overall among the context-aware fault-tolerant circuits, while for $m\geq2$, the other circuits become cheaper in terms of the number of Bell pairs required.
For the example of $\ebitnoise=3$ the Bell flagging strategy $BF(\ebitnoise,1)$ achieves the minimum number of required Bell pairs among the circuits (as does $\text{W6}_{42}(m,m,1)$). In comparison to the context-free $\text{W6}_{33}({\ebitnoise,\ebitnoise,\ebitnoise| \ebitnoise})$ strategy, it has a $\sim 77.8\%$ reduced Bell Pair cost.

For lattice surgery (considering an $XX$ measurement for concreteness, with the $ZZ$ measurement following the same pattern), 
extended local plaquettes can be improved using the same improvement strategies as in the memory setting. 
Only the seam stabilizers that measure the joint logical $XX$ operator need different improvement strategies, which are reported in \autoref{tb:cc_str_ls}.
Since we measure $X$ stabilizers only at the seam, the $X$ circuit distance is affected only by insufficient improvement, either by reducing space-like distance via hooks or time-like distance via readout errors.
Since this context is close to the context-free case, only a small reduction in Bell pair cost can be achieved. For $m=3$, compared to $\text{W6}_{33}({\ebitnoise,\ebitnoise,\ebitnoise| \ebitnoise})$, the improvement strategy $\text{BF}(\ebitnoise,1)$ saves $33.3\%$ in Bell pairs required.

\begin{table}[h]
  \centering

  \begingroup
  \renewcommand{\arraystretch}{1.2}
  \setlength{\tabcolsep}{10pt}

  \begin{tabular}{l ccc  cc}
                                                                       & \multicolumn{3}{c}{\bf Error Probability} & \multicolumn{2}{c}{\bf Circuit Distance}                                                                                                                                                                                              \\\cmidrule(lr){2-4}\cmidrule(lr){5-6}
    \textbf{Strategy}                                                  & W-2 Hook                                  & W-3 Hook                                 & Readout            & $X$                                                                               & $Z$                                                                               \\
    \midrule
    $\text{W6}_{33}({\ebitnoise, \ebitnoise,\ebitnoise | \ebitnoise})$ & $p^2$                                     & $p^3$                                    & $p$                & $d$                                                                               & $d$                                                                               \\
    $\text{W6}_{33}({\ebitnoise, 1,\ebitnoise|\ebitnoise})$            & $p^2$                                     & $p^2$                                    & $p$                & $d$                                                                               & $d$                                                                               \\
    $\text{W6}_{33}({\ebitnoise, 1,\ebitnoise|1})$                     & $p^2$                                     & $p^2$                                    & $p^{1/\ebitnoise}$ & $d$                                                                               & $d$                                                                               \\        \midrule
    $\text{W6}_{42}({\ebitnoise,\ebitnoise|\ebitnoise})$               & $p^2$                                     & $p^3$                                    & $p$                & $d$                                                                               & $d$                                                                               \\
    $\text{W6}_{42}({\ebitnoise,\ebitnoise|1 })$                       & $p^2$                                     & $p^3$                                    & $p^{1/\ebitnoise}$ & $d$                                                                               & $d$                                                                               \\
    $\text{W6}_{42}({1,1|1 })$                                         & $\mathbf{p^{1+1/\ebitnoise}}$             & $p^{2+1/\ebitnoise}$                     & $p^{1/\ebitnoise}$ & $\mathbf{\lceil\frac{d}{2}\rceil+ \lfloor\frac{d}{2}\rfloor\frac{1}{\ebitnoise}}$ & $\mathbf{\lceil\frac{d}{2}\rceil+ \lfloor\frac{d}{2}\rfloor\frac{1}{\ebitnoise}}$ \\        \midrule
    $\text{BF}({\ebitnoise,\ebitnoise})$                               & $p^2$                                     & $p^2$                                    & $p$                & $d$                                                                               & $d$                                                                               \\
    $\text{BF}({\ebitnoise,1})$                                        & $p^2$                                     & $p^2$                                    &
    $p^{1/\ebitnoise}$                                                 & $d$                                       & $d$                                                                                                                                                                                                                                   \\
    $\text{BF}({1,1})$                                                 & $\mathbf{p^{1/\ebitnoise}}$               & $\mathbf{p^{1/\ebitnoise}}  $            &
    $p^{1/\ebitnoise}$                                                 & $\lceil\mathbf{\frac{d}{2}}\rceil$        & $\lceil\mathbf{\frac{d}{2}}\rceil$                                                                                                                                                                                                    \\        \midrule
    $\text{SD}({\ebitnoise,\ebitnoise})^*$                             & $p^2$                                     & $p^2$                                    & $p$                & $d$                                                                               & $d$                                                                               \\
    \bottomrule
  \end{tabular}
  \flushleft {\footnotesize *error probabilities and circuit distance assuming fault tolerance of the monolithic counterpart, which is not proven}

  \endgroup

  ~\\

  \caption{Distributed color code memory improvement strategies: The strategies are applied for both $X$ and $Z$ measurements. Bell pair requirements can be reduced to allow weight-3 and weight-2 hook errors with probabilities up to $\geq p^2$. Error probabilities that reduce circuit distance are indicated in bold.
    For the 3-3 split W6 circuit, the improvement of $d_\text{center}$ can be dropped.
    Readout errors are also benign and need no fault-improvement.
    Only the superdense circuit still requires improvement against both types of errors as it measures $X$ and $Z$ together.}
  \label{tb:cc_str_mem}
\end{table}

\begin{table}[h]
  \centering

  \begingroup
  \renewcommand{\arraystretch}{1.2}
  \setlength{\tabcolsep}{10pt}

  \begin{tabular}{l ccc  cc}
                                                                      & \multicolumn{3}{c}{\bf Error Probability} & \multicolumn{2}{c}{\bf Circuit Distance}                                                                                                                           \\\cmidrule(lr){2-4}\cmidrule(lr){5-6}
    \textbf{Strategy}                                                 & W-2 Hook                                  & W-3 Hook                                 & Readout                       & $X$                                                                               & $Z$ \\
    \midrule
    $\text{W6}_{33}({\ebitnoise, \ebitnoise,\ebitnoise| \ebitnoise})$ & $p^2$                                     & $p^3$                                    & $p$                           & $d$                                                                               & $d$ \\
    $\text{W6}_{33}({\ebitnoise, 1,\ebitnoise| \ebitnoise})$          & $p^2$                                     & $p^2$                                    & $p$                           & $d$                                                                               & $d$ \\
    $\text{W6}_{33}({\ebitnoise, 1,\ebitnoise| 1})$                   & $p^2$                                     & $p^2$                                    & $\mathbf{p^{1/\ebitnoise}}  $ & $\mathbf{\frac{d}{\ebitnoise}}$                                                   & $d$ \\        \midrule
    $\text{W6}_{42}({\ebitnoise,\ebitnoise| \ebitnoise})$             & $p^2$                                     & $p^3$                                    & $p$                           & $d$                                                                               & $d$ \\
    $\text{W6}_{42}({\ebitnoise,\ebitnoise|1})$                       & $p^2$                                     & $p^3$                                    & $\mathbf{p^{1/\ebitnoise}}  $ & $\mathbf{\frac{d}{\ebitnoise}}$                                                   & $d$ \\
    $\text{W6}_{42}({1,1|\ebitnoise })$                               & $\mathbf{p^{1+1/\ebitnoise}}$             & $p^{2+1/\ebitnoise}$                     & $p$                           & $\mathbf{\lceil\frac{d}{2}\rceil+ \lfloor\frac{d}{2}\rfloor\frac{1}{\ebitnoise}}$ & $d$ \\        \midrule
    $\text{BF}({\ebitnoise,\ebitnoise})$                              & $p^2$                                     & $p^2$                                    & $p$                           & $d$                                                                               & $d$ \\
    $\text{BF}({\ebitnoise,1})$                                       & $p^2$                                     & $p^2$                                    &
    $\mathbf{p^{1/\ebitnoise}}  $                                     & $\mathbf{\frac{d}{\ebitnoise}}$           & $d$                                                                                                                                                                \\
    $\text{BF}({1,\ebitnoise})$                                       & $\mathbf{p^{1/\ebitnoise}}$               & $\mathbf{p^{1/\ebitnoise}}  $            &
    $p$                                                               & $\lceil\mathbf{\frac{d}{2}}\rceil$        & $d$                                                                                                                                                                \\        \midrule
    $\text{SD}({\ebitnoise,\ebitnoise})^*$                            & $p^2$                                     & $p^2$                                    & $p$                           & $d$                                                                               & $d$ \\
    \bottomrule
  \end{tabular}
  \flushleft {\footnotesize *error probabilities and circuit distance assuming fault tolerance of the monolithic counterpart, which is not proven}

  \endgroup

  ~\\

  \caption{Distributed color code $M_{XX}$ lattice surgery improvement strategies: The strategies are applied to only the plaquettes that measure the joint logical $XX$ operator. All other distributed plaquettes are assumed to be improved to fault-tolerance of the memory context.
    In addition to hook errors, high probability readout errors are now malign and degrade logical performance at $\leq p$. Since the plaquettes only measure $X$ stabilizers, only the time- and space-like distances in the $X$ basis are affected.}
  \label{tb:cc_str_ls}
\end{table}

\FloatBarrier

\section{Simulation Results}\label{sec:numerics}

While fault equivalence provides guarantees on the slopes of the logical error curves in the regime of noise, or equivalently on the exponent of the leading term in $p$, it neither determines the prefactor of the leading term nor the subleading terms. To benchmark the performance of the different distributed syndrome extraction circuits and obtain logical error curves, we turn to Monte Carlo simulations. The simulations are performed using the library stim~\cite{gidneyStimFastStabilizer2021}.

Distributed entanglement sharing is not natively supported in stim, but the flexibility of the library allows for the implementation of Bell pairs that follow the relevant noise model defined in~\autoref{sec:prelim}.
Unless stated differently, simulations are performed at a Bell pair noise level of $p^{1/3}$ compared to the circuit level noise level $p$. We implement a noisy Bell pair as follows: First, the two qubits that host the Bell pair are reset in the $Z$ and $X$ bases and entangled via a CNOT gate. This is done without noise. Then, a two-qubit depolarizing noise channel with an error rate of $p^{1/3}$ is applied to model the Bell pair noise.
Together with such noisy Bell pairs representing the interconnects, we assume a native gateset of $\{\text{Reset}_Z, \text{Reset}_X, M_Z, M_X, CNOT\}$ as local operations on each QPU. The classical corrections in the circuits, which take the form of $X$ and $Z$ gates conditioned on measurement results, can be tracked classically and are handled in postprocessing for all simulations.

We model the QPUs locally as superconducting systems, using the established SI1000 noise model~\cite{gidneyFaultTolerantHoneycombMemory2021, gidneyBenchmarkingPlanarHoneycomb2022} (a variant of circuit level noise with prefactors adjusted to mimic the noise in superconducting hardware -- for details, see \autoref{app:CLN}). We consider a distributed setup consisting of boundary-connected QPUs, where any auxiliary qubits involved in the distributed stabilizer measurements sit near the boundary.
The connectivity in the bulk of each QPU is restricted to a square grid architecture layout, typical for superconducting systems. Meanwhile, the connectivity is unrestricted for the
communication qubits and auxiliary qubits near the boundary.
This allows for simpler implementations, with the benchmarks showing the consequences of the main tradeoffs between space and time requirements of the plaquette designs without additional routing overhead.
For completeness, we also show a space-efficient implementation of surface code lattice surgery which is fully realizable on a square grid layout, including for the communication qubits and auxiliary qubits near the boundary.

\subsection{Distributed Weight-4 and Weight-6 Plaquettes in the [[4,2,2]] and [[6,4,2]] Codes}

We first benchmark the distributed weight-4 and weight-6 stabilizer measurement circuits derived in \autoref{sec:FT_dist_stab_meas} in a particularly simple setting: codes consisting \emph{only} of a single distributed weight-4 or weight-6 plaquette that hosts one $X$-type stabilizer and one $Z$-type stabilizer. The codes in question are the distributed $[[4,2,2]]$ code (with stabilizers $XXXX$ and $ZZZZ$) and the distributed $[[6,4,2]]$ code (with stabilizers $XXXXXX$ and $ZZZZZZ$).
For both codes, we first conduct memory experiments: starting with an initial state, we apply $d=2$ rounds of syndrome extraction cycles and then measure out the data qubits to compute the parity of the logical operators, checking whether or not the logical information was preserved.
The usage of $d=2$ rounds of syndrome extraction ensures a time-like code distance equal to the space-like one.
Since the distance $d=2$ only allows for detection of single data qubit faults, but not correction, all simulations at this distance were performed using postselection on any detection event. For each simulation, we used a sample size of $10^{10}$ shots.

\begin{figure}[!htb]
  \begin{subfigure}[ht]{0.56\textwidth}
    \centering
    \[ \tikzfig[scale=0.5]{07_sim_results/422_mem_circ} \]
    \caption{}
    \label{fig:422_mem_circ}
  \end{subfigure}
  \hfill
  \begin{subfigure}[ht]{0.39\textwidth}
    \centering
    \includegraphics[width=\textwidth]{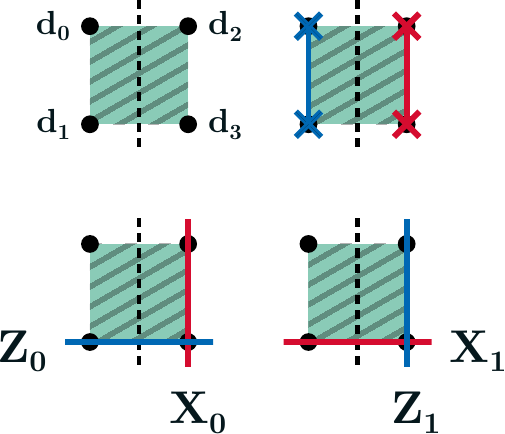}
    \caption{}
    \label{fig:422_mem_ops}
  \end{subfigure}
  \caption{Distributed $[[4,2,2]]$ code: (a) Memory circuit consisting of an initial state preparation, followed by two rounds of syndrome extraction via a single weight-4 $X$ and $Z$ plaquette, concluding with single qubit readouts. (b) The logical operators of the two logical qubits form horizontal and vertical pairs. Given a 2-2 split of the data qubits, only one logical operator per logical qubit is affected by distributed hook errors.}
  \label{}
\end{figure}

Starting with the $[\![4,2,2]\!]$ code, we consider a 2-2 split, with the circuit of the $[[4,2,2]]$ memory setting shown in \autoref{fig:422_mem_circ}.
The logical operators of the $[[4,2,2]]$ code are such that neither of the two logical qubits can be defined on only one of the two QPUs:
either the $X$ or $Z$ operator crosses the distribution boundary, as shown in \autoref{fig:422_mem_ops}.

Under fault-tolerant syndrome extraction, the lowest-weight logical errors are those that correspond to a pair of $X$ or $Z$ data qubit errors, which form a logical operator. These occur with probability $p^2$.
For distributed stabilizer measurements (circuit on the right hand side of \autoref{ft_dist_circuit}), a single fault on the distribution edge
can propagate to logical operators aligned with the seam: either $X_0$ or $Z_1$ (see \autoref{fig:422_mem_ops}).
Therefore, if the distribution edge is not fault-improved, but rather implemented with a single noisy Bell pair, the $Z_0$ and $X_1$ observables will be flipped with probability $p^{1/\ebitnoise}$.
\begin{figure}[!htb]
  \centering
  \includegraphics[width=\textwidth]{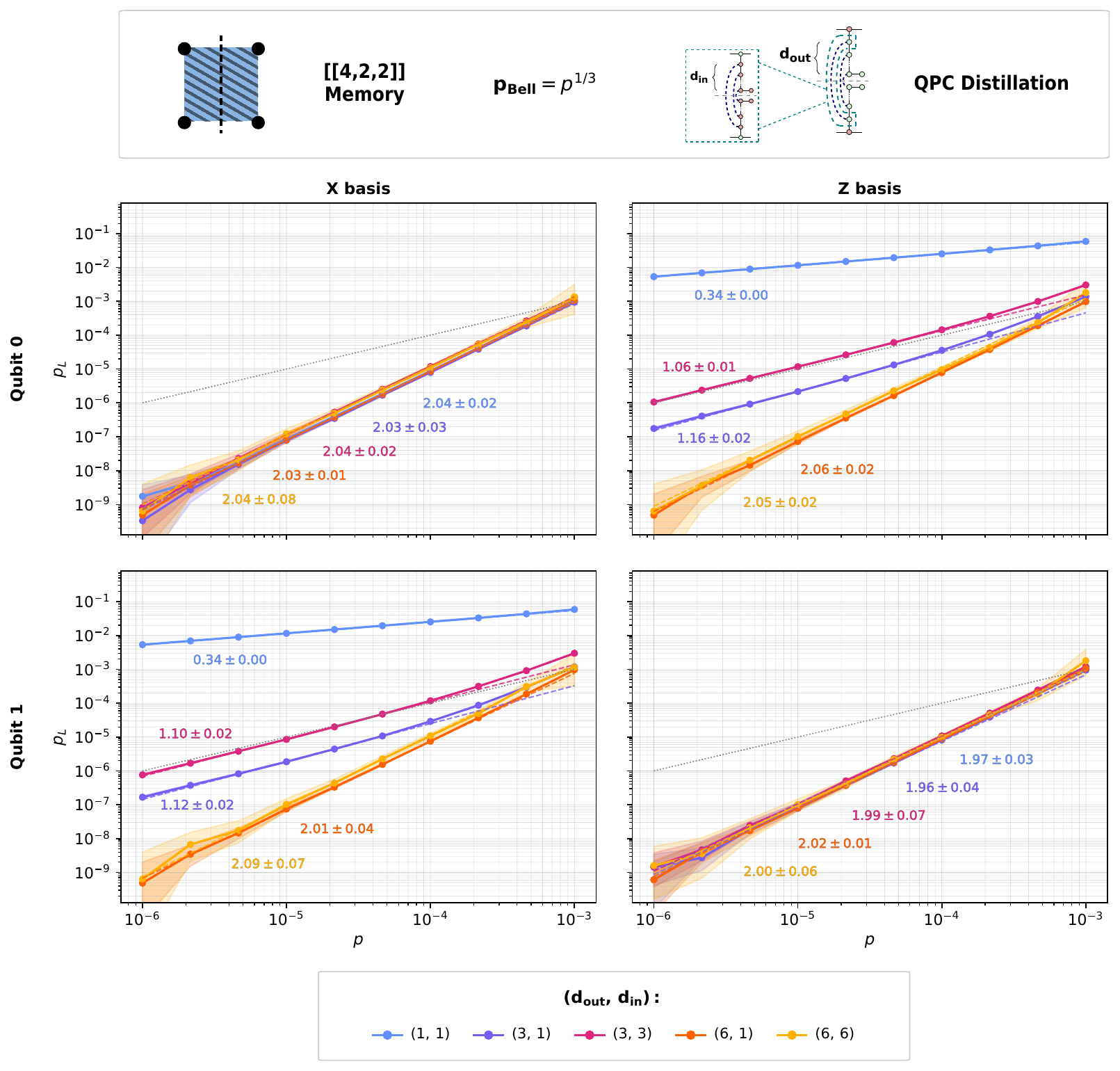}
  \caption{Distributed $[[4,2,2]]$ code: Memory simulation results for the logicals $X_0$, $Z_0$, $X_1$, $Z_1$ for Bell pair noise level $p^{1/3}$. The plaquettes were fault-improved by the \ed approach for combinations of outer and inner improvements $\dout, \din \in \{1,3,6\}$. We observe that Bell pair readout errors are always benign and that for $X_0$ and $Z_1$ with benign hook errors, no fault improvement is needed. For $Z_0$ and $X_1$ with malign hook errors, only the expected improvement of $2\ebitnoise=6$ outer improvement achieves the full circuit distance.}
  \label{fig:422_mem_res}
\end{figure}

The simulation results for the setting of a distributed memory are shown in \autoref{fig:422_mem_res} for the
two logical qubits
of the $[[4,2,2]]$ code, each measured in both the $X$ and $Z$ basis, with fault-improved plaquettes using the \ed approach.
Under error detection, we expect a slope of $d=2$ when the plaquettes are implemented fault-tolerantly.
(We recall from \autoref{sec:FT_dist_stab_meas} that distributed fault-tolerance is reached  with outer improvement $\dout=2\ebitnoise$ and an inner improvement $\din = \ebitnoise $.)
To demonstrate  the effect of varying degrees of fault-improvement
in the setting of Bell pair noise strength $\ebitnoise=3$, we probe combinations of outer and inner improvements $\dout, \din \in \{1,3,6\}$.
The results show the anticipated split between logicals that are parallel/orthogonal to the distribution seam and have benign/malign propagated errors.
The effect of a logical $X_0$ error is seen in the $Z$ basis measurement of qubit 0 (vice versa for $Z_0$), and same for qubit 1.
For $Z_0$ and $X_1$ errors, we reach the desired slope of $d=2$ in the corresponding measurements regardless of the amount of outer improvement, as the propagation of faults on the distribution noise does not contribute to such errors.
In contrast, for $X_1$ and $Z_0$ errors we see that an improvement of $\dout = 2\ebitnoise$ is necessary: there is a clear separation in the low-$p$ slopes between $\dout=1$, $3$ and $6$, with slopes of logical curves that converge to their asymptotic value as $p$ is lowered until we reach the expected $d=2$ logical performance for $\dout=6$.

We observe that $\din$ has no effect on the asymptotic behavior in any of the memory simulations.
The inner improvement affects only time-like errors, and the logical input state provides fixed expected outcomes for the first stabilizer measurement.
\begin{figure}[!htb]
  \centering
  \includegraphics[width=\textwidth]{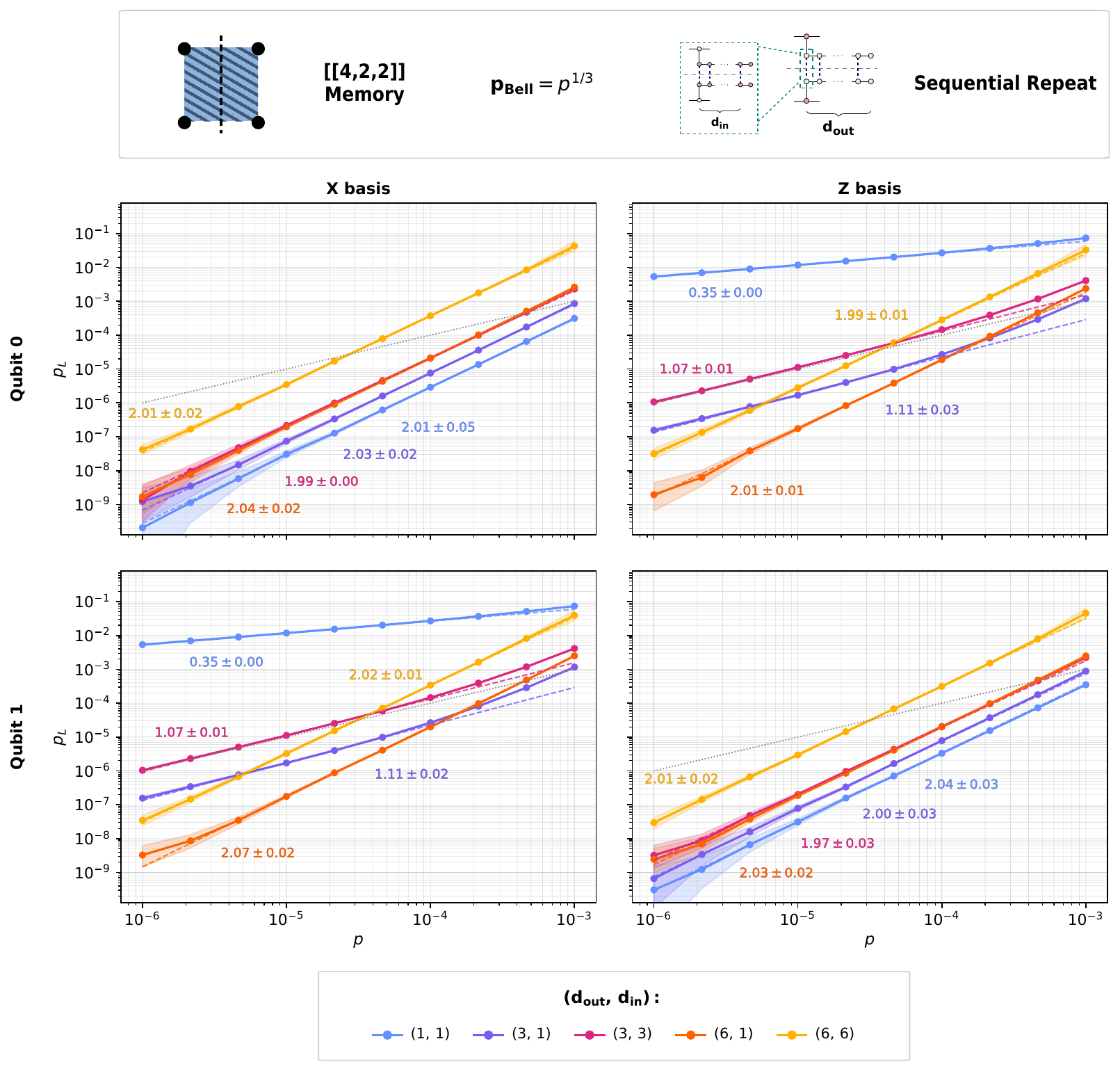}
  \caption{Distributed $[[4,2,2]]$ code: Memory simulation results for the logicals $X_0$, $Z_0$, $X_1$, $Z_1$ for Bell pair noise level $p^{1/3}$ using the \vone approach for combinations of outer and inner improvements $\dout, \din \in \{1,3,6\}$. We observe similar recovery of asymptotic behavior as for the \ed approach. Additional idling due to the sequential Bell pair processing shifts the curves with increasing improvement.}
  \label{fig:422_mem_res2}
\end{figure}

The \vone approach, shown in \autoref{fig:422_mem_res2}, reaches the same asymptotic behavior as the \ed approach. However, this example highlights how an increase of inner improvement can lead to a \emph{decrease} in overall logical performance, seen by the vertical shift of the logical error curves.
For the \ed approach, this effect can only be seen when the outer fault improvement is insufficient, while for the \vone approach it is always visible (though even more prominent in the insufficient cases).
Two mechanisms cause this behavior:
\begin{enumerate}
  \item The impact of $\dout$ and $\din$ on the execution time $T$ of a distributed stabilizer measurement circuit.
        During the duration of $T$ the data qubits are subject to idle noise, meaning that
        if an increase in $\dout$ or $\din$ results in an increased $T$, the probability of a data qubit idle error during the stabilizer measurement also increases.
        The \ed approach has a favorable $T$ scaling, as all Bell pairs are processed in parallel and the effect is barely visible compared to the sequential approach.
  \item The increase of idle noise and gate noise on auxiliary qubits inside the distributed stabilizer measurement circuit.
        This increases the probability of undetectable sets of intra-QPU errors or undetectable combinations of intra-QPU and Bell pair faults, increasing both the prefactor to the leading $p^2$ term and contributing to higher-order subleading terms.
\end{enumerate}

\begin{figure}[!htb]
  \begin{subfigure}[ht]{0.66\textwidth}
    \centering
    \[ \tikzfig[scale=0.6]{07_sim_results/422_lr_circ} \]
    \caption{}
    \label{fig:422_lr_circ}
  \end{subfigure}
  \hfill
  \begin{subfigure}[ht]{0.29\textwidth}
    \centering
    \includegraphics[width=\textwidth]{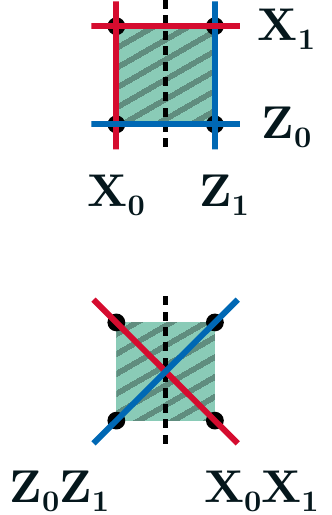}
    \caption{}
    \label{fig:422_lr_ops}
  \end{subfigure}
  \caption{Distributed $[[4,2,2]]$ code: (a) Lattice surgery circuit consisting of an initial state preparation, followed by two rounds of syndrome extraction via a single weight-4 $X$ and $Z$ plaquette and a measurement of a joint logical $ZZ$ or $XX$ via a weight-2 plaquette. (b) The logical operators of the two logical qubits are formed by a horizontal and a vertical qubit pair. The joint logical operators $XX$ and $ZZ$ then form pairs along any of the diagonals.}
  \label{fig:422_lr}
\end{figure}

We also consider a setting that mimics distributed lattice surgery in the [[4,2,2]] code: distributed logical two-qubit measurements.
More specifically, we measure the $X_0 X_1$ or $Z_0 Z_1$ operators via weight-2 measurements: \autoref{fig:422_lr_ops} shows the placements of the joint $XX$ and $ZZ$ operators as diagonal data qubit pairs.
The new weight-2 ``lattice surgery stabilizer'' is measured after the weight-4 $X$ and $Z$ stabilizers in each syndrome extraction cycle, as shown in \autoref{fig:422_lr_circ}.
The results are shown in \autoref{fig:422_ls_res}.

\begin{figure}[!htb]
  \centering
  \includegraphics[width=\textwidth]{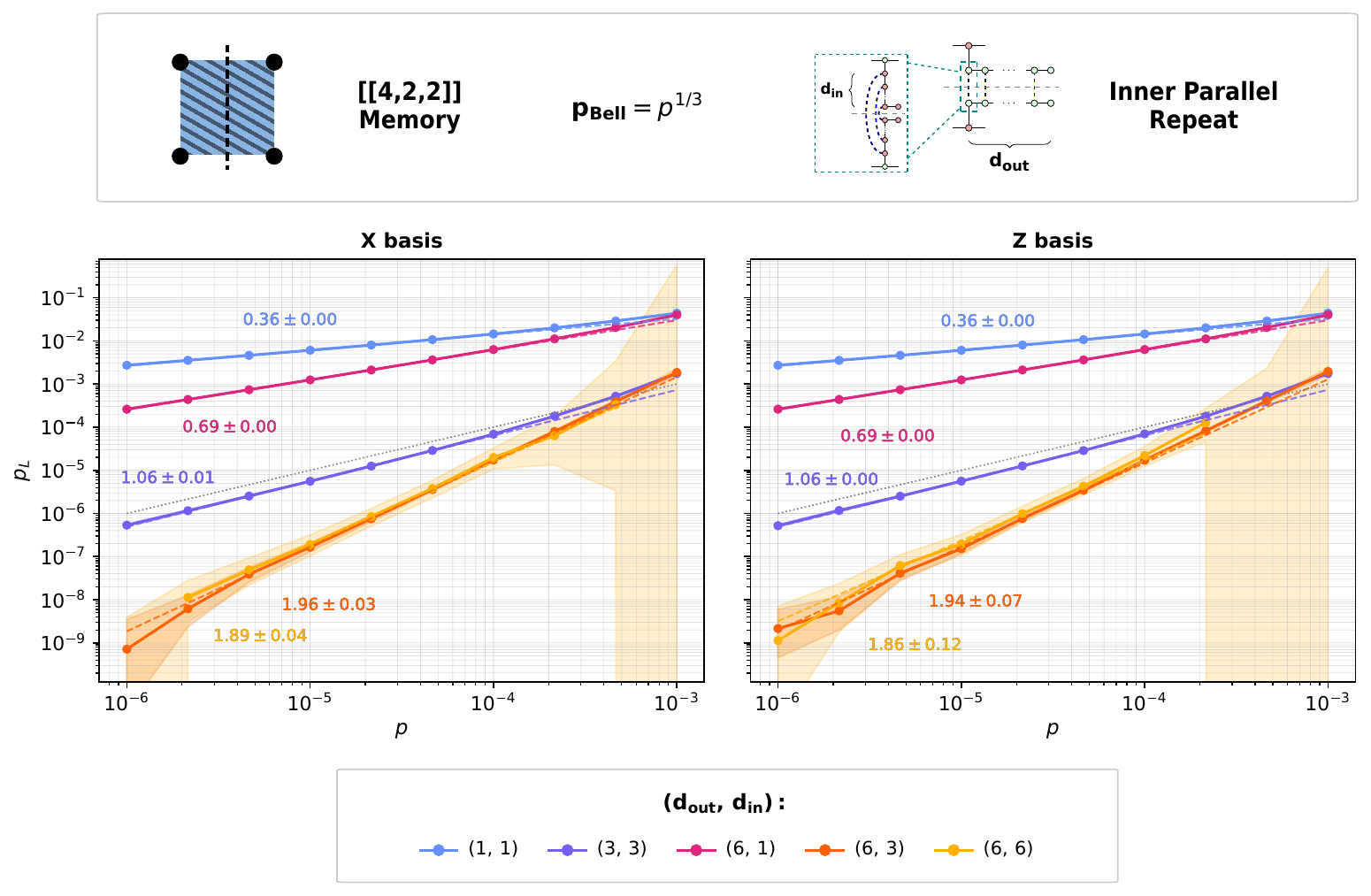}
  \caption{Distributed $[[4,2,2]]$ code: Lattice surgery simulation results for the logicals (a) $X_0 X_1$, (b) $Z_0 Z_1$ for Bell pair noise level $p^{1/3}$. The plaquettes were fault-improved by the \ed approach for combinations of outer and inner improvements $\dout, \din \in \{1,3,6\}$. We observe that Bell pair readout and hook errors are malign. Only with the improvement of $\dout \geq 2\ebitnoise=6$ and $\din\geq \ebitnoise=3$ the full circuit distance is achieved, as expected.}
  \label{fig:422_ls_res}
\end{figure}

Since the joint operators form diagonals, they always cross the distribution boundary and must be measured using distributed circuits.
The distributed weight-2 measurements do not produce hook errors themselves, but are affected by weight-2 propagated errors of the distributed weight-4 stabilizer measurements.
Compared to the memory setting, readout errors also become relevant: unlike the weight-4 stabilizers defining the code space, the new weight-2 ``lattice surgery stabilizer'' measurements lack a baseline to compare to. When repeated twice they form a single detector between them that can catch readout errors, so that the time-like distance now relies on their readout error probability.
In the case of $\dout=1$, hook errors dominate as the leading error term and $\din$ only impacts through higher-order effects.
For $\dout=3$ and $\din=1$, the insufficient suppression of readout errors dictates the logical performance, and only with $\din \geq3$ do we recover the different asymptotic behaviors seen in the memory setting.
Finally, for $\dout=6$ and $\din\geq 3$ the stabilizer measurements become fault-tolerant and the joint measurement achieves the full circuit distance $d=2$.

We proceed with the [[6,4,2]] code, focusing on the memory setting. (The effect of noisy readouts requiring $\din$ improvement for [[6,4,2]] ``lattice surgery'' is analogous to the [[4,2,2]] case.)
We simulate both a 3-3 split and a 4-2 split, using the distributed weight-6 stabilizer measurement circuits derived in \autoref{subsec:Weight-6-stabilizers}.
The results demonstrate the different effects of error propagation to data qubits.
For the sake of brevity, we refrain from reporting the individual results of all eight possible $X$ or $Z$ logical measurements of the four logical qubits, and instead report the logical error rates of \emph{any} logical being flipped.
This naturally shifts the onset of the asymptotic behavior to lower noise regimes, as there are now more higher-order error mechanisms  contributing to the error rate.
\begin{figure}[!htb]
  \centering
  \includegraphics[width=\textwidth]{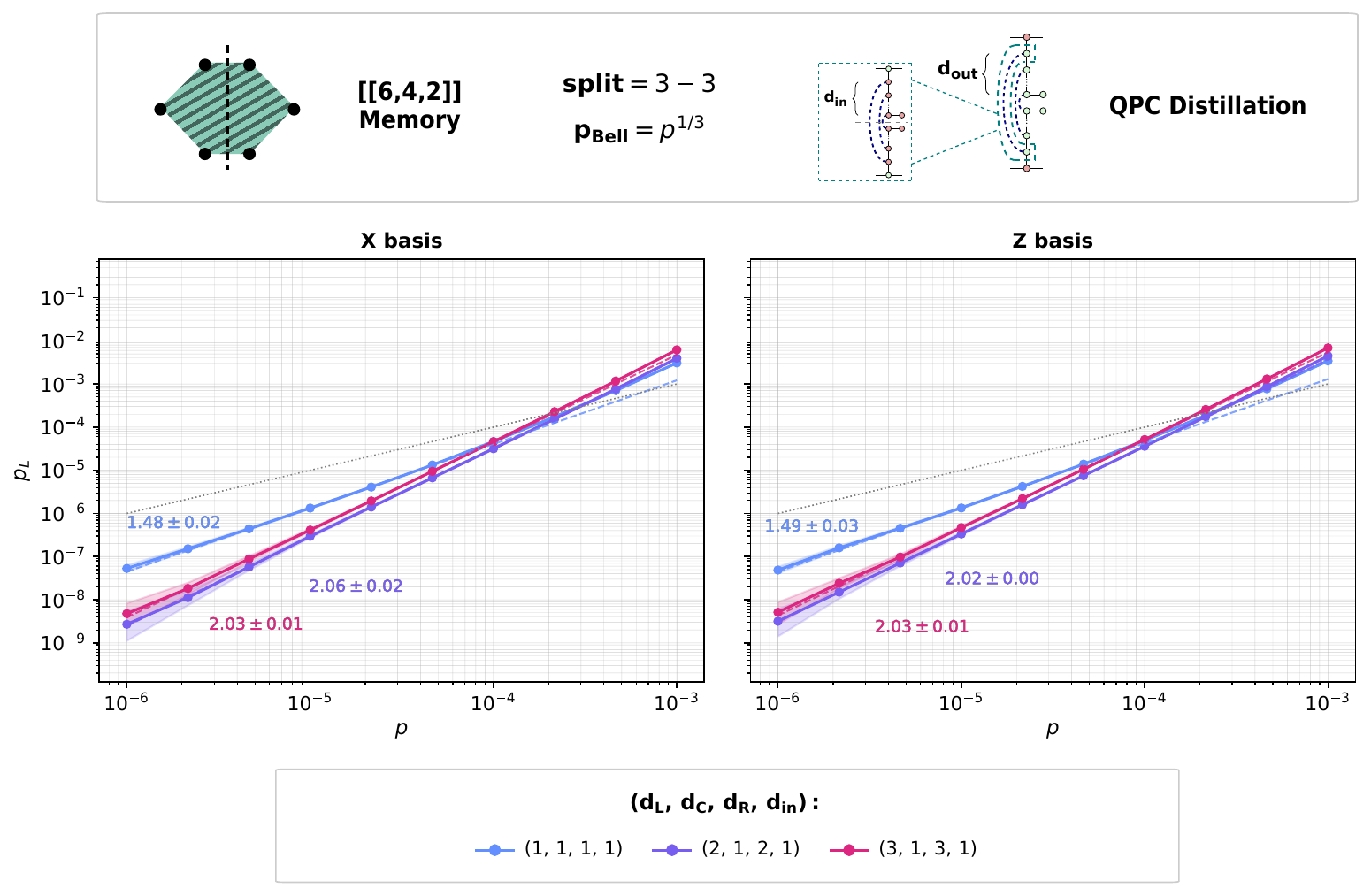}
  \caption{Distributed $[[6,4,2]]$ code: Memory simulation results of the \emph{3-3} split for Bell pair noise level $p^{1/3}$. The plaquettes were fault-improved by the \ed approach for combinations of outer improvements $d_L, d_C, d_R$ and inner improvements $\din$. We observe that Bell pair readout errors are benign and weight-2 and weight-3 hook errors are malign. With the improvement of $d_L, d_R \geq \ebitnoise=3$ the full circuit distance is achieved, as expected.}
  \label{fig:642_mem_33}
\end{figure}

\autoref{fig:642_mem_33} shows the results for the 3-3 split.
The logicals of the $[[6,4,2]]$ code are of weight two, meaning that a weight-3 propagated error or a weight-2 propagated error originating from the distribution edges can flip a logical on their own.
By choosing $(d_L, d_C, d_R, \din)$ such that we only increase the outer improvements $d_L$ and $d_R$, at $(3,1,3,1)$ the weight of both weight-2 and weight-3 propagated errors is $w\geq 2$, and we recover the code distance of $d=2$.
Interestingly, for $(2,1,2,1)$ we only see a hint of the curve bending away from a slope of two around a physical error rate $p^{-6}$. $(1,1,1,1)$ also shows only a moderate drop to a slope of $\approx 1.5$.
Their asymptotic slopes are expected to be $5/3$ and $1$, respectively.
We attribute this high-$p$ lack of asymptotic behavior to the aforementioned increase in higher-order contributions,
which could
place us within
the so-called \emph{waterfall} regime introduced by~\cite{guScalableNeuralDecoders2026}. This would mean that $w=2$ faults can dominate the logical error rate in this regime even
when lower-weight faults can contribute, due to differing prefactors.
We expect the curves will reach the asymptotic slopes in lower physical error regimes.

\begin{figure}[!htb]
  \centering
  \includegraphics[width=\textwidth]{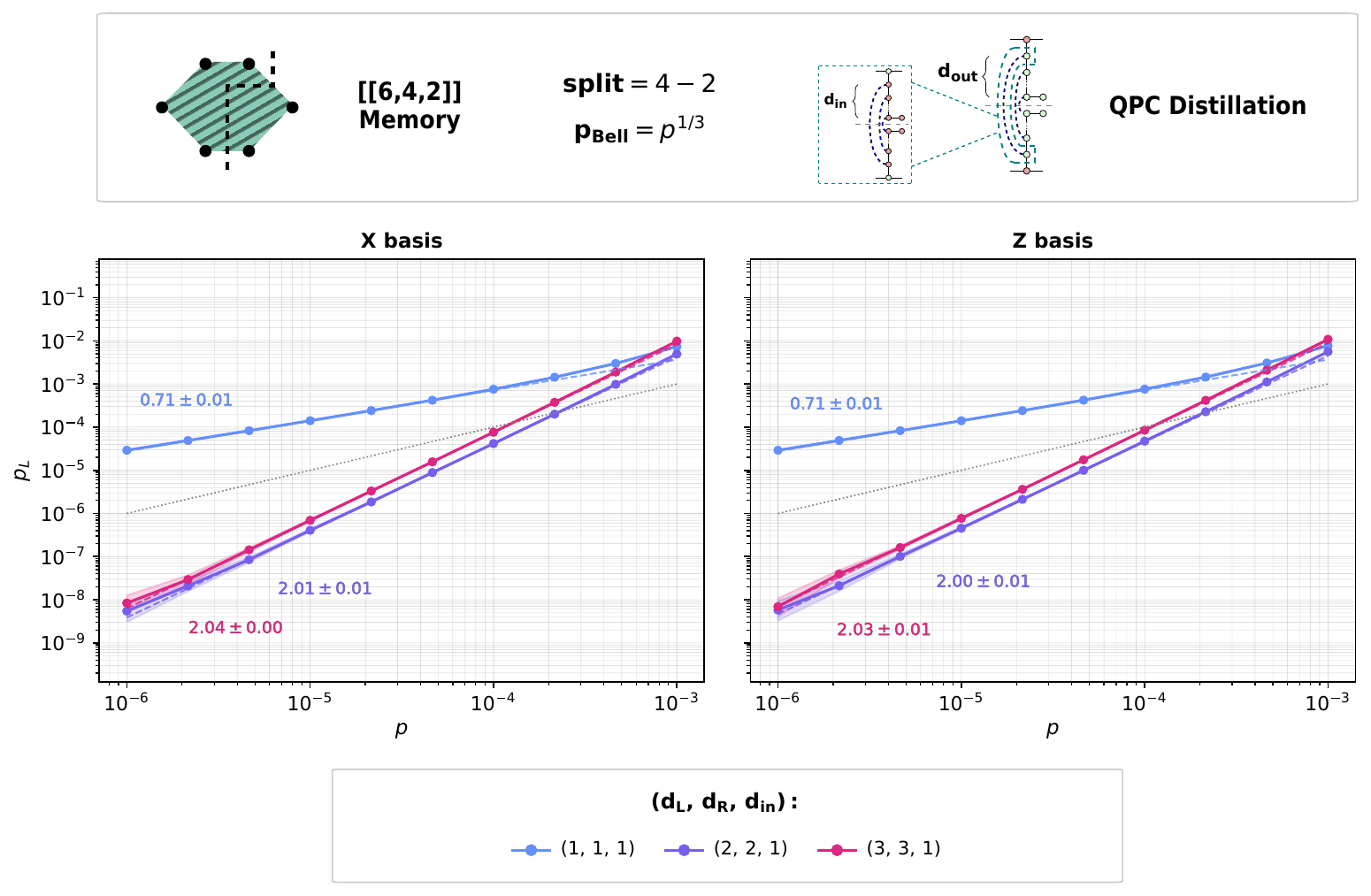}
  \caption{Distributed $[[6,4,2]]$ code: Memory simulation results of the \emph{4-2} split for Bell pair noise level $p^{1/3}$. The plaquettes were fault-improved by the \ed approach for combinations of outer improvements $d_L, d_C, d_R$ and inner improvements $\din$. We observe that Bell pair readout errors are benign and weight-2 and weight-3 hook errors are malign. With the improvement of $d_L, d_R \geq \ebitnoise=3$ the full circuit distance is achieved, as expected.}
  \label{fig:642_mem_42}
\end{figure}

The results of the 4-2 split in \autoref{fig:642_mem_42} show a similar behavior.
For improvements $(d_L, d_R, \din)$ of $(3,3,1)$ we improve weight-2 propagated faults to $w=2$, making the stabilizer measurements fault-tolerant.
However, we again observe that a $(2,2,1)$ improvement, with an expected residual distance of $4/3$, barely lowers the logical performance in the sampled physical error regime.
Meanwhile the slope of $(1,1,1)$ is closer to its expected asymptotic value of $2/3$.
\begin{figure}[!htb]
  \centering
  \includegraphics[width=\textwidth]{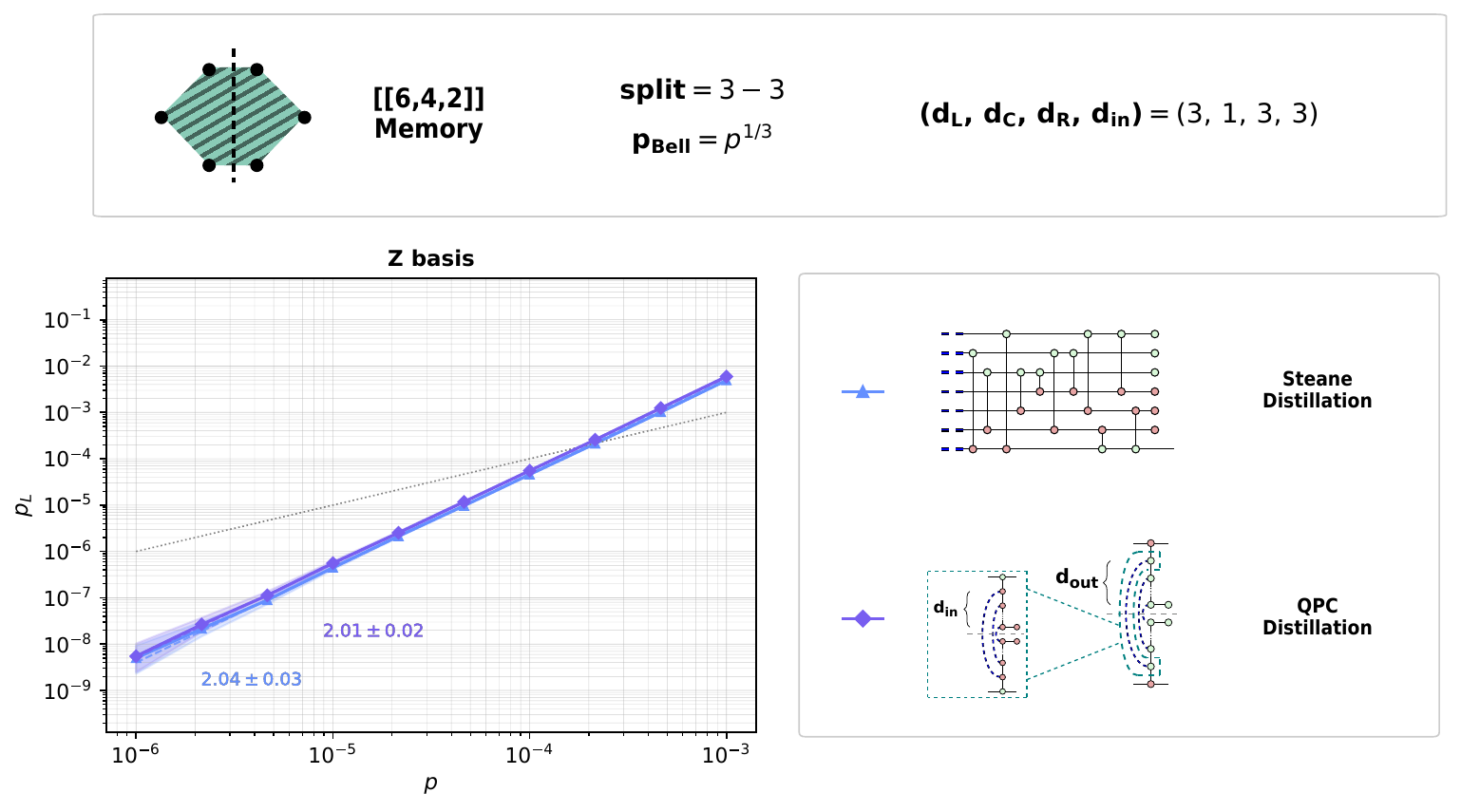}
  \caption{Distributed $[[6,4,2]]$ code: Memory simulation results of the 3-3 split for Bell pair noise level $p^{1/3}$, comparing fault-improvement with the \ed approach to Steane code distillation. Improvements are chosen as $(3,1,3,3)$, in accordance to the Steane code distances $d_X, d_Z=3$. We observe that both approaches achieve the full circuit distance.}
  \label{fig:642_mem_dist}
\end{figure}

As a final demonstration in the [[6,4,2]] setting, we implement distillation based on the Steane code for the 3-3 split, as shown in \autoref{fig:642_mem_dist}.
Since the Steane code has fixed distances $d_X, d_Z=3$, it can improve a distribution edge up to $w_X=d_X/\ebitnoise=1$,  $w_Z=d_Z/\ebitnoise=1$ in the scenario of $\ebitnoise=3$.
We choose both the left and right distribution edges to be improved by the Steane distillation and use the QPC distillation scheme to improve the center distribution edge to $\dout=1, \din=3$, resulting in an overall $(3,1,3,3)$ improvement.
Comparing this strategy to the usage of \ed for all distribution edges, we find that both have close to identical logical curves and recover full circuit distance for the $[[6,4,2]]$ code.

\FloatBarrier
\subsection{Distributed Surface Code}

Within this subsection we present the results of simulating distributed memory and lattice surgery in the rotated surface code, comparing the context-aware strategies developed in \autoref{sec:surface_code} to context-free strategies.

We consider distances $d>2$, performing no postselection but instead correcting errors based on their syndromes. This requires choosing a decoder.
The surface code has the advantage of being \emph{matchable}, making it fast to decode~\parencite{higgottSparseBlossomCorrecting2025}. However, this advantage only holds when the syndrome extraction circuits preserve matchability. While the simple syndrome extraction circuit commonly used in the monolithic setting preserves matchability, this is not the case for the distributed circuits derived in the present work: the additional detecting regions inside the distribution edges break the matchability of the overall code.

As discussed in \autoref{sec:integrated_decoding}, one option is to split the decoding into a two-step process.
While the decoding of the distillation step depends on the chosen improvement approach, the outer code is again a matchable surface code. However, the simplicity of this decoder construction comes at the cost of requiring twice the distance for the fault-improvement in the distillation step.
For this reason, we instead opted for integrated decoding, allowing for lower Bell pair requirements.
We use the decoder Tesseract~\parencite{beniTesseractSearchBasedDecoder2025}, both for the surface code simulations and for the color code simulations in the next subsection.
To keep decoding times manageable, we use the \emph{short-beam} configuration of Tesseract, trading accuracy for runtime.
\begin{figure}[!htb]
  \centering
  \includegraphics[width=\textwidth]{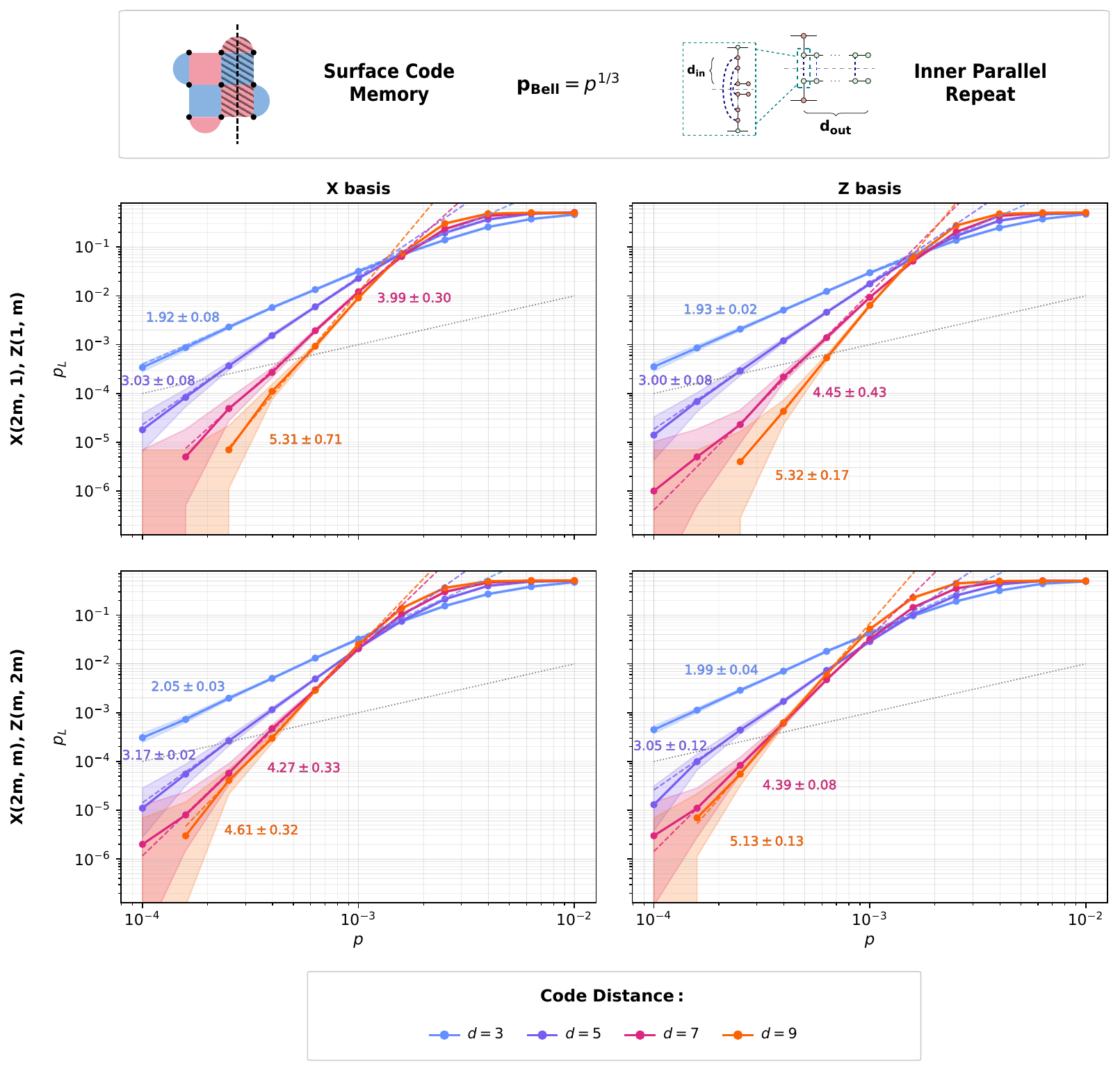}
  \caption{Distributed surface code: Memory simulation results for Bell pair noise level $p^{1/3}$ and distributed 2-2 split plaquettes improved by the \vtwo approach. Both the context-aware (top) and context-free (bottom) improvement strategies achieve the full circuit distance for both $X$ basis (left) and $Z$ basis (right), with the context-aware improvement having a slightly higher threshold.}
  \label{fig:SC_mem_dist}
\end{figure}

\autoref{fig:SC_mem_dist} shows the results of the distributed memory simulations in the $X$ and $Z$ basis, using 2-2 split plaquettes and taking $\ebitnoise=3$ . We use the \vtwo approach for both the context-free fault-improvement strategy $X({2\ebitnoise,\ebitnoise})$, $Z({\ebitnoise,2\ebitnoise})$ and the context-aware fault-improvement strategy $X({2\ebitnoise, 1})$, $Z({1 ,\ebitnoise})$, simulating the code distances $d\in \{3,5,7,9\}$.
Each configuration was evaluated with $10^7$ shots.

As is seen in the figure, both the context-free and context-aware strategy reach the desired slopes for each distance and Pauli basis.
For distances 7 and 9, the logical error curves are even surpassing the expected slopes of $4$ and $5$. We tentatively attribute this behavior to the physical noise rates still being within the waterfall regime, but also note that the results at the lower end of the error range come with a large variance (due to the need for larger sample sizes to observe a sufficient number of logical errors), and that this variance adds more uncertainty to the estimated slope.

\begin{figure}[!htb]
  \centering
  \includegraphics[width=\textwidth]{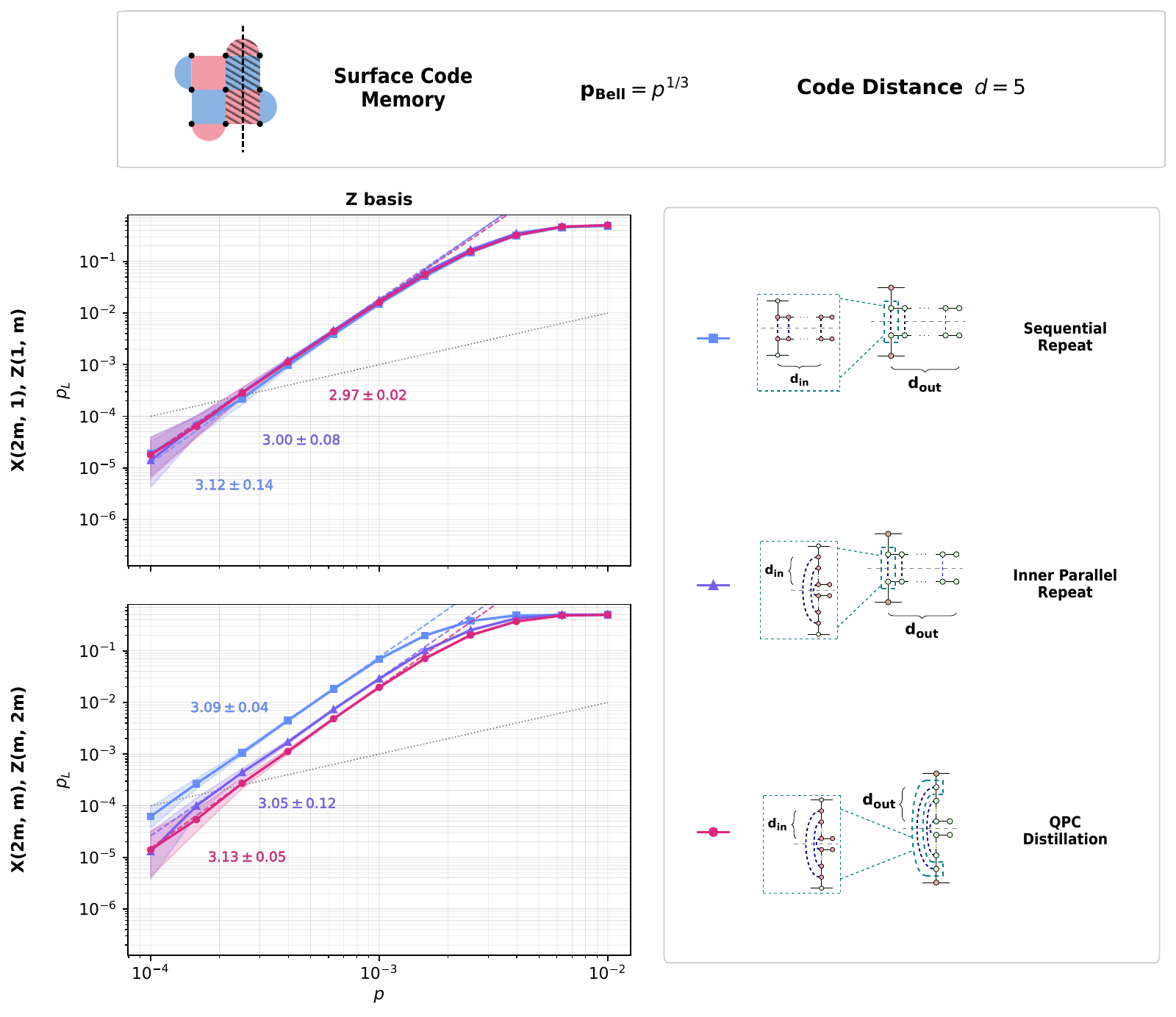}
  \caption{Distributed surface code: Memory simulation results for Bell pair noise level $p^{1/3}$ and distributed 2-2 split plaquettes for fixed distance $d=5$ in the $Z$ basis. Both the context-aware (top) and context-free (bottom) are compared for three different improvement approaches. The additional unnecessary improvement of the context-aware strategy increases idle-time for the (partially) sequential approaches, decreasing their logical performance.}
  \label{fig:SC_mem_dist2}
\end{figure}

An important point to note in \autoref{fig:SC_mem_dist}  is that in addition to its lower resource requirements, the context-aware strategy also has a better threshold compared to the context-free strategy.
Both the amount of performance improvement seen from context-awareness and the overall performance itself depend on how sequential the fault-improvement approach is.
\autoref{fig:SC_mem_dist2} shows a comparison of the \vone\!\!, \vtwo and \ed approaches for fixed distance $d=5$.
While the three approaches have similar error curves with the context-aware strategy, the \vtwo and especially \vone approach have a shift in logical performance with the context-free strategy.
The reason is that the unnecessary improvements induce an increase in idle time on the data qubits for the (partially) sequential approaches as their cycle times $T$ increase.
In contrast, for the parallel \ed approach both strategies have similar logical performance (with the context-aware approach still holding an advantage over the context-free approach in terms of the number of Bell pairs required).

\begin{figure}[!htb]
  \centering
  \includegraphics[width=\textwidth]{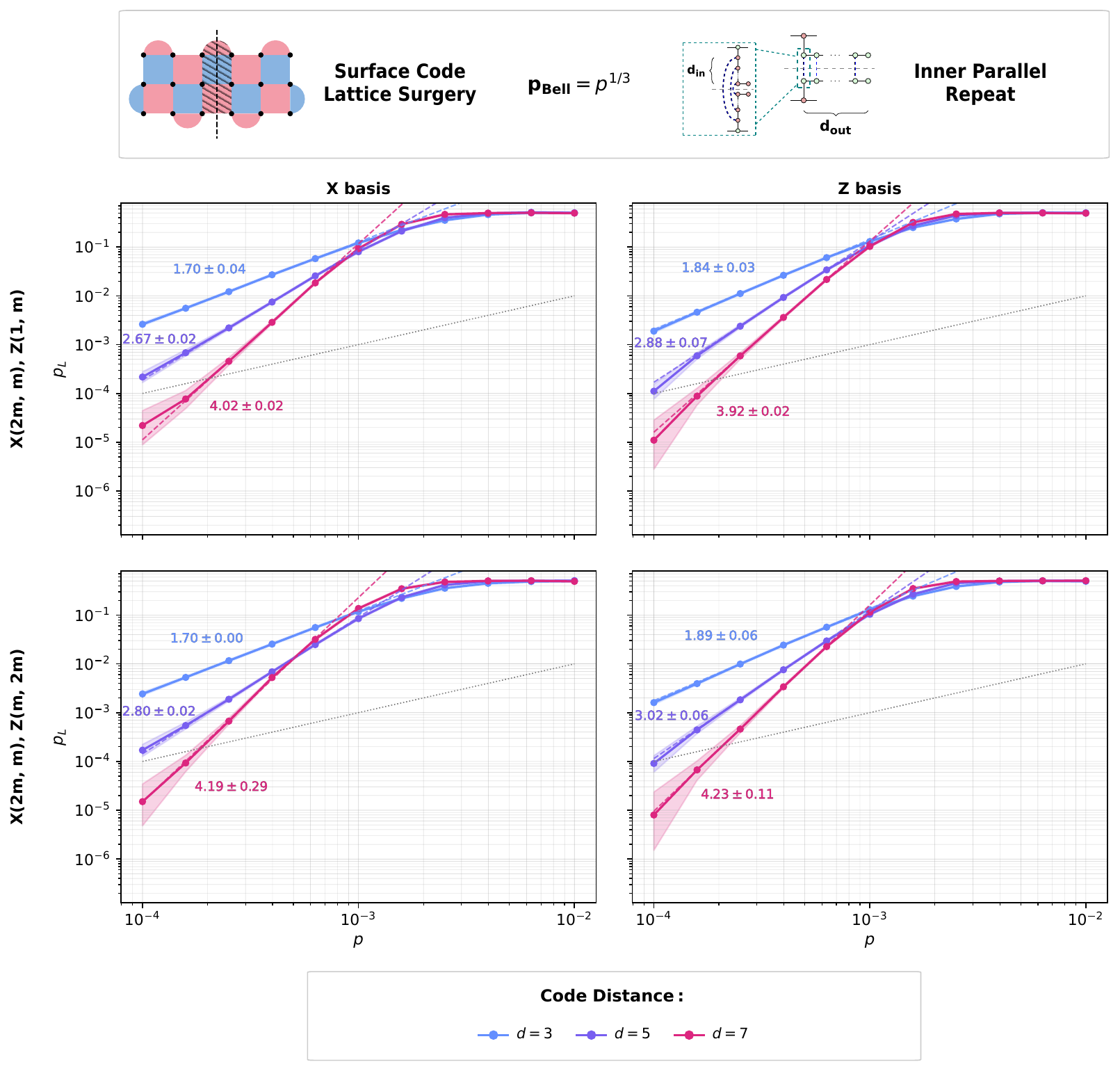}
  \caption{Distributed surface code: Lattice surgery simulation results for Bell pair noise level $p^{1/3}$ and distributed 2-2 split surgery plaquettes improved by the \vtwo approach. Both the context-aware (top) and context-free (bottom) improvement strategies achieve the code distance for both $X$ basis (left) and $Z$ basis (right) with an equal threshold of $\sim 10^{-3}$.}
  \label{fig:SC_LS_dist}
\end{figure}

We next show the simulation results for distributed lattice surgery, which we benchmark through the state teleportation protocol of \autoref{fig:tp_ls} in the $X$ and $Z$ bases, taking $\ebitnoise=3$.
The $XX$ surgery is realized with a single-column seam according to \autoref{fig:sc_surg_comb} with 2-2 split plaquettes.
The plaquettes are fault-improved with the \vtwo approach for both the context-free improvement strategy $X({2\ebitnoise,\ebitnoise})$, $Z({\ebitnoise,2\ebitnoise})$ and the context-aware improvement strategy for lattice surgery $X({2\ebitnoise, \ebitnoise})$, $Z({1 ,\ebitnoise})$.
We only simulated up to code distances $d\in \{3,5,7\}$, as the simulation and decoding of the larger and more complex circuits involved in lattice surgery are more demanding.
Each configuration was evaluated with $10^7$ shots, and the results are shown in \autoref{fig:SC_LS_dist}.

Compared to the memory setting, the difference between the context-free and context-aware fault improvements is less visible. The main difference is seen near threshold the $X$ basis: in the context-free case, the crossing point of $d=5$ and $d=7$ is shifted to the left. Meanwhile, in the low-$p$ regime both reach full slope and similar logical error rates as with the context-aware fault-improvement.
The difference in near-threshold behavior between the $X$ and $Z$ bases could be an artifact of decoding with the short beam Tesseract configuration; we note that since the chosen teleportation circuit utilizes $XX$ lattice surgery, the $X$ basis teleportation contains a more complex correlation surface involving the joint measurement plaquettes at the seam.

\begin{figure}[!htb]
  \centering
  \includegraphics[width=\textwidth]{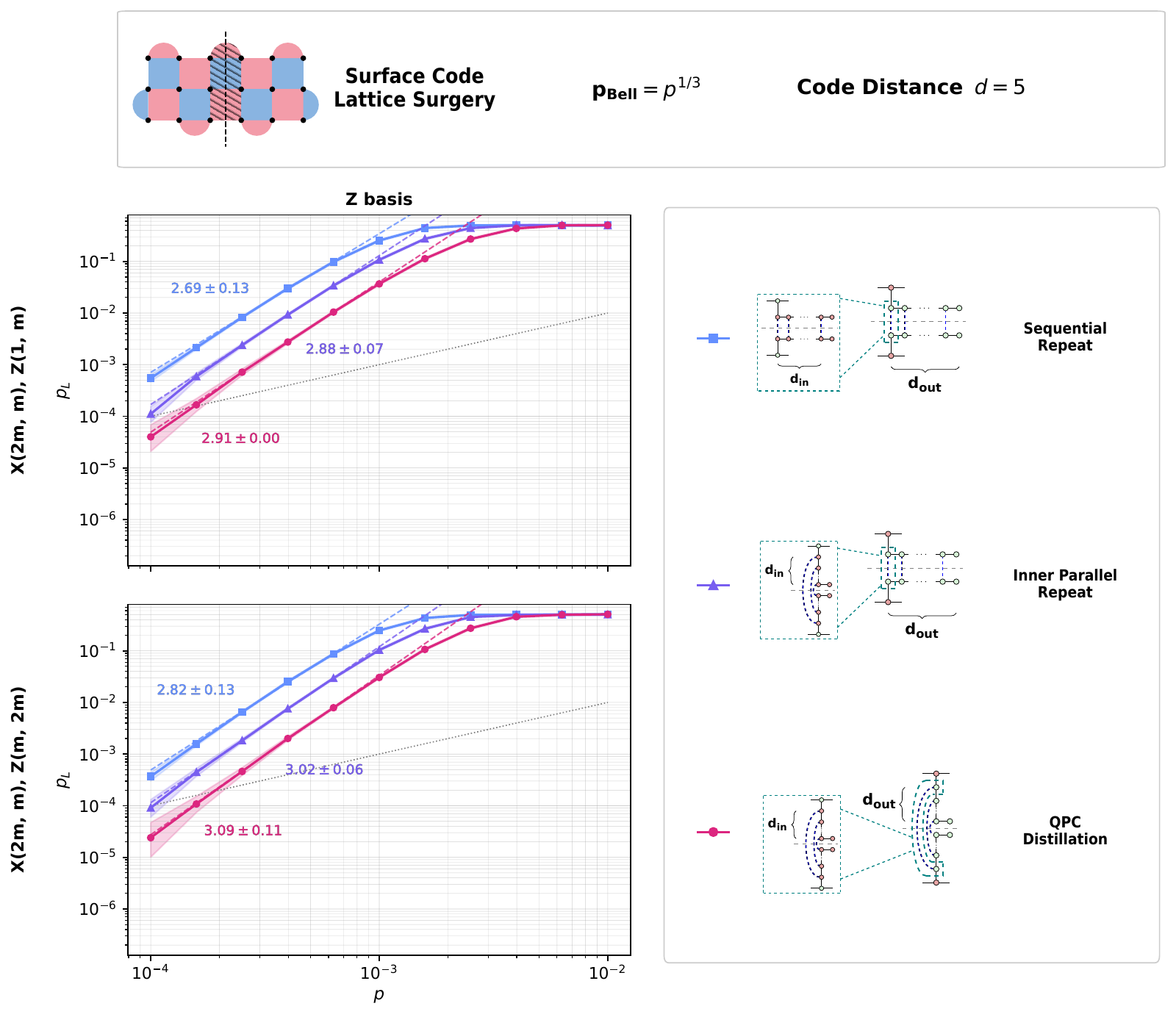}
  \caption{Distributed surface code: Lattice surgery simulation results for Bell pair noise level $p^{1/3}$ and distributed 2-2 split surgery plaquettes for fixed distance $d=5$ in the $Z$ basis. Both the context-aware (top) and context-free (bottom) are compared for three different improvement approaches. Contrary to the memory scenario, the approaches have equal curve shifts for both the context-free and context-aware strategy.}
  \label{fig:SC_LS_dist2}
\end{figure}

Apart from the near-threshold difference in behavior, there is no significant performance difference between context-free and context-aware improvement for lattice surgery. (But, again, the context-free improvement saves on resources.)
We can see that shifts in logical performance between sequential and parallel fault-improvement approaches are now apparent for \emph{both} the context-free and context-aware improvement, as shown in \autoref{fig:SC_LS_dist2} for fixed distance $d=5$.
The reason for observing the same behavior between context-free and context-aware improvement is that the type of plaquette with the highest Bell pair requirement (here: $X$-type) bottlenecks the overall execution time, hence determining the amount of idle noise.
Although the $Z$ plaquettes of the context-aware strategy require only a fraction of the Bell pairs of the context-free plaquettes, the $X$ plaquettes still need the same amount of fault-improvement $X({2\ebitnoise,\ebitnoise})$.
For the partially or fully sequential approaches, this means that while the $Z$ stabilizer measurements will be finished earlier for the context-aware strategy,  the data qubits will still idle until the $X$ stabilizer measurements are finished.

\FloatBarrier

\subsection{Distributed Color Code}

Within this subsection we present the results of simulating distributed memory and lattice surgery in the triangular color code, focusing on context-aware fault-improvement.
The color code is not matchable, and Tesseract is commonly used for its decoding already in the monolithic setting. In the distributed setting we use it to perform integrated decoding, again with the {short-beam} configuration.

\begin{figure}[!htb]
  \centering
  \includegraphics[width=\textwidth]{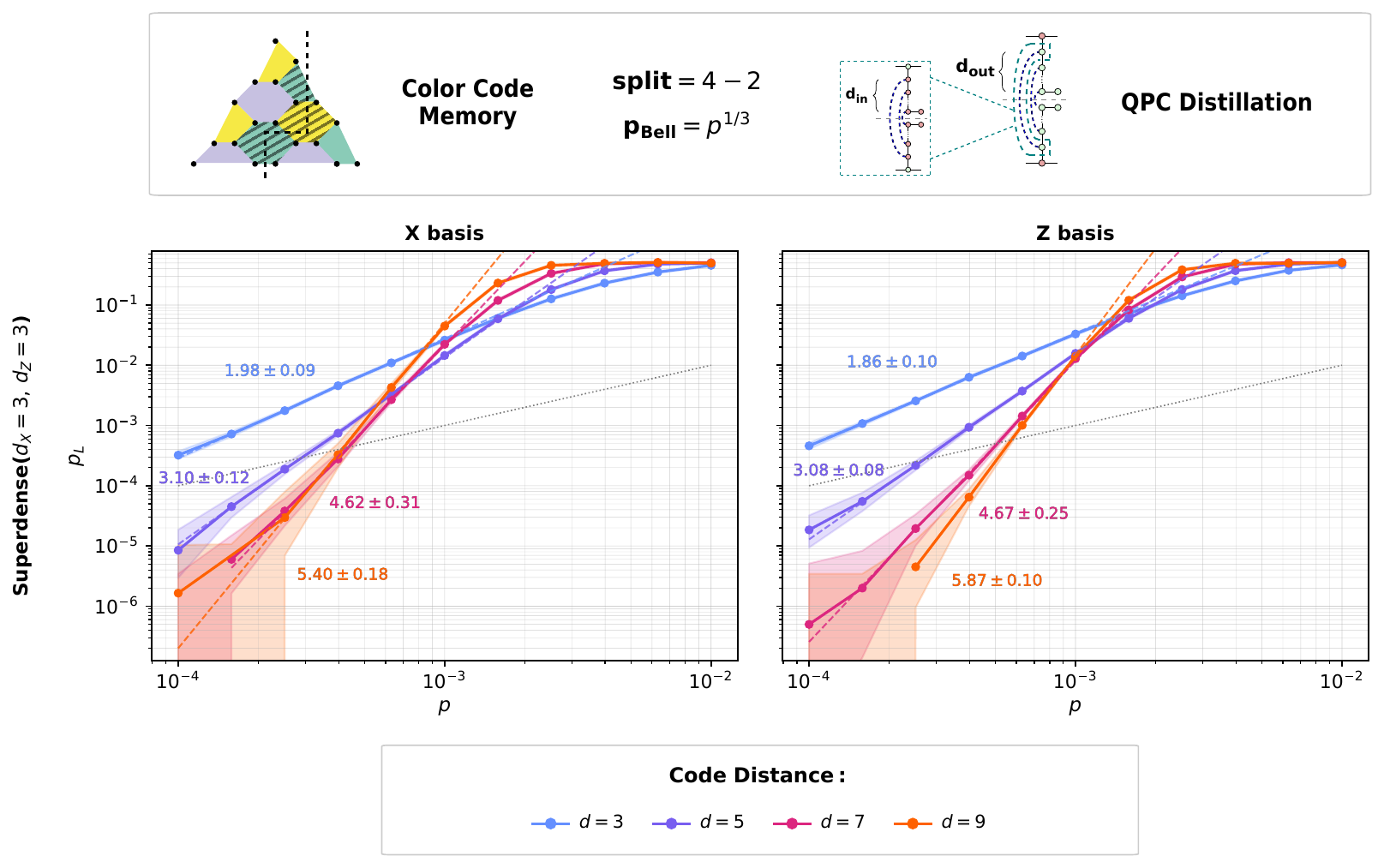}
  \caption{Distributed color code: Memory simulation results for Bell pair noise level $p^{1/3}$ and distributed 4-2 split plaquettes improved by the \vtwo approach. With improvements of $(d_X=3, d_Z=3)$ for the distribution edges in the superdense extraction circuits the full circuit distances are achieved.}
  \label{fig:CC_mem_dist}
\end{figure}

The results of the distributed memory simulations for $\ebitnoise=3$ are shown in \autoref{fig:CC_mem_dist}. We chose the distribution seam that utilizes 4-2 split plaquettes with distributed superdense circuits.
We fault-improve the plaquettes with the \ed approach, improving the distribution edges by $({\ebitnoise,\ebitnoise})$ and simulating code distances $d\in \{3,5,7,9\}$. Each configuration was evaluated with $10^6$ shots.
We see that we reach the desired slopes for each distance in both the $X$ and $Z$ Pauli basis.
For distances 7 and 9 we attribute the steepness of the slopes in the $Z$ basis
as indicating an even more prominent waterfall regime than in the surface code.

We find that in the $X$ basis, the logical error curves for distance 7 and 9 exhibit a slight shift compared to the $Z$ basis. This is unexpected, since the color code is symmetric between $X$ and $Z$. The superdense circuit is symmetric between $X$ and $Z$ measurements up to a time ordering, and the fault-improvement is also symmetric.
The origin of the shift is unclear, with one possibility being that it is an artifact of the decoder (either directly induced by the time-asymmetry, or originating from a random ordering of detectors that happens to disfavor the $X$ basis). We leave it as an open question.

\begin{figure}[!htb]
  \centering
  \includegraphics[width=\textwidth]{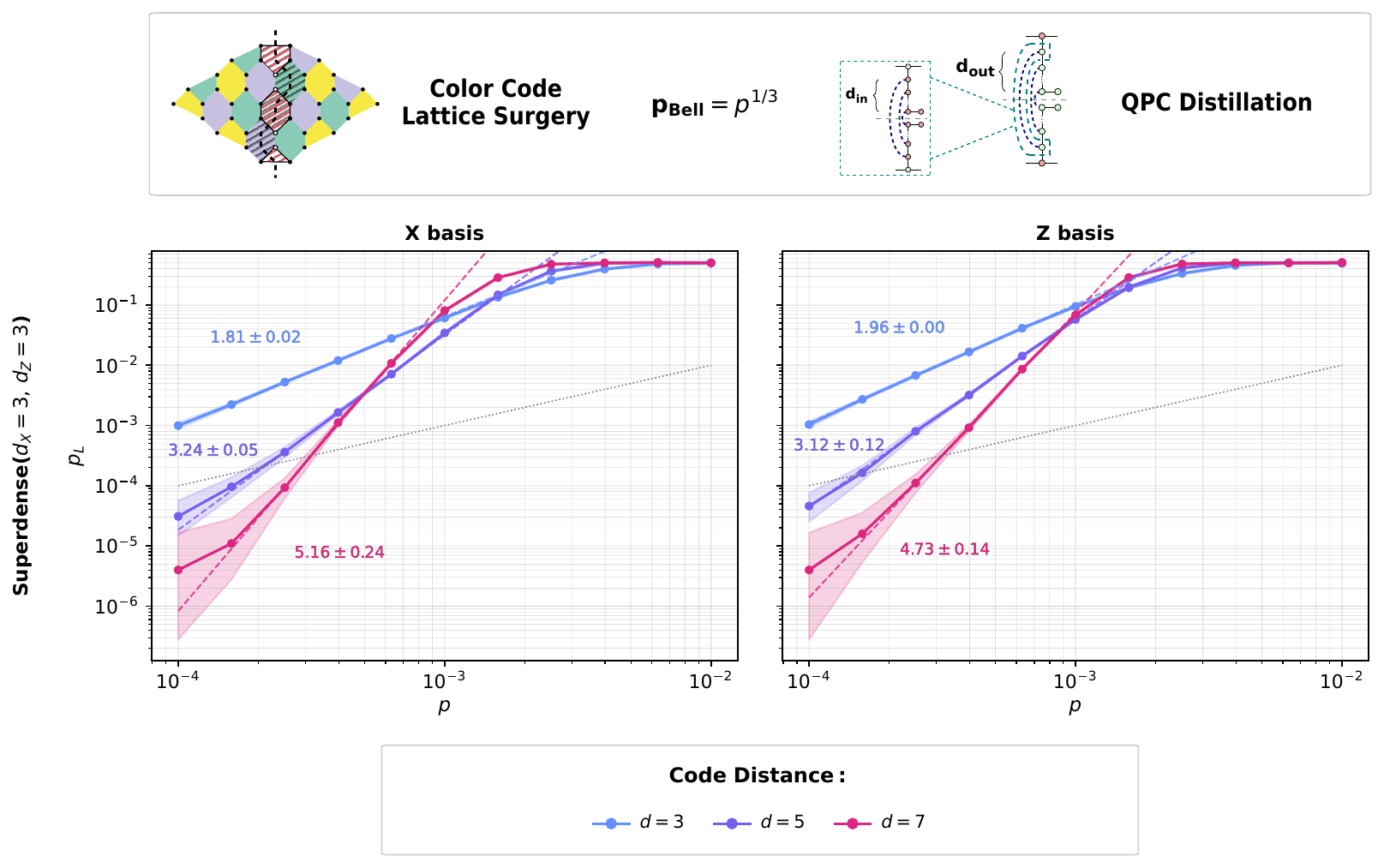}
  \caption{Distributed surface code: Lattice surgery simulation results for Bell pair noise level $p^{1/3}$ and a combination of distributed 3,3 and 4-2 split surgery plaquettes improved by the \vtwo approach. With improvements of $(d_X=3, d_Z=3)$ for the distribution edges in the superdense extraction circuits the full circuit distances are achieved.}
  \label{fig:CC_ls_dist}
\end{figure}

We next show simulation results for distributed lattice surgery in the color code, following the same teleportation-based benchmarking protocol as with the surface code.
The $XX$ surgery is realized according to \autoref{fig:cc_ls}. It involves extended boundary plaquettes hosted locally, extended boundary plaquettes that are distributed with a 4-2 split, and finally new distributed plaquettes with a 3-3 split that allow for extraction of the logical measurement outcome.
Both the 4-2 and 3-3 distributed plaquettes are implemented using the superdense extraction circuit, with distribution edges improved to $({\ebitnoise,\ebitnoise})$ using the \ed approach.
As in the surface code, we restrict the simulations to $d\in \{3,5,7\}$ due to the increased computational demand of lattice surgery.
Each configuration was evaluated with $10^6$ shots.

The lattice surgery results are shown in \autoref{fig:SC_LS_dist}, where the logical error curves show the expected asymptotic behavior.
Similar to what was seen in the surface code, we observe that the crossing point of the $d=5$ and $d=7$ logical error curves is shifted to the left in the $X$ basis.
Given the asymmetry of the teleportation circuit, this asymmetry between the $X$ and $Z$ bases is less surprising than in the color code memory setting. As noted for the surface code, the decoder choice could be a possible origin of the behavior.

\FloatBarrier

\subsection{Distributed Surface Code with Connectivity-Constrained QPU Boundaries}

The distributed circuits used in the simulations above have not been designed for restricted connectivity among the auxiliary qubits along the QPU boundary; only the connectivity of bulk plaquettes has been restricted to the square lattice. The assumed boundary flexibility allows, for instance, the high degree of parallelization in the \ed approach, but may not be realistic. To conclude the present section, we therefore design and simulate circuits that can be implemented when the qubits near the boundaries are constrained to square-lattice connectivity, focusing on the setting of distributed lattice surgery in the rotated surface code.
More specifically, we assume that for each distributed plaquette we can utilize \emph{one} Bell pair channel (connecting to one
communication qubit per QPU), and that we are furthermore limited to \emph{two} auxiliary qubits per QPU.
Both the auxiliary qubits and communication qubits follow the local square-lattice connectivity, as shown in \autoref{fig:sc_sq_grid_layout}.

\begin{figure}[!htb]
  \begin{subfigure}[ht]{0.47\textwidth}
    \centering
    \includegraphics[width=\textwidth]{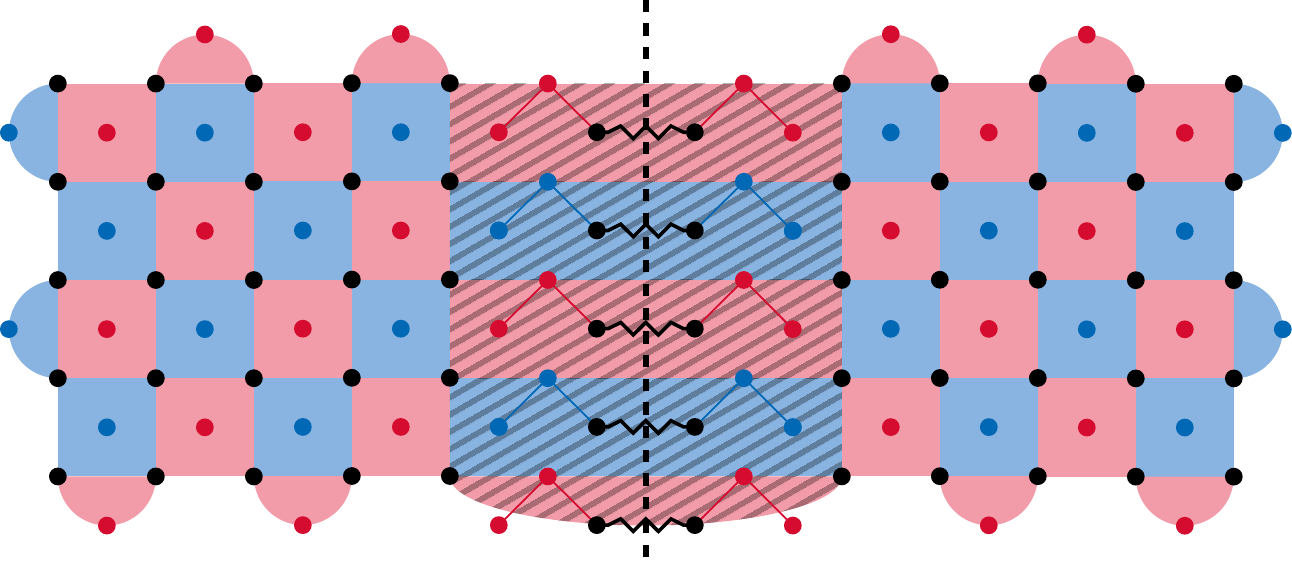}
    \caption{}
    \label{fig:sc_sq_grid_layout}
  \end{subfigure}
  \hfill
  \begin{subfigure}[ht]{0.47\textwidth}
    \centering
    \[ \tikzfig[scale=0.4]{07_sim_results/square_grid_circ} \]
    \caption{}
    \label{fig:sc_sq_grid_circ}
  \end{subfigure}
  \caption{Distributed lattice surgery under hardware constraints: (a) Each distributed plaquette only has access to one shared Bell pair channel and two auxiliary qubits per QPU; they all follow the local square grid connectivity layout (the three local auxiliary qubits are then restricted to a line connectivity). (b) Fault improvement inside the distributed plaquettes is realized with a modified \vone approach. Since only one Bell pair is available at a time, the first generated Bell pair is locally swapped to an auxiliary qubit. Then, the following Bell pairs are processed sequentially to realize repeated measurements on the swapped initial Bell pair. We here show the SWAPs using circuit notation rather than swapping the lines in the ZX diagram.}
  \label{}
\end{figure}

To account for these additional constraints, we design an adapted
version of the \vone circuit, shown in \autoref{fig:sc_sq_grid_circ}.
Notice that on each side of the seam, the auxiliary qubit that connects to the communication qubit never interacts directly with data qubits, only with the other auxiliary qubit.
This only requires 1D nearest-neighbor connectivity among the auxiliary and communication qubits, shown as a ``zig-zag line'' on the square grid. The auxiliary qubits are never shared between different plaquettes. Each SWAP in the circuit is implemented using three CNOT gates.

The resource demands of this approach scale linearly with the distance of the surface code, requiring in particular $d$ Bell pair channels. We note, however, that it is possible to design implementations with fewer than $d$  channels.
One straightforward option is to measure the distributed plaquettes in multiple timesteps, such as measuring all $X$ plaquettes first and then all $Z$ plaquettes. The two-step option halves the number of required channels without increasing the number of local operations, since each communication qubit is nearest-neighbor to both an $X$-plaquette auxiliary qubit and a $Z$-plaquette auxiliary qubit.

Another option for adapting to a limited Bell pair generation rate (either due to limited channels or a limited rate for each channel) is to buffer Bell pairs using additional auxiliary qubits. In a lattice surgery setting where a particular two-qubit logical measurement is only occasionally performed, this buffering could be done in-between such measurements.
The proposed circuit is in fact already an example of buffering, since we swap the first Bell pair transmitted through the channel onto the inner auxiliary qubit before use, instead of having two Bell pairs simultaneously available from two separate channels.

\begin{figure}[!htb]
  \centering
  \includegraphics[width=\textwidth]{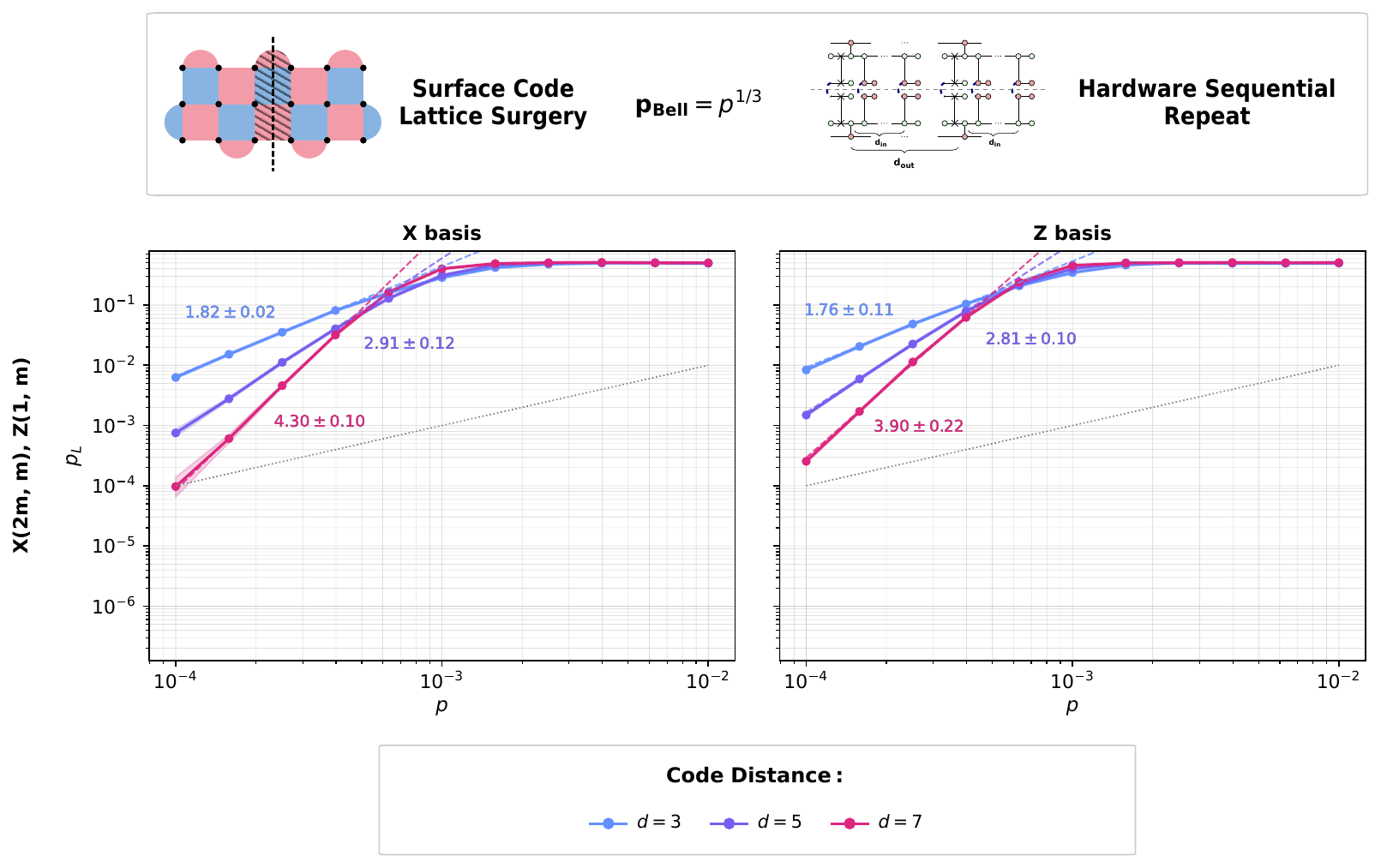}
  \caption{Distributed surface code: Lattice surgery under hardware constraints. Evaluation for distance $d=3,5,7$ for Bell pair noise $p^{1/3}$ with improvement strategy $X({2\ebitnoise, \ebitnoise})$, $Z({1 ,\ebitnoise})$ using the adapted \vone distribution circuit. We observe the expected asymptotic slope behavior with a shifted threshold compared to the unrestricted \vone circuit.}
  \label{fig:SC_LS_HW_dist_3}
\end{figure}

We benchmark the constrained lattice surgery circuits through a teleportation circuit, as in previous subsections. The results for $\ebitnoise=3$ and distance $d=3,5,7$ are shown in \autoref{fig:SC_LS_HW_dist_3}.
The fault-improvement is done through the context-aware $X({2\ebitnoise, \ebitnoise})$, $Z({1 ,\ebitnoise})$ strategy.
While the slopes reach the expected asymptotic behavior, we see a clear shift in the threshold from the less constrained
surface code lattice surgery. This shift is expected, given the increase in idle noise of the hardware restricted circuit.

\begin{figure}[!htb]
  \centering
  \includegraphics[width=\textwidth]{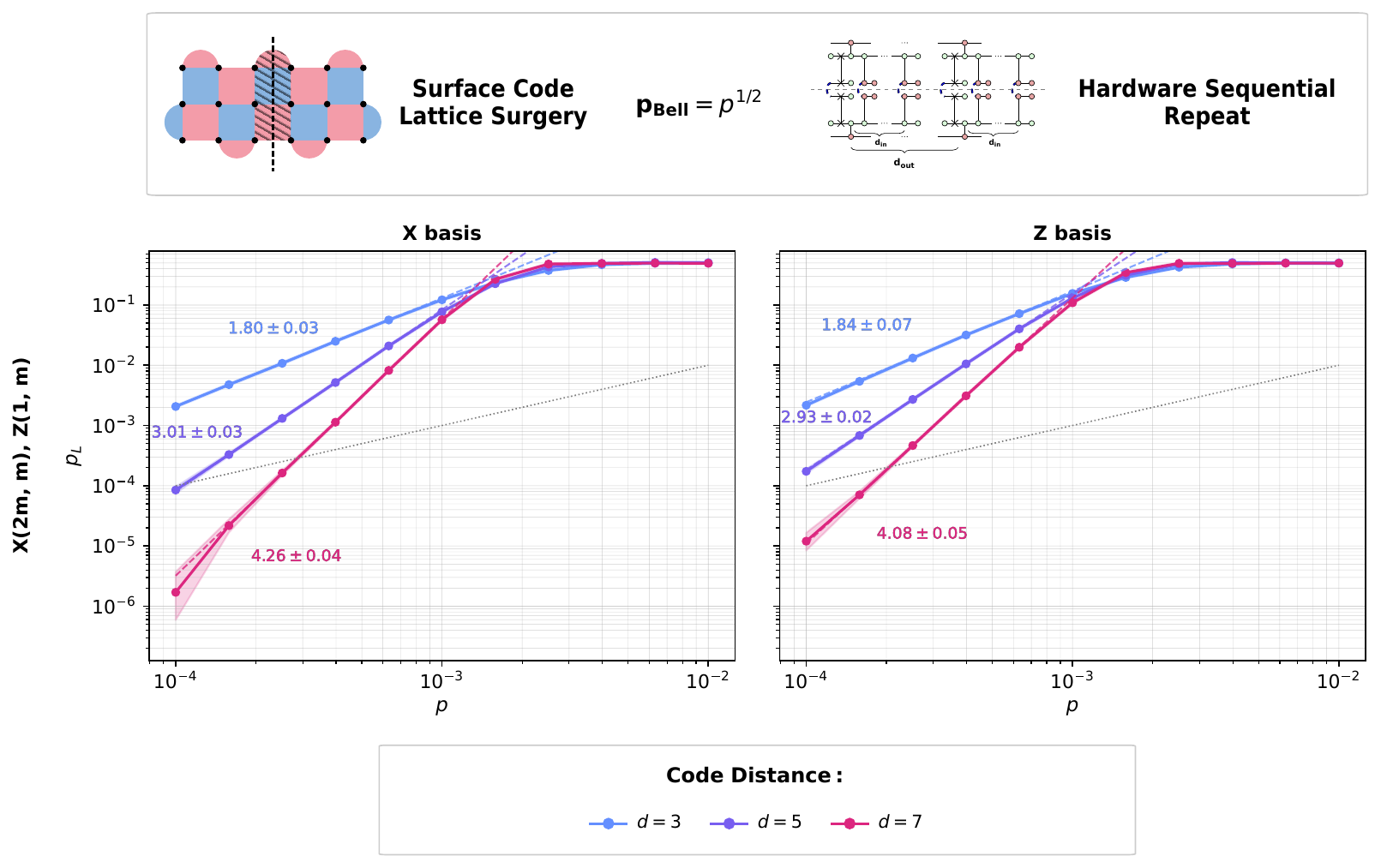}
  \caption{Distributed surface code: Lattice surgery under hardware constraints. Evaluation for distance $d=3,5,7$ for a less pessimistic Bell pair noise $p^{1/2}$ with improvement strategy $X({2\ebitnoise, \ebitnoise})$, $Z({1 ,\ebitnoise})$ using the adapted \vone distribution circuit. We observe the expected asymptotic slope behavior with a threshold at physical error rates $p \sim 10^{-3}$.}
  \label{fig:SC_LS_HW_dist_2}
\end{figure}

For comparison, we also benchmark the constrained circuits under a Bell pair noise level of $\ebitnoise=2$. Apart from modifying the amount of fault-improvement accordingly, the circuit is identical. The results are shown in \autoref{fig:SC_LS_HW_dist_2}, and we see that the reduction in required improvement significantly improves the logical performance.
The threshold is now around physical error rates $p \sim 10^{-3}$, with the corresponding Bell pair noise at $p_{Bell}\sim 10^{-3/2} \sim 0.03$.

\FloatBarrier
\section{Discussion}

In this work, we have introduced a new approach to distributed fault-tolerance, and used it to derive resource-efficient distributed circuits for stabilizer measurements. There are two main conceptual features of the approach:
First, integrating the protection against interconnect noise into the larger context, rather than treating it as a separate process.
Second, modeling the integrated problem using ZX-diagrams under weighted adversarial noise, which allows for a natural extension of previous work on fault tolerance, and especially of \emph{fault tolerance by construction}.
The new circuits reduce the overall Bell pair overhead compared to separate entanglement distillation, and can be adapted to different hardware settings through space-time tradeoffs.

We perform extensive numerical benchmarking and find that, even with very conservative assumptions on the noise levels on interconnects, it is possible to implement distributed fault-tolerant circuits with the same sub-threshold scaling as their monolithic counterparts. In the final example of \autoref{sec:numerics}, we see that for distributed surface code lattice surgery with only square-lattice connectivity among on-chip qubits, and only one Bell pair channel per distributed stabilizer, on-chip noise levels below $p\sim 10^{-3}$ and interconnect noise levels below $p_{\rm Bell}\sim 0.03$ are sufficient for successful (sub-threshold) operation at full circuit distance.
Given recent progress on distributed hardware, such noise levels do not seem unrealistic in the near term.

Outside of the setting of distributed quantum computing, we also note that a noise model containing significantly different noise levels in different parts of the circuit is more generally applicable, with one example being a monolithic setting with long-range intra-QPU connections that are noisier than nearest-neighbor gates. While such settings can also be treated by modifying prefactors rather than fault weights ($p$ versus $\alpha p$ instead of $p$ versus $p^w$), weighted adversarial noise simplifies the derivations: tracking and updating prefactors during ZX rewrites is more cumbersome than tracking and updating fault weights.

There are several directions for future work.  
Within the present work, we treated the Bell pair channel and the transduction to on-chip qubits as a black box, and focused on correlated Pauli noise on the shared Bell pair. While performing the fault-improvement using only on-chip operations has the advantage that such operations experience only on-chip levels of noise, 
 an approach that does not abstract away some of the details of the Bell pair channel would allow us to incorporate the effects of loss (in photonic channels) and leakage on the communication qubits. In the case of the former, one could draw lessons from fusion-based quantum computing and integrate them. This would require extending the framework of fault tolerance by construction to account for both leakage and loss.

We have seen that context-aware fault tolerance can reduce resource requirements over context-free fault tolerance. The context-aware analysis within the present work relied on known properties of the surface code and color code, which are among the most well-studied codes to date. To be able to extend the benefits of context-aware fault tolerance to more general settings, it is important to find systematic ways for deducing when fault propagation is benign within a given context.

In the distillation approach to fault-improvement, the present work has focused on $k=1$ distillation codes. As noted in \autoref{sec:fault_improvement}, $k>1$ generically requires a higher distance, to compensate for the lack of guarantees around how many logical Bell pairs are affected by faults on the interconnects. An open direction is how to identify or construct $k>1$ codes (together with suitable encoding circuits) that have guarantees such that this distance increase is not needed. Using such codes could then potentially lead to even further Bell pair savings.

A particularly important direction for future work is \emph{decoding}, given that integrated decoding has a substantial advantage in terms of lower resource requirements. Due to the difficulty of the integrated decoding problem, approximations are generally necessary, and we consider it likely that such approximations explain some of the shifts in the logical curves seen in \autoref{sec:numerics}. 
A more thorough investigation of decoding schemes tailored to the distributed setting is therefore a natural extension of this work.

Finally, as distributed hardware becomes available, the outcome of distributed fault-tolerance experiments will guide the theoretical development of new hardware-tailored protocols. Given the flexibility and resource savings that arise from integrating the treatment of interconnect noise into the larger circuit context, we believe that this will be a useful approach for designing fault-tolerant circuits for future distributed quantum computers.

\section*{Acknowledgements}
BR acknowledges support from Simon Harrison via the Wolfson Harrison UK Research Council Quantum Foundation Scholarship.
BP is supported by the Engineering and Physical Sciences Research Council grant number EP/Z002230/1, ``(De)constructing quantum software (DeQS)''.
EM and MS acknowledge funding by the German Ministry of Research, Technology and Space (BMFTR) through the project QuMAL-KI under Grant 50RA2208A (German Research Center for Artificial Intelligence) and Grant 50RA2208B (University of Bremen) by the German Aerospace Center (DLR) and through the project QuaSa under Grant No. 13N17300 administered by the VDI/VDE Innovation + Technik GmbH (VDI).
LGS is supported through a Leverhulme-Peierls Fellowship at the University of Oxford, funded by grant no. LIP-2020-014.

\medskip
\printbibliography

\onecolumn
\begin{appendix}
    \section{On-Chip Noise Models}\label{app:CLN}

For the numerical simulation, the SI1000 model from~\cite{gidneyBenchmarkingPlanarHoneycomb2022} was used.
It builds on the noise channels and rules of a circuit-level noise model:

\begin{table}[h]
    \centering
    \resizebox{\linewidth}{!}{
        \begin{tabular}{|c|l|}
            \hline
            \textbf{Noisy Gate}       & \textbf{Definition}                                                                                 \\
            \hline
            $\text{AnyClifford}_2(p)$ & \text{Any two-qubit Clifford gate, followed by a two-qubit depolarizing channel of strength $p$.}   \\
            \hline
            $\text{AnyClifford}_1(p)$ & Any one-qubit Clifford gate, followed by a one-qubit depolarizing channel of strength $p$.          \\
            \hline
            $\text{Reset}_Z(p)$       & Initialize the qubit as $\ket{0}$, followed by a bitflip channel of strength $p$.                   \\
            \hline
            $M_Z(p)$                  & Precede with a bitflip channel of strength $p$, and measure the qubit in the $Z$-basis.             \\
            \hline
            $\text{Idle}(p)$          & If the qubit is not used in this time step, apply a one-qubit depolarizing channel of strength $p$. \\
            \hline
            $\text{ResonatorIdle}(p)$ & If the qubit is not measured or reset in a time step during which other qubits are                  \\ &  being measured or reset, apply a one-qubit depolarizing channel of strength $p$. \\
            \hline
        \end{tabular}
    }
    \caption{
        Noise channels for circuit-level noise models, with noise parameters $p$ tunable to realize specific models such as standard depolarizing circuit-level noise and SI1000 circuit-level noise.
    }
    \label{tbl:noise}
\end{table}
Compared to ``default'' (standard depolarizing) circuit-level noise (CLN) the SI1000 noise model distinguishes between default idle noise and resonator idle noise, i.e. it models longer idle periods for time steps that include measurements or resets. Additionally, the noise strengths of the CLN noise channels are tuned with prefactors to mimic the noise of superconducting hardware more closely:
\begin{table}[h]
    \centering
    \begin{tabular}{|r|l|l|}
        \hline
        \textbf{Abbreviation}
         & CLN
         & SI1000
        \\\hline
        \textbf{Name}
         & \begin{tabular}{@{}l@{}}Standard\\Depolarizing\end{tabular}
         & \begin{tabular}{@{}l@{}}Superconducting\\Inspired\end{tabular}
        \\\hline
        \textbf{Noisy Gateset}
         & \noindent\begin{tabular}{@{}l@{}}
                        $\text{AnyClifford}_2(p)$ \\
                        $\text{AnyClifford}_1(p)$ \\
                        $\text{Reset}(p)$         \\
                        $M(p)$                    \\
                        $\text{Idle}(p)$          \\
                    \end{tabular}
         & \begin{tabular}{@{}l@{}}
               \vspace{-0.25cm}
               {}                           \\
               $\text{AnyClifford}_2(p)$    \\
               $\text{AnyClifford}_1(p/10)$ \\
               $\text{Reset}(2p)$           \\
               $M(5p)$                      \\
               $\text{Idle}(p/10)$          \\
               $\text{ResonatorIdle}(2p)$   \\
           \end{tabular}
        \\\hline
    \end{tabular}
    \caption{
       Parameters for standard depolarizing circuit-level noise and SI1000 circuit-level noise. The superconducting-inspired acronym SI1000 refers to an expected cycle time of about $1000$ nanoseconds.
    }
    \label{tbl:models}
\end{table}
    \section{Propagation of Y Faults During Fault-Improvement of Distribution Edges through Repetition}\label{sec:yproof}

In this section, we show the effects of $Y$ errors on fault improvement by repetition and prove the weights stated in \autoref{prop:fault-improve}.

A $Y$ error corresponds to a correlated $X$ and $Z$ error, or in the language of ZX as a fault gadget with an $X$ and $Z$ target on the same edge:
\[ \tikzfig[scale=0.5]{A03_app_FT_w/y_gadget} \]

For simplicity, we will not construct the ZX diagram with all X, Y, and $Z$ fault gadgets; instead, we will use insights from the isolated $X$ and $Z$ cases to argue for all relevant cases involving $Y$ errors.

First, we can see that single internal $Y$ errors will be detected by the inner detecting region due to their $X$ target.
Therefore, we only have to check pairs of errors that include $Y$ errors.

Any combination of errors on the internal horizontal edges with weight $w=1$ is fault-equivalent to a combination of outside faults with similar weight:
\[ \tikzfig[scale=0.5]{A03_app_FT_w/y_proof1} \]

Also a combination of a fault on an internal horizontal edge and a vertical edge does not generate a lower weight outside fault. For instance, a combination of $Y$ and $X$ error always propagates to an outside $Y$ error of $w\geq 1$:
\[ \tikzfig[scale=0.5]{A03_app_FT_w/y_proof2} \]

The only remaining relevant faults are combinations on the vertical edges. Two $Y$ errors propagate to a combination of two outside $X$ errors, which corresponds to an $X$ error on the internal edge of the distribution edge diagram:
\[ \tikzfig[scale=0.5]{A03_app_FT_w/y_proof3} \]
Since there are two independent mechanisms on the LHS of \autoref{prop:fault-improve} that correspond to the hook mechanism on the RHS, the mechanism with the higher probability dictates the noise behavior of the diagram. This corresponds to the final edge weight as the minimum, i.e., $\min(v_X + w_X, v_Y + w_Y)$.

Similarly, any of the two combinations of $X$ and $Y$ error propagates to a correlated outside $X$ and $Y$ error, the same effect as a $Y$ error on the internal edge of the distribution edge diagram:
\[ \tikzfig[scale=0.5]{A03_app_FT_w/y_proof4} \]
Again, for the two independent mechanisms on the LHS of \autoref{prop:fault-improve}, the one with the lower weight dictates the edge weight $\min(v_X + w_Y, v_Y + w_X)$.

Finally, since $Y$ errors on the vertical edges do not introduce any low weight $Z$ errors, we have recovered all edge weights of \autoref{prop:fault-improve}. \qed

\section{Implementability and Classical Corrections}\label{sec:corrections}

To make the ZX diagrams implementable as distributed circuits, we consider an on-chip gate set of CNOTs and single qubit gates.
For operations across QPUs, we use the established model in which Bell pairs are locally swapped to the communication qubits of the respective QPUs.
In terms of ZX diagrams this means that any inter-QPU Bell state generation must be subject to noise in accordance with the noise model for interconnects described in \autoref{sec:prelim}.
Our proposed distribution edge with two edges of modified noise can be fault-tolerantly rewritten into the desired, implementable form using individual local measurements on each QPU:
\begin{equation}
  \tikzfig[scale=0.5]{A_app_corr/impl_ebit_cx}
\end{equation}

Notice that the measurements are assumed to be zero in a noise free setting.
When implementing the diagrams as a circuit, this need not always hold; it depends on the context of use for the distributed edge.
One option for implementing the circuit in this form is to explicitly post-select on zero measurement outcomes.
As an alternative to this costly approach, we can track the effect of possible measurement results by introducing a pair of parametrized phase gates at the measurement, pushing one into the measurement spider and the other one to the output edges of the diagram.
For the example of the $XX$ distribution edge, we begin by interpreting each vertical edge as a two-qubit measurement with individual $k_1$ and $k_2$ measurement results, and then push them out.
We can make use of the fact that the two measurements form a detecting region, which implies that $k_1 = k_2 = k$.
We then have two remaining phase spiders at the outside of the diagram which correspond to classical corrections based on the measurement results.
\[\tikzfig[scale=0.5]{A_app_corr/impl_nondet_m}\]
This also extends to arbitrary repetitions:
\[\tikzfig[scale=0.5]{A_app_corr/impl_nondet_m2}\]
Coming back to our distributed implementation with two local measurements, the $x_i$ measurement result is now split into two local components $k_i =k_{i1} \oplus k_{i2}$.
Even though the local measurements do not stay deterministic, their parity and consequently the classical corrections do.
\[\tikzfig[scale=0.5]{A_app_corr/impl_ebit_cx2}\]

The phase $k\pi$ of any given pushed-out correction term needs to be known (determined from a measurement outcome in the past) at the time step where the correction is applied. In our example, the right correction term is pushed out to a time step at which $k$ is determined by the repeated measurements. On the other hand, the left correction is pushed to an earlier time step where no measurement has happened yet.
Therefore, when the diagram is embedded into a larger context, we have to further push the leftmost correction term out until the correction occurs at a later time step than the measurement.
In the following example, the distribution edge is an outer improvement inside a weight-4 measurement:
\[\tikzfig[scale=0.5]{A_app_corr/cc_past}\]
Both corrections are pushed out on outer qubits and can be moved to a time step after the measurements are executed. When the corrections are pushed through the CNOTs, they are also applied to the $\ket{0}$ state initialization and the $l_1$ measurement on the auxiliary qubit. Since both operations are of the opposite Pauli type, the correction has no effect on them.

In the context of our proposed distribution edges, if we not only protect against one type of error but have a nested structure with inner and outer repetitions, only our construction for outer repetitions requires classical corrections.
This can be seen by starting a Pauli web from the first measurement to the beginning of the circuit to track what previous elements determine its stabilizer information:
\[\tikzfig[scale=0.5]{A_app_corr/impl_baselines}\]
The Pauli web of the first inner repetition traces back to the initial Bell pair, forming a detecting region. Hence, the first inner distributed measurement is deterministic and equal to 0.
For the first outer repetition, the Pauli web is not contained inside, as we have individual $\ket{+}$ initializations on the auxiliary qubits, which are not $Z$ eigenstates.
The Pauli web flow reaches the boundaries of the diagram. Therefore, the first measurement is not deterministic and classical corrections are required.

    \section{QPC Encoding Circuit}\label{app:Shor}

In this section, we prove the improvement of entanglement distillation with the QPC encoder of \autoref{sec:QPC_encoder}.

We show fault equivalence by collapsing the entanglement distillation circuit into the improved edge using our set of fault-improving rewrites.
\[ \tikzfig[scale=0.4]{A_app_shor/ED_shor1} \]
We begin with the repeated inner protection blocks.
First, we merge in the measurements and adjust the colors of two-legged spiders.
\[ \tikzfig[scale=0.4]{A_app_shor/ED_shor2} \]
Now we can iteratively use our fault-improving rewrite (\autoref{prop:fault-improve}).
\[ \tikzfig[scale=0.4]{A_app_shor/ED_shor3} \]
With the $Z$-improvement complete, the same approach can now be applied to the outer protection blocks.
\[ \tikzfig[scale=0.4]{A_app_shor/ED_shor4} \]
    \section{CSS Code Generalization: Steane Code Example}\label{app:Steane}

In Section~\ref{sec:QPC_encoder} we showed how the proposed communication-edge improvements relate to the Quantum Parity Code (QPC). For this family of codes, the repeat-construction of \autoref{prop:nested_reps} provides a space--time tradeoff between parallel distillation and sequential measurement sequences. The implementation can therefore be adapted to the available resources, such as the Bell pair production rate and the number of communication qubits.

A natural question is whether similar tradeoffs can be achieved for other CSS codes. In this appendix, we investigate this question for the Steane code. We construct a Steane-based distribution edge and compare the resulting spacetime tradeoff with the QPC-based approach. As in the QPC case in the main text, we present parallel and sequential constructions below.

If seven Bell pairs are simultaneously available, efficient encoder circuits can be directly used for distillation~\cite{xuVariationalCircuitCompiler2021,buchbinderEncodingArbitraryState2011}:
\begin{equation}
    \label{eq:steane_distill}
    \tikzfig{A_app_steane/steane_distill_2}
\end{equation}

To trade off space for time, we build on existing work on representing CSS codes in ZX-calculus~\parencite{kissinger2022,huang2023,khesin2025}.
Specifically, we use the ZX normal form for CSS subsystem codes, as introduced in Section $3.1$ of~\cite{huang2023}, for the Steane code. This is shown as the RHS of \autoref{eq:steane_nf}:
\begin{equation}
    \label{eq:steane_nf}
    \tikzfig{A_app_steane/steane_graph}
\end{equation}
Through the bipartition of the normal form, each Bell pair is attached to a separate $Z$-spider, and all $Z$-spiders are connected to $X$-spiders, which represent the $Z$-stabilizer checks and $Z$-logicals of the Steane code.
To construct a sequential encoding circuit, we reorder and unfuse the normal form such that each $X$-spider is realized on an individual auxiliary qubit and each $Z$-spider gets executed one after another on a single communication qubit:
\begin{equation}
    \label{eq:steane_seq_1q}
    \tikzfig[scale=.35]{A_app_steane/steane_seq_1q}
\end{equation}
Notice that the encoding circuit is not built through fault-equivalent rewrites, but according to \autoref{prop:on-chip-dist} is still sufficient for on-chip distillation.

Compared to the simple structure of the QPC, the weight-4 stabilizers of the Steane code impose more constraints: None of the $Z$-stabilizer measurements can be easily executed sequentially to reuse auxiliary qubits.

Analogously to the approach presented in \autoref{section4_3}, we can interpolate between strictly sequential and parallel processing.
For instance, with two Bell pairs available in parallel, we can reduce the number of auxiliary qubits from four to two by carefully ordering Bell pairs and CNOTs, such that two $X$-spiders are finished before the other remaining two start:
\begin{equation}
    \label{eq:steane_seq_2q}
    \tikzfig{A_app_steane/steane_seq_2q}
\end{equation}

As a final example, \autoref{eq:steane_5q} shows a circuit variant that requires five Bell pairs and already integrates the Steane distillation inside a distribution edge: The circuit, together with its counterpart on the other QPU, executes a fault-improved $XX$ measurement acting on their respective auxiliary qubit, as shown in \autoref{eq:steane_example}.
\begin{equation}
    \label{eq:steane_5q}
    \tikzfig{A_app_steane/steane_5q}
\end{equation}

Compared with the QPC construction of \autoref{sec:QPC_encoder}, the Steane code provides a higher encoding rate at the cost of a larger encoding circuit. Since the Steane code is not a concatenated code and has a denser connectivity structure, the extent to which the interaction patterns can be separated and redistributed in time is limited. In contrast, the structure of concatenated repetition codes underlying the QPC construction is more separable and therefore offers greater flexibility in trading communication qubits for additional time steps. This comes from two properties of the QPC encoder. First, it is a concatenated CSS code, which is inherently decomposable. More specifically, consider the QPC encoder from \autoref{def:shor_encoder}; the $d_X$ $Z$-spiders representing the inner code can be executed independently producing $d_X$ intermediate logicals. Therefore, the $Z$-spiders can be executed sequentially on one register of $d_Z$ qubits while storing the resulting logicals in a $d_X$ qubit register. Second, all stabilizers of the underlying repetition codes have weight $2$, which imposes few constraints and allows for high flexibility, since only two qubits must be present simultaneously.

A general investigation of trade-off possibilities and limitations for CSS codes and systematic circuit construction procedures is outside the scope of this paper and is left as future work.

\FloatBarrier

\end{appendix}

\end{document}